\documentclass[11pt, letterpaper]{article}

\usepackage[T1]{fontenc}
\usepackage[utf8]{inputenc}
\usepackage{lmodern}
\usepackage{amsmath,amssymb,amsthm,mathtools,bm}
\usepackage[margin=1in]{geometry}
\usepackage{microtype}
\usepackage{hyperref}
\usepackage{cleveref}
\usepackage{enumitem}
\usepackage{algorithm}
\usepackage{algpseudocode}
\usepackage{caption}
\usepackage{float}
\usepackage{tikz}
\usetikzlibrary{arrows.meta,positioning,shapes.geometric}
\usepackage{quantikz}
\usepackage{aliascnt}
\usepackage{xcolor}
\usepackage{booktabs}
\usepackage{xfrac}
\usepackage{graphicx}
\usepackage{authblk}

\hypersetup{colorlinks=true,linkcolor=blue,citecolor=blue,urlcolor=blue}
\allowdisplaybreaks

\newtheorem{theorem}{Theorem}[section]

\newaliascnt{lemma}{theorem}
\newtheorem{lemma}[lemma]{Lemma}
\aliascntresetthe{lemma}

\newaliascnt{proposition}{theorem}

\aliascntresetthe{proposition}

\newaliascnt{corollary}{theorem}

\aliascntresetthe{corollary}

\theoremstyle{definition}

\newaliascnt{definition}{theorem}

\aliascntresetthe{definition}

\newaliascnt{assumption}{theorem}

\aliascntresetthe{assumption}

\newaliascnt{remark}{theorem}

\aliascntresetthe{remark}

\newtheorem{promise}{Promise}
\newtheorem{problem}{Problem}

\DeclareMathOperator{\Tr}{Tr}
\newcommand{\E}{\mathbb E}

\newcommand{\e}{\mathrm e}
\newcommand{\iu}{\mathrm i}

\newcommand{\cSWAP}{\operatorname{cSWAP}}
\newcommand{\SWAPECHO}{\operatorname{SWAP\text{-}ECHO}}

\title{Near-Optimal Ground-State Preparation without Controlled Hamiltonian Evolutions}

\author[1]{Dhrumil Patel\thanks{djp265@vt.edu}}
\author[1,2]{Steven T. Flammia}
\author[2,3]{Raúl García-Patrón}

\affil[1]{Virginia Tech}
\affil[2]{Phasecraft}
\affil[3]{University of Edinburgh}

\begin{document}

\maketitle

\begin{abstract}
Ground-state preparation, energy estimation, and property estimation are fundamental tasks in quantum simulation.
Existing ground-state preparation algorithms achieve near-optimal complexity using block encodings or controlled Hamiltonian evolution, raising a natural question: is controlled Hamiltonian evolution necessary, or can the same scaling be attained in the weaker access model that provides access to only uncontrolled Hamiltonian evolution?
We resolve this question by giving two algorithms that access the input Hamiltonian $H$ only through queries to $U_H=\e^{-\iu H\tau}$ and its inverse $U_H^\dagger$, where $\|H\|\le\alpha_H$ and $\tau=\Theta(1/\alpha_H)$.
Given an input state with ground-state weight at least $\eta$ and a spectral gap of $H$ of at least $\gamma$, our first algorithm uses fresh copies of the input state and requires $\widetilde O(\alpha_H/(\gamma\eta))$ queries to $U_H$ and $U_H^\dagger$, while our second algorithm uses coherent access to the input-state preparation unitary and its inverse to improve this complexity to $\widetilde O(\alpha_H/(\gamma\sqrt{\eta}))$.
The latter bound matches the known lower bound up to logarithmic factors, showing that algorithms with additional access to controlled Hamiltonian evolution cannot improve the query complexity for ground-state preparation beyond logarithmic factors.
Both algorithms are based on a simple mechanism that we call the spectral descent principle, in which the algorithm starts from the input state $|\psi\rangle$ and progressively descends through the energy spectrum, moving closer to the ground state until the desired fidelity is reached.
We also extend this framework to ground-state property and energy estimation and give explicit complexity guarantees for both tasks.
\end{abstract}

\newpage

\tableofcontents

\section{Introduction}
\label{sec:intro}

Given a quantum Hamiltonian, preparing its ground state, determining its ground-state energy or estimating physical observables in that state are central tasks in quantum simulation.
These problems arise naturally in quantum chemistry, condensed-matter physics, and materials simulation, where ground-state energies determine quantities such as reaction energies and phase stability, while ground-state expectation values determine observables such as correlation functions, response properties, and molecular dipole moments~\cite{AspuruGuzik2005,Cao2019,McArdle2020,Bauer2020}. 
Ground-state preparation is also a basic subroutine
for other quantum simulation applications, like estimating Green's functions used in Dynamical Mean-Field Theory \cite{Ayral2025},
but also for purely quantum algorithmic applications like quantum linear algebra.

The difficulty is that these problems are computationally hard for a completely generic Hamiltonian. 
For example, deciding whether the ground-state energy of a local Hamiltonian lies below or above two promised thresholds is QMA-complete~\cite{Kitaev2002,Kempe2006}. 
The closely related problem of estimating a local observable in the ground state, formalized through the approximate-simulation problem APX-SIM, is in fact 
$\mathrm{P}^{\mathrm{QMA}[\log]}$-complete~\cite{Ambainis2014,Gharibian2020}. 
Thus, an efficient ground-state algorithm cannot be expected without some prior information about the particular instance. 
To circumvent this hardness, standard algorithms assume we are provided with an initial state that has a non-negligible overlap with the true ground state~\cite{Poulin2009,Ge2019}. 
For example, in quantum chemistry one may initialize the algorithm with a Hartree-Fock or Kohn-Sham state, a small multi-determinant expansion, or a classically computed tensor-network state; such approximations are routinely used as trial states and can have substantial overlap with the true ground state~\cite{Berry2018,Fomichev2024,Berry2025}. 
Furthermore, tasks such as preparing the ground state or estimating observables in the ground state require the ground state to be distinguished from the rest of the spectrum~\cite{Lin2020,Dong2022,Zhang2022}. For these tasks, we additionally assume a known spectral gap. In practical scenarios, the spectral gap may be partially inferred from classical excited-state calculations or spectroscopic information about the system~\cite{Dreuw2005}.

To make these promises precise, let $H\in\mathbb C^{d\times d}$ be a Hamiltonian with known spectral-norm bound $\|H\|\le\alpha_H$ and with spectral decomposition $H=\sum_{j=0}^{d-1}E_j|\phi_j\rangle\langle\phi_j|$, where $\{|\phi_j\rangle\}_{j=0}^{d-1}$ are the eigenvectors of $H$ and the corresponding eigenvalues are ordered as $E_0<E_1\le\cdots\le E_{d-1}$. 
The ground state $|\phi_0\rangle$ is assumed to be unique. 
We are also given an input state $|\psi\rangle=\sum_jc_j|\phi_j\rangle$ whose ground-state weight (squared overlap) satisfies $|\langle\phi_0|\psi\rangle|^2=|c_0|^2\ge\eta$ for some known $\eta>0$. 
We refer to this as the \emph{overlap promise}. 
When coherent access to this state is needed, we assume black-box access to a preparation oracle $U_\psi$ and its inverse $U_\psi^\dagger$, where $U_\psi$ acts on an $n$-qubit system register, with $n=\lceil \log_2 d\rceil$, and satisfies $U_\psi|0^n\rangle=|\psi\rangle$. As discussed before, we are often given a \emph{gap promise}, i.e., a known spectral-gap lower bound $E_1-E_0\ge\gamma>0$. Note that we do not require these quantities to be known exactly: conservative lower bounds on the ground-state weight and the spectral gap are sufficient.

With the overlap and gap promises in place, we can now state the ground-state tasks considered in this work. 
In \emph{ground-state preparation} (GSP), both promises are given, and for a target infidelity $\xi\in(0,1)$ the goal is to output a state $\rho$ satisfying $\langle\phi_0|\rho|\phi_0\rangle\ge1-\xi$. 
In \emph{ground-state property estimation} (GSPE), an observable $O$ is additionally given and the goal is to estimate $\langle\phi_0|O|\phi_0\rangle$ to additive error $\varepsilon>0$. 
We assume that $O$ is specified classically in a form that admits an efficient measurement procedure, for example through an efficiently measurable Pauli decomposition. 
\textit{Ground-state energy estimation} (GSEE) asks instead for an estimate of $E_0$ to additive error $\varepsilon>0$. 
When the gap promise is available and $H$ itself can be efficiently measured from its classical description, GSEE is contained in the GSPE setting by taking the observable to be $H$, or equivalently $H/\alpha_H$ after normalization. 
GSEE can also be formulated under the overlap promise alone, without assuming a spectral gap, 
which is the reason behind keeping it separate from the other two. 
We will formally define these tasks later in \Cref{sec:ground-state-tasks}.

The computational complexity of these tasks depends on how the algorithm accesses the Hamiltonian $H$.
Our focus is the distinction between uncontrolled and controlled Hamiltonian evolution.
In the uncontrolled time-evolution model, the algorithm has access to $U_H\coloneqq \e^{-\iu H\tau}$ and its inverse $U_H^\dagger=\e^{+\iu H\tau}$ for a specified elementary time step $\tau>0$.
We count the total number $N_H$ of queries to these oracles and write $T_H=N_H\tau$ for the total Hamiltonian-evolution time.
For our algorithms, we choose $\tau=\Theta(1/\alpha_H)$.
On the other hand, the controlled time-evolution model additionally provides access to $\mathrm{ctrl}\text{-}U_H\coloneqq |0\rangle\langle0|\otimes I+|1\rangle\langle1|\otimes U_H$ and its inverse $(\mathrm{ctrl}\text{-}U_H)^\dagger$.
For a general black-box unitary, access to $U_H$ and $U_H^\dagger$ does not imply access to $\mathrm{ctrl}\text{-}U_H$~\cite{Araujo2014}.
Controlled Hamiltonian evolution is therefore an additional oracle assumption.

Existing GSP algorithms achieve near-optimal complexity under stronger Hamiltonian-access assumptions.
Given an $\alpha_H$-block encoding $U_H^{\rm BE}$ of $H$, Lin and Tong gave a GSP algorithm using $\widetilde O(\alpha_H/(\gamma\sqrt{\eta}))$ queries to $U_H^{\rm BE}$, $(U_H^{\rm BE})^\dagger$, and their controlled versions~\cite{Lin2020}.\footnote{Throughout this paper, $\widetilde O(\cdot)$, $\widetilde \Omega(\cdot)$, and $\widetilde \Theta(\cdot)$ suppress polylogarithmic factors.}
For $\tau=\Theta(1/\alpha_H)$, the generic QET-U algorithm likewise uses $\widetilde O(\alpha_H/(\gamma\sqrt{\eta}))$ queries to $\mathrm{ctrl}\text{-}U_H$ and $(\mathrm{ctrl}\text{-}U_H)^\dagger$, corresponding to $\widetilde O(1/(\gamma\sqrt{\eta}))$ total Hamiltonian-evolution time~\cite{Dong2022}.
Concurrent work by Jeffery et al.~\cite{Jeffery2026} achieves optimal GSP query complexity using $\mathrm{ctrl}\text{-}U_H$ and $(\mathrm{ctrl}\text{-}U_H)^\dagger$.
Chen et al.~\cite{Chen2026} also obtain optimal query bounds using $U_H^{\rm BE}$, $(U_H^{\rm BE})^\dagger$, and their controlled versions, assuming a known energy threshold suitably located inside the spectral gap.

These results are now known to be optimal up to logarithmic factors, after the recent work by Somma and de Wolf~\cite{Somma2026} establishing an error-dependent lower bound that holds even when $\mathrm{ctrl}\text{-}U_H$ and $(\mathrm{ctrl}\text{-}U_H)^\dagger$ are available.
Translated into our notation, their result requires $\Omega(\alpha_H\log(1/\xi)/(\gamma\sqrt{\eta}))$ queries to $U_H$, $U_H^\dagger$, $\mathrm{ctrl}\text{-}U_H$, and $(\mathrm{ctrl}\text{-}U_H)^\dagger$ in total for worst-case GSP instances with $\tau=\Theta(1/\alpha_H)$.\footnote{Their GSP lower bound is stated in terms of trace-distance error $\varepsilon$.
For a target ground state $|\phi_0\rangle$ and an output state $\rho$, the trace distance satisfies
\begin{equation}
\frac12\left\|\rho-|\phi_0\rangle\langle\phi_0|\right\|_1 \le \sqrt{1-\langle\phi_0|\rho|\phi_0\rangle},
\end{equation}
where $\|\cdot\|_1$ denotes the trace norm.
Thus, the infidelity guarantee $1-\langle\phi_0|\rho|\phi_0\rangle\le\xi$ implies trace-distance error at most $\sqrt{\xi}$.
Applying their lower bound with $\varepsilon=\sqrt{\xi}$ and using $\log(1/\sqrt{\xi})=\tfrac12\log(1/\xi)$ gives the stated logarithmic dependence.}
The bound holds even when arbitrarily many queries to $U_\psi$ and $U_\psi^\dagger$ are free, and corresponds to $\Omega(\log(1/\xi)/(\gamma\sqrt{\eta}))$ total Hamiltonian-evolution time.

Our goal is to attain this lower bound up to logarithmic factors using only queries to $U_H$ and $U_H^\dagger$.
This uncontrolled time-evolution model is the weakest of the three models discussed above: it provides neither a block encoding of $H$ nor access to $\mathrm{ctrl}\text{-}U_H$ or $(\mathrm{ctrl}\text{-}U_H)^\dagger$.
While QET-U removes controlled evolution for certain spin Hamiltonians by exploiting anticommutation properties, that construction does not extend to generic Hamiltonians satisfying only the standard overlap and gap promises~\cite{Dong2022}.

Removing the controlled Hamiltonian evolution is particularly relevant when the Hamiltonian dynamics is easier to realize than coherent control of the entire Hamiltonian evolution. 
On a fully fault-tolerant digital quantum computer, controlled Hamiltonian evolution can in principle be synthesized when a suitable circuit description of the simulation is available, but this generally introduces additional logical gates and coherent control. 
This overhead is especially relevant in the early fault-tolerant regime, where circuit depth remains limited and reducing the controlled structure of an algorithm is itself an important design objective~\cite{Dong2022,Zhang2022}.

The distinction is even more direct in analog and hybrid analog-digital quantum simulators. 
In these platforms, the desired many-body Hamiltonian is engineered physically and the system then evolves under that Hamiltonian, so ordinary uncontrolled Hamiltonian evolution is a native operation rather than a digitally synthesized subroutine~\cite{Georgescu2014}. 
Programmable Rydberg-atom arrays realize tunable interacting spin Hamiltonians in this way~\cite{Bernien2017,Browaeys2020,Ebadi2021}. For example, an analog-Hamiltonian interface can directly program the coherent drive, detuning, and Rydberg interaction governing the atoms' many-body evolution \cite{wurtz2023aquila}. 
Long-range Ising dynamics is likewise implemented directly in trapped-ion quantum simulators~\cite{Zhang2017}, Hubbard-type Hamiltonians arise naturally in ultracold atoms in optical lattices~\cite{Esslinger2010,Schaefer2020}, and superconducting analog and analog-digital processors can realize interacting Bose-Hubbard or spin dynamics directly in hardware~\cite{Hangleiter2024,Andersen2025}. 
In all of these settings, native evolution under a Hamiltonian does not by itself supply an ancillary qubit that coherently switches the complete evolution on and off.

There has recently been growing interest in avoiding controlled Hamiltonian evolution in ground-state related tasks~\cite{Huggins2020,Lu2021,Russo2021,Schuckert2023,Hemery2024,Yang2024,Clinton2024,Chan2024,Lee2025,Kuji2025,Schwartzman2026,Patel2026}.
These approaches show that coherent control can often be avoided, but typically by exploiting additional structure, such as restricted variational ans\"atze, symmetries, or known reference eigenstates and eigenvalues.
Moreover, many of these works address energy or observable estimation rather than ground-state preparation itself; the few approaches aimed directly at GSP either rely on additional Hamiltonian structure~\cite{Dong2022} or do not provide a generic efficient algorithm with rigorous guarantees~\cite{Schwartzman2026}.
Thus, to the best of our knowledge, up to now there was no generic efficient GSP algorithm with optimal complexity that assumes only the standard overlap and gap promises and relies only on uncontrolled time-evolution oracles $U_H$ and $U_H^\dagger$.
Since efficient GSP can in turn be used to obtain ground-state property and energy-estimation algorithms, resolving this question also provides a general route to these tasks without controlled Hamiltonian evolution. This motivates the central question of this work:

\bigskip

\emph{Can we solve GSP using a near-optimal $\widetilde O(\alpha_H/(\gamma\sqrt{\eta}))$ number of queries to $U_H$ and $U_H^\dagger$ in the uncontrolled time-evolution model, with $\tau=\Theta(1/\alpha_H)$, under only the standard spectral-gap and overlap promises?}

\subsection{Contributions}
\label{sec:contributions}

We answer this question affirmatively by giving two generic GSP algorithms in the uncontrolled time-evolution model.
For $\tau=\Theta(1/\alpha_H)$, our first algorithm, \emph{Spectral Descent} (SD), uses $\widetilde O(\alpha_H/(\gamma\eta))$ queries to $U_H$ and $U_H^\dagger$ with fresh-copy access to $|\psi\rangle$.
Our second algorithm, \emph{Coherent Spectral Descent} (CSD), uses $\widetilde O(\alpha_H/(\gamma\sqrt{\eta}))$ queries to $U_H$ and $U_H^\dagger$ with access to $U_\psi$ and $U_\psi^\dagger$.
Since the matching lower bound holds even when controlled Hamiltonian evolution is available~\cite{Somma2026}, CSD establishes that controlled evolution offers no improvement beyond logarithmic factors in the worst-case Hamiltonian-query complexity of GSP under coherent input-state access.
Our algorithms are based on a simple mechanism that we call the \emph{spectral descent principle}.

\paragraph{Spectral Descent Principle.} A fundamental aspect of both algorithms is the use of two system registers. 
One is the \emph{anchor register}, which stores the current \emph{anchor state} that we seek to drive toward the ground state $|\phi_0\rangle$. 
The other is the \emph{candidate register}, which is prepared in the input state $|\psi\rangle$ and stores the \emph{candidate state}.
The basic idea is to compare these two states and update the anchor register with the candidate state only when the candidate state provides evidence of a lower-energy direction. 
We perform this comparison using an \emph{energy comparator} acting on the anchor register, candidate register, and a single \emph{witness qubit}. 
To understand this comparison spectrally, expand the anchor and candidate states in the eigenbasis of $H$, and consider a pair of eigenstate components $|\phi_a\rangle$ and $|\phi_b\rangle$, with respective energies $E_a$ and $E_b$.
The comparator effectively answers the question: is $E_b < E_a$ by a sufficient margin?
Crucially, it does this relative check without evaluating either $E_a$ or $E_b$ directly.
We measure the witness qubit in the computational basis $\{|0\rangle,|1\rangle\}$, where outcome $1$ indicates a \emph{witness}, meaning that the candidate state has sufficiently lower energy than the current anchor state, while outcome $0$ indicates \emph{no witness}.
Upon observing a witness, we swap the candidate and anchor registers.
By iteratively executing these successful updates, the anchor state monotonically descends through the energy spectrum until it reaches the ground state $|\phi_0\rangle$. 
This is why we refer to this mechanism as the spectral descent principle.
We describe this principle and its action on general anchor and candidate states in more detail in \Cref{sec:comparison-intuition}.

\paragraph{Energy Comparator and SWAP Echo.} The core technical challenge is implementing this energy comparator using only the uncontrolled oracles $U_H$ and $U_H^\dagger$. 
We solve this by introducing a new primitive that we refer to as the \emph{SWAP echo}.
For an anchor state and candidate state with energies $E_a$ and $E_b$ respectively, the SWAP echo interlaces forward and backward uncontrolled time evolution with controlled-SWAP gates (see \Cref{sec:swap-echo} for its circuit). 
This interference writes the relative energy difference $E_a-E_b$ directly into the phase of the witness qubit. 
We then process this relative phase using a Laurent quantum signal processing (QSP) sequence to construct the energy comparator. 
Because the QSP sequence only queries the SWAP echo directly, the entire comparator circuit does not require $\mathrm{ctrl}\text{-}U_H$.

\bigskip

Given that the candidate state $|\psi\rangle$ always possesses a ground-state component of weight at least $\eta$ from the overlap promise, an excited anchor state will trigger a successful downward update with a probability of order $\Omega(\eta)$. 
We present two variations of the algorithm that handle this success probability differently, depending on how the input state $|\psi\rangle$ is accessed.
Our first algorithm, SD, is designed for a weak input-state access model where the algorithm can only produce fresh copies of $|\psi\rangle$ and does not have access to its preparation unitary $U_\psi$. 
In each descent step, we prepare $|\psi\rangle$ in the candidate register, apply the energy comparator, and measure the witness qubit. 
If the measurement yields the outcome $0$ (no witness), we do not swap the anchor, discard the candidate register, and retry. 
Since every excited anchor component can pair with the ground-state component of the candidate state, each trial succeeds with a probability of order $\Omega(\eta)$. 
Therefore, $O(1/\eta)$ trials suffice to obtain an outcome 1 (witness) with constant probability.
This completes one descent step.
SD consists of multiple such descent steps.

\begin{theorem}[GSP via Spectral Descent, Informal]
\label{thm:intro-SD}
Assume the standard gap promise $\gamma$ and overlap promise $\eta$, and suppose we can prepare fresh copies of the input state $|\psi\rangle$. Then, the Spectral Descent algorithm (\Cref{alg:spectral-descent}) prepares a state $\rho$ satisfying $\langle\phi_0|\rho|\phi_0\rangle \ge 1-\xi$, using only the uncontrolled oracles $U_H$ and $U_H^\dagger$, fresh preparations of $|\psi\rangle$, controlled-SWAP gates, and single-qubit operations.

By setting the elementary time step to $\tau = \Theta(1/\alpha_H)$, the algorithm requires $\widetilde{O}(\alpha_H/(\gamma\eta))$ queries to $U_H$ and $U_H^\dagger$, and uses $\widetilde{O}(1/\eta)$ fresh copies of $|\psi\rangle$. 
This corresponds to a total Hamiltonian-evolution time of $\widetilde{O}(1/(\gamma\eta))$. Crucially, the algorithm does not require access to~$\mathrm{ctrl}\text{-}U_H$.
\end{theorem}

The $1/\eta$ factor in the query complexity for SD arises solely from repeatedly preparing and measuring the candidate state to find a witness. 
If instead we have access to the preparation oracles $U_\psi$ and $U_\psi^\dagger$, then, by exploiting amplitude amplification ideas, we can perform this search quadratically faster. 
Our second algorithm, CSD, is based on this stronger access model.
CSD employs the exact same anchor register, candidate register, SWAP echo, and energy comparator as SD. The only modification lies in how the witness is located. 
Instead of measuring the witness qubit and re-preparing the candidate state upon failure, CSD leaves the candidate register and witness qubit coherent. 
It applies the fixed-point amplitude amplification primitive of Ref.~\cite{Yoder2014} to the witness branch, amplifying the initial $\Omega(\eta)$ probability of seeing a witness to a constant using only $O(1/\sqrt{\eta})$ coherent applications of the energy comparator.
This completes one descent step, and, as with SD, CSD also consists of multiple such descent steps.

\begin{theorem}[GSP via Coherent Spectral Descent, Informal]
\label{thm:intro-CSD}
Assume the standard gap promise $\gamma$ and overlap promise $\eta$, and suppose we have coherent access to the state preparation oracle $U_\psi$ and its inverse $U_\psi^\dagger$. 
Then, the Coherent Spectral Descent algorithm prepares a state $\rho$ satisfying $\langle\phi_0|\rho|\phi_0\rangle \ge 1-\xi$, using only the uncontrolled oracles $U_H$ and $U_H^\dagger$, the state oracles $U_\psi$ and $U_\psi^\dagger$, controlled-SWAP gates, and single-qubit operations. 

By setting the elementary time step to $\tau = \Theta(1/\alpha_H)$, the algorithm requires $\widetilde{O}(\alpha_H/(\gamma\sqrt{\eta}))$ queries to $U_H$ and $U_H^\dagger$, and uses $ \widetilde{O}(1/\sqrt{\eta})$ queries to $U_\psi$ and $U_\psi^\dagger$. 
This corresponds to a total Hamiltonian-evolution time of $\widetilde{O}(1/(\gamma\sqrt{\eta}))$. 
Crucially, the algorithm does not require access to~$\mathrm{ctrl}\text{-}U_H$.
\end{theorem}

We now formalize these tasks (\Cref{sec:ground-state-tasks}) before introducing the 
spectral descent mechanism (\Cref{sec:comparison-intuition}), the SWAP echo (\Cref{sec:swap-echo}) and energy comparator (\Cref{sec:qsp-comparator}).
The main result concerning the spectral descent mechanism is presented in \Cref{sec:fresh-descent} and its coherent analogue in \Cref{sec:coherent-descent}. In addition, we also extend our result to ground-state property and energy estimation, with the corresponding theorems, complexity guarantees, and additional measurement-access assumptions given in \Cref{sec:ground-state-estimation}. Finally, we provide additional discussions in \Cref{sec:conclusions}.

\section{Ground-State Task Definitions}
\label{sec:ground-state-tasks}

We now formalize the ground-state tasks considered in this paper. 
Let $H\in\mathbb C^{d\times d}$ be a Hamiltonian satisfying $\|H\|\le\alpha_H$, with spectral decomposition $H=\sum_{j=0}^{d-1}E_j|\phi_j\rangle\langle\phi_j|$, where the eigenvalues are ordered as $E_0<E_1\le\cdots\le E_{d-1}$ and the ground state $|\phi_0\rangle$ is unique. We are additionally given an input state $|\psi\rangle=\sum_jc_j|\phi_j\rangle$.

\begin{promise}[Overlap promise]
\label{prom:overlap}
For a known parameter $\eta\in(0,1]$, the pair $(H,|\psi\rangle)$ satisfies the \emph{overlap promise} if the ground-state weight of $|\psi\rangle$ satisfies $|\langle\phi_0|\psi\rangle|^2=|c_0|^2\ge\eta$. Thus, $\eta$ is a known lower bound on the actual ground-state weight, which may be larger.
\end{promise}

\begin{promise}[Gap promise]
\label{prom:gap}
For a known parameter $\gamma>0$, the Hamiltonian $H$ satisfies the \emph{gap promise} if its ground-state spectral gap satisfies $E_1-E_0\ge\gamma$.
\end{promise}

With that in place, we next define the ground-state tasks.

\begin{problem}[Ground-State Preparation (GSP)]
\label{prob:gsp}
Given a Hamiltonian $H$ and an input state $|\psi\rangle$ satisfying Promises~\ref{prom:overlap} and~\ref{prom:gap}, together with a target infidelity $\xi\in(0,1)$, the \emph{ground-state preparation} (GSP) problem is to output a quantum state $\rho$ satisfying $\langle\phi_0|\rho|\phi_0\rangle\ge1-\xi$.
\end{problem}

\begin{problem}[Ground-State Property Estimation (GSPE)]
\label{prob:gspe}
Given a Hamiltonian $H$ and an input state $|\psi\rangle$ satisfying Promises~\ref{prom:overlap} and~\ref{prom:gap}, an observable $O$ satisfying $\|O\|\le1$, a target additive error $\varepsilon>0$, and a failure probability $\delta\in(0,1)$ the \emph{ground-state property estimation} (GSPE) problem is to output an estimate $\widehat \mu$ of $\mu \coloneqq \langle\phi_0|O|\phi_0\rangle$ satisfying $|\widehat \mu-\mu|\le\varepsilon$ with probability at least $1-\delta$.
\end{problem}

\begin{problem}[Ground-State Energy Estimation (GSEE)]
\label{prob:gsee}
Given a Hamiltonian $H$ and an input state $|\psi\rangle$ satisfying Promise~\ref{prom:overlap}, together with a target additive error $\varepsilon>0$ and a failure probability $\delta\in(0,1)$, the \emph{ground-state energy estimation} (GSEE) problem is to output a classical estimate $\widehat E_0$ satisfying $|\widehat E_0-E_0|\le\varepsilon$ with probability at least $1-\delta$. Unlike GSP and GSPE as defined above, GSEE does not intrinsically require the gap promise.
\end{problem}

GSEE may be viewed as the special case of GSPE obtained by taking the observable $O$ to be $H/\alpha_H$. 
Nevertheless, we distinguish the two problems because GSEE can be posed under Promise~\ref{prom:overlap} alone, whereas the formulation of GSPE requires both Promises~\ref{prom:overlap} and~\ref{prom:gap}.

\section{The Spectral Descent Principle}
\label{sec:comparison-intuition}

Before presenting the SD and CSD algorithms, we first describe in more detail the common mechanism underlying both algorithms, which we refer to as the spectral descent principle.

Both algorithms operate on two system registers. 
The first is the \emph{anchor register}, which stores the current \emph{anchor state}, that is, the state the algorithm aims to iteratively improve and drive it towards the true ground state $|\phi_0\rangle$. 
The second is the \emph{candidate register}, which stores a \emph{candidate state} initialized to fresh copies of the input state $|\psi\rangle$.

To rigorously understand this principle, it is useful to begin with an idealized picture in which the anchor state happens to be an energy eigenstate $|\phi_a\rangle$ with eigenvalue $E_a$.
This is strictly for intuition; the algorithm does not assume access to energy eigenstates, nor does it project the anchor register into the energy basis.
The full algorithms and their analysis will later show how the same mechanism acts on a general anchor state.

Now prepare the candidate register in the input state $|\psi\rangle$. 
Writing $ |\psi\rangle=\sum_b c_b|\phi_b\rangle$ with $p_b=|c_b|^2$, the candidate state contains the energy component $|\phi_b\rangle$ with weight $p_b$.

For the moment, assume that we have access to an ideal energy comparison test, which we call the \emph{energy comparator}. 
Given the anchor state $|\phi_a\rangle$ and a candidate component $|\phi_b\rangle$, this ideal comparator flags a success, producing a \emph{witness}, if and only if the candidate energy is strictly lower than the anchor energy:
\begin{equation}
E_b < E_a. \label{eq:ideal-comparison-rule}
\end{equation}
Operationally, observing a witness guarantees that the candidate state contains a spectral component whose energy is closer to the ground state energy $E_0$ than the energy of the current anchor state.

We analyze the behavior of this rule in two distinct cases.
First, suppose the anchor state has already reached the ground state, so $a=0$. Because $E_0$ is the absolute minimum of the spectrum, there exists no candidate energy $E_b$ satisfying $E_b < E_0$. 
Consequently, the ideal comparator will never produce a witness:
\begin{equation}
\Pr[\text{witness} \mid a=0] = 0. \label{eq:ideal-ground-no-witness}
\end{equation}
This ensures stability: once the algorithm finds the ground state $|\phi_0\rangle$, it halts any further updates.

Second, suppose the anchor state is an excited state, meaning $a>0$. 
By the gap promise introduced earlier, we know that $E_a - E_0 \ge \gamma > 0$. 
Therefore, the strict inequality $E_0 < E_a$ is satisfied. 
This guarantees that the ground-state component $|\phi_0\rangle$ present in every fresh candidate state $|\psi\rangle$ will always trigger the witness condition in~\eqref{eq:ideal-comparison-rule}.

We can lower-bound the probability of observing this witness using the overlap promise.
The candidate state $|\psi\rangle$ possesses a ground-state weight $p_0 \ge \eta$. 
Because the ideal energy comparator deterministically produces a witness for this ground-state component whenever $a>0$, the overall probability of observing a witness is bounded by:
\begin{equation}
\Pr[\text{witness} \mid a>0] \ge p_0 \ge \eta. \label{eq:ideal-excited-witness}
\end{equation}
Other candidate components $|\phi_b\rangle$ with energies $E_b < E_a$ will also contribute to this probability, making~\eqref{eq:ideal-excited-witness} a strict lower bound. 
The important point here is that regardless of the complex structure of the excited spectrum, the ground-state component alone ensures a success probability of at least $\eta$ for any excited anchor state.

This inherent probabilistic separation shown by~\eqref{eq:ideal-ground-no-witness} and~\eqref{eq:ideal-excited-witness}, that is, zero witness probability for the ground state and at least $\eta$ for any excited state, dictates the update protocol for both SD and CSD. 
The procedure for a single update attempt is illustrated in \Cref{fig:spectral-descent-step}. 
We start with the current anchor state in register $A$, prepare a fresh candidate state $|\psi\rangle$ in register $B$, and apply the energy comparator alongside a single witness qubit $q$. 
If the measurement yields outcome $q=0$ (no witness), we discard the candidate register $B$, retain the state remaining in the anchor register $A$, and repeat the process with a newly prepared candidate state. 
If the measurement yields outcome $q=1$ (witness), we execute a SWAP operation to replace the current anchor state with the successful candidate state, and continue the descent from this newly updated anchor.

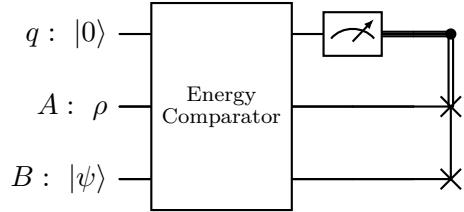
\begin{figure}[H]
\centering
\begin{quantikz}[row sep=0.55cm,column sep=0.42cm]
\lstick{$q:\ |0\rangle$}
& \gate[3]{\substack{\mathrm{Energy}\\\mathrm{Comparator}}}
& \meter{}
& \cw
& \cwbend{1}
\\
\lstick{$A:\ \rho$}
& \qw
& \qw
& \qw
& \swap{1}
\\
\lstick{$B:\ |\psi\rangle$}
& \qw
& \qw
& \qw
& \targX{}
\end{quantikz}
\caption{
One update attempt in the spectral descent principle.
}
\label{fig:spectral-descent-step}
\end{figure}

An unsuccessful comparison ($q=0$) therefore does not mean that the ground-state preparation has failed. 
It only means that the particular candidate state used in that comparison did not produce a useful downward update. 
The anchor state is neither discarded nor re-prepared. 
It remains stored in the anchor register and is compared with another fresh candidate state.
Furthermore, in the case when $q=1$, a successful candidate state also need not be the ground state. 
If the candidate state is now an intermediate eigenstate $|\phi_b\rangle$ such that
\begin{equation}
E_0 \le E_b < E_a,
\end{equation}
the condition~\eqref{eq:ideal-comparison-rule} is still met. 
Replacing the anchor state $|\phi_a\rangle$ with $|\phi_b\rangle$ successfully lowers the energy of the anchor register, even if it has not yet reached $E_0$. 
Subsequent successful comparisons will continue to drive the anchor energy downward.

The overlap promise ensures that an excited anchor state cannot become trapped at a local minimum. 
If the anchor state is fixed at an excited eigenstate $|\phi_a\rangle$, each independent attempt with a fresh candidate state $|\psi\rangle$ yields a downward witness with probability at least $\eta$, as established in~\eqref{eq:ideal-excited-witness}. 
Therefore, if we execute $M$ independent comparison trials, the probability of failing to obtain any witness decays exponentially:
\begin{equation}
\Pr[\text{no witness in } M \text{ trials}] \le (1-\eta)^M \le \mathrm{e}^{-\eta M}. \label{eq:ideal-M-failure}
\end{equation}
Consequently, executing $M = \mathcal{O}(1/\eta)$ repetitions guarantees a successful downward update with constant probability.

Notice that this rigorous bound relies exclusively on the ground-state component of $|\psi\rangle$. 
We do not require any assumptions about the gaps between adjacent excited states. 
The excited spectrum can be arbitrarily closely spaced or highly degenerate. 
The only spectral separation required is the initial gap promise $E_1 - E_0 \ge \gamma$, which guarantees that the target ground state $|\phi_0\rangle$ is safely separated from any excited anchor state $|\phi_a\rangle$.

Up to this point, our analysis has relied on an ideal energy comparator that perfectly evaluates the strict inequality $E_b < E_a$. 
Also an algorithm must execute this relative comparison without measuring the registers in the energy basis and without utilizing $\mathrm{ctrl}\text{-}U_H$. 
To construct this practical comparator, we begin by introducing a new primitive in the next section called the \emph{SWAP echo}.

\section{The SWAP Echo}
\label{sec:swap-echo}

Let $A$ denote the anchor register and $B$ the candidate register. 
Define the two-register difference Hamiltonian as follows:
\begin{equation}
D_{AB}\coloneqq H_A \otimes I_B - I_A \otimes H_B. \label{eq:DAB}
\end{equation}
This Hamiltonian is useful because, on a joint energy eigenstate $|\phi_a\rangle_A|\phi_b\rangle_B$, it acts as
\begin{equation}
D_{AB} |\phi_a\rangle_A |\phi_b\rangle_B = (E_a - E_b) |\phi_a\rangle_A |\phi_b\rangle_B.
\end{equation}
Thus the eigenvalue of $D_{AB}$ is exactly the quantity needed for the energy comparison.

At first sight, implementing evolution under $D_{AB}$ for some time $\tau$ might appear nontrivial.
In fact, it follows directly from the fact that $H_A$ and $H_B$ act on different registers, they commute, and therefore
\begin{equation}
\e^{-\iu D_{AB}\tau} = \e^{-\iu H_A\tau}\otimes\e^{+\iu H_B\tau}. \label{eq:difference-evolution-factorization}
\end{equation}

This gives access to the phase $\e^{-\iu(E_a-E_b)\tau}$, but by itself it does not place that phase on some signal qubit that we can eventually process. 
For the energy comparator, we want a qubit whose two computational-basis states see opposite signs of the energy difference. 
A controlled-SWAP operation provides exactly the required mechanism.

Introduce a signal qubit $q$, and let
\begin{equation}
\cSWAP \coloneqq |0\rangle\!\langle0|_q\otimes I_{AB} + |1\rangle\!\langle1|_q\otimes\operatorname{SWAP}_{AB}, \label{eq:controlled-swap}
\end{equation}
denote the controlled-SWAP operation.
Now consider the following sequence. 
First apply $\cSWAP$ on all registers, then evolve $A$ forward and $B$ backward, and finally apply the $\cSWAP$ again on all registers (see \Cref{fig:echo} for its circuit, ignoring the first phase rotation gate on $q$). 
Formally, we denote this sequence as
\begin{equation}
\SWAPECHO(\tau) \coloneqq \cSWAP \bigl(I_q\otimes \e^{-\iu D_{AB}\tau} \bigr)\cSWAP . \label{eq:W}
\end{equation}

One can easily understand the role of the two $\cSWAP$s by following the two branches of the signal qubit. 
If $q=0$, neither $\cSWAP$ acts, so the two system registers undergo $\e^{-\iu D_{AB}\tau}$. 
If $q=1$, the first $\cSWAP$ exchanges $A$ and $B$ before the evolution. 
The forward evolution is therefore applied to what was originally the candidate register, while the backward evolution is applied to what was originally the anchor register. 
The second $\cSWAP$ restores the registers to their original locations. 
The net effect is to reverse the sign of the difference Hamiltonian, giving $\e^{+\iu D_{AB}\tau}$.
The resulting operation is therefore an exact qubit-controlled sign choice for the difference Hamiltonian.

\begin{lemma}[Energy-difference signal]
\label{lem:exact-signal}
The SWAP echo in~\eqref{eq:W} satisfies
\begin{equation}
\SWAPECHO(\tau) = |0\rangle\langle0|_q\otimes\e^{-\iu D_{AB}\tau} + |1\rangle\langle1|_q\otimes\e^{+\iu D_{AB}\tau} = \e^{-\iu Z_q\otimes D_{AB}\tau}, \label{eq:W-exact}
\end{equation}
where $Z_q$ is the Pauli-Z operator acting on the signal qubit $q$.
Furthermore, for any classically known scalar shift $s \in \mathbb{R}$, applying an initial single-qubit phase rotation $e^{-i s \tau Z_q}$ yields the ``shifted SWAP echo":
\begin{equation}
\SWAPECHO_s(\tau) \coloneqq \e^{-\iu s\tau Z_q}\SWAPECHO(\tau) = \e^{-\iu Z_q\otimes(D_{AB}+sI)\tau}. \label{eq:Ws}
\end{equation}
The quantum circuit for $\SWAPECHO_s(\tau)$ is provided in \Cref{fig:echo}.
\end{lemma}

\begin{proof}
See~\Cref{sec:proof-of-exact-signal-lemma}.    
\end{proof}

\begin{figure}[t]
\centering
\begin{quantikz}[row sep=0.42cm,column sep=0.42cm]
\lstick{$q$}
& \gate{\e^{-\iu s\tau Z}}
& \ctrl{1}
& \qw
& \ctrl{1}
& \qw \\
\lstick{$A$}
& \qw
& \swap{1}
& \gate{\e^{-\iu H\tau}}
& \swap{1}
& \qw \\
\lstick{$B$}
& \qw
& \targX{}
& \gate{\e^{+\iu H\tau}}
& \targX{}
& \qw
\end{quantikz}
\caption{
Circuit for the shifted SWAP echo:
$\SWAPECHO_s(\tau)=\exp[-\iu Z_q\otimes(D_{AB}+sI)\tau]$.
}
\label{fig:echo}
\end{figure}
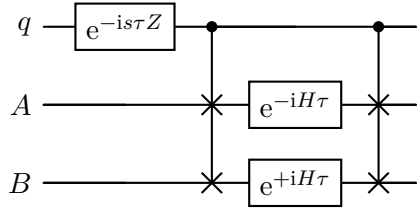

Suppose, following the intuition of the previous section, that the anchor state has energy $E_a$ and the candidate state has energy $E_b$. 
Since $D_{AB}$ has eigenvalue $E_a-E_b$ on this pair, the shifted SWAP echo $\SWAPECHO_s(\tau)$ acts on the signal qubit $q$ as
\begin{equation}
\SWAPECHO_s(\tau) |q\rangle|\phi_a\rangle|\phi_b\rangle = \e^{-\iu(E_a-E_b+s)\tau z_q} |q\rangle|\phi_a\rangle|\phi_b\rangle, \label{eq:signal-on-eigenpair}
\end{equation}
where $z_q=+1$ for $q=0$ and $z_q=-1$ for $q=1$.

Thus the two energies $E_a$ and $E_b$ never need to be measured individually. 
Their difference appears directly as opposite phases on the two signal-qubit branches,
\begin{equation}
\e^{-\iu(E_a-E_b+s)\tau} \qquad\text{and}\qquad \e^{+\iu(E_a-E_b+s)\tau}.
\end{equation}
The shift $s$ allows us to choose where the comparison threshold lies. In the next section, we choose $s$ so that “candidate not below the anchor” and “candidate lower by at least $\gamma$” lie on opposite sides of the phase threshold.
Finally, it is important to note that $\SWAPECHO_s(\tau)$ does not use any ctrl-$U_H$.

\section{Energy Comparator}
\label{sec:qsp-comparator}

The SWAP echo from the previous section gives us the quantity we need for energy comparison, but not yet in the form required by our SD and CSD algorithms. 
For a given pair of energy eigenstates $|\phi_a\rangle_A|\phi_b\rangle_B$, we must process this phase so that the signal qubit acts as a binary indicator, explicitly flagging when the candidate state has a sufficiently lower energy than the anchor state.

More precisely, if the candidate state energy is higher than or equal to the anchor state energy ($E_b\ge E_a$), the energy comparator must almost never produce a witness. 
On the other hand, if the candidate state energy is lower than the anchor state energy by at least the promised gap $\gamma$, that is, $E_a-E_b\ge\gamma$, then it should produce a witness with high probability. 
We do not need to decide what happens sharply for the intermediate case when $0<E_a-E_b<\gamma$. 
In particular, the energy comparator does not need to resolve arbitrarily small spacings between excited energies.

We therefore aim to construct a \emph{soft energy comparator}, rather than the ideal energy comparator discussed previously in \Cref{sec:comparison-intuition}.
For a chosen error parameter $\zeta>0$, let $r_{ab}$ denote the witness probability of this comparator, namely, the probability of obtaining the outcome $q=1$.
We require that $r_{ab}$ satisfy
\begin{align}
E_a-E_b\le0 &\quad\Longrightarrow\quad r_{ab}\le\zeta, \label{eq:soft-comparator-left} \\
E_a-E_b\ge\gamma &\quad\Longrightarrow\quad r_{ab}\ge1-\zeta. \label{eq:soft-comparator-right}
\end{align}
The transition between these two regimes is allowed to be smooth.
We construct this comparator using $\SWAPECHO_s(\tau)$ defined in~\eqref{eq:Ws} in two steps. 
First, we choose the shift $s$ and evolution time $\tau$ so that the two regimes in~\eqref{eq:soft-comparator-left}--\eqref{eq:soft-comparator-right} map to separated regions of a nonaliased phase interval. 
We then approximate a smooth transition between these regions and realize the resulting witness probability using Laurent quantum signal processing (Laurent QSP) applied directly to the $\SWAPECHO_s(\tau)$.

\subsection{\texorpdfstring{Choosing the Shift $s$ and Evolution Time $\tau$}{Choosing the Shift s and Evolution Time tau}}

We begin by placing the two comparison regimes in~\eqref{eq:soft-comparator-left}--\eqref{eq:soft-comparator-right} symmetrically around zero. 
The shifted echo $\SWAPECHO_s(\tau)$ allows us to add any known scalar $s$ to the energy difference. 
We choose $s=-\gamma/2$, so that the midpoint between the two relevant thresholds, $E_a-E_b=0$ and $E_a-E_b=\gamma$, is mapped to zero.

With this specific shift, we define the gap-shifted SWAP echo as
\begin{equation}
\SWAPECHO_\gamma(\tau) \coloneqq \e^{+\iu(\gamma\tau/2)Z_q}\SWAPECHO(\tau) = \e^{-\iu Z_q\otimes(D_{AB}-\gamma I/2)\tau}. \label{eq:Wgamma}
\end{equation}
When applied to the energy eigenstate pair $|\phi_a\rangle_A|\phi_b\rangle_B$, this operator generates a shifted relative phase given by
\begin{equation}
\theta_{ab} \coloneqq \tau\left( E_a-E_b-\frac{\gamma}{2} \right). \label{eq:theta-ab}
\end{equation}
This shift correctly places the two regimes on strictly opposite sides of zero:
\begin{align}
E_a-E_b\le0 &\Longleftrightarrow \theta_{ab}\le-\frac{\gamma\tau}{2}, \label{eq:soft-comparator-left-theta} \\
E_a-E_b\ge\gamma &\Longleftrightarrow \theta_{ab}\ge+\frac{\gamma\tau}{2}. \label{eq:soft-comparator-right-theta}
\end{align}

Next, we choose the evolution time $\tau$. 
Because the signal is a phase, energy differences that wrap around the unit circle could become ambiguous. 
We therefore restrict all occupied phases to a single interval on which the ordering of energy differences is preserved.
Note that $|E_a-E_b|\le 2 \|H\| \leq 2 \alpha_H$
for every energy pair that can occur in the two system registers. 
After the shift by $-\gamma/2$, this gives $|E_a-E_b-\sfrac{\gamma}{2}| \le 2 \alpha_H +\sfrac{\gamma}{2}$.
With that, we choose
\begin{equation}
\tau \coloneqq \frac{\pi}{4\left(2 \alpha_H+\frac{\gamma}{2}\right)}. \label{eq:tau}
\end{equation}
Then every occupied energy pair satisfies $|\theta_{ab}| \le \tau (2 \alpha_H+\sfrac{\gamma}{2}) = \frac{\pi}{4}$.
Hence all relevant phases lie in the interval $[-\pi/4,\pi/4]$. 
This prevents phase wrapping and, equally importantly, gives us an interval on which a simple real function of the phase is monotone.

We therefore define
\begin{equation}
y_{ab} \coloneqq \sin(2\theta_{ab}), \qquad \Delta \coloneqq \sin(\tau\gamma). \label{eq:y-and-Delta}
\end{equation}
Note that the function $\sin(2\theta)$ is monotone increasing on $[-\pi/4,\pi/4]$. 
Therefore the ordering established above in~\eqref{eq:soft-comparator-left-theta}--\eqref{eq:soft-comparator-right-theta} is preserved:
\begin{align}
E_a-E_b\le0 &\Longleftrightarrow \theta_{ab}\le-\frac{\gamma\tau}{2} \Longleftrightarrow y_{ab}\le-\Delta, \label{eq:y-left} \\
E_a-E_b\ge\gamma &\Longleftrightarrow \theta_{ab}\ge+\frac{\gamma\tau}{2} \Longleftrightarrow y_{ab}\ge+\Delta. \label{eq:y-right}
\end{align}
The energy-comparison problem has thus been reduced to distinguishing two well-separated regions of a real variable $y_{ab}$: values satisfying $y_{ab} \le -\Delta$ must yield a low witness probability, while values satisfying $y_{ab} \ge +\Delta$ must yield a high witness probability.
The parameter $\Delta$ determines how difficult this approximation will be. 
Using the fact that $\sin u\ge 2u/\pi$ for $u\in[0,\pi/2]$, we have
\begin{equation}
\label{eq:Delta-lower} \Delta \ge\frac{2\tau\gamma}{\pi}=\frac{\gamma}{2\left(2 \alpha_H+\frac{\gamma}{2}\right)}.
\end{equation}
The next subsection constructs a bounded soft step function across this separation and then approximates it by a polynomial suitable for QSP.

\subsection{Soft Energy Comparison Probability}
\label{sec:soft-comparator}

Fix a target energy comparator error $\zeta\in(0,1/4)$. 
We now choose the probability with which the signal qubit should indicate a witness as a function of the real signal variable $y$.

A sharp step at $y=0$ would reproduce the ideal comparison rule, but such a discontinuous function cannot be implemented exactly by a finite-degree polynomial. 
It is also unnecessary for spectral descent. 
We only require a small witness probability when $y\le-\Delta$ and a large witness probability when $y\ge+\Delta$; inside the transition region $(-\Delta,\Delta)$, the probability may vary smoothly.

We additionally require this variation to be monotone. 
For a fixed anchor state, a candidate state lying farther below it in energy corresponds to a larger value of $y$. 
Its witness probability should therefore not be smaller. 
This monotonicity will later be used to show that conditioning on a witness biases the candidate state toward lower energies.

For this purpose, we employ the following smooth function of $y$:
\begin{equation}
r_\star(y) \coloneqq \frac{1+\operatorname{erf}(\kappa y)}{2}, \qquad \kappa \coloneqq \frac{1}{\Delta} \sqrt{\log\frac{8}{\zeta}}. \label{eq:rstar}
\end{equation}
Note that the function $r_\star$ is monotone increasing, with values between $0$ and $1$. 
The parameter $\kappa$ controls how rapidly it changes around $y=0$. 
Because $\kappa$ scales as $1/\Delta$, the transition becomes sharper when the two comparison regions are closer together.
The particular constant in the definition of $\kappa$ is chosen so that the two regions already have a comfortable error margin before any polynomial approximation is introduced. 

Using the standard bound $\operatorname{erf}(u)\ge 1-\e^{-u^2}$ for $u\ge0$ in~\eqref{eq:rstar}, we obtain
\begin{align}
E_a-E_b\le0 &\Longleftrightarrow \theta_{ab}\le-\frac{\gamma\tau}{2} \Longleftrightarrow y_{ab}\le-\Delta \Longrightarrow r_\star(y) \le \frac{\zeta}{16}, \label{eq:rstar-left} \\
E_a-E_b\ge\gamma &\Longleftrightarrow \theta_{ab}\ge+\frac{\gamma\tau}{2} \Longleftrightarrow y_{ab}\ge+\Delta \Longrightarrow r_\star(y) \ge 1-\frac{\zeta}{16}. \label{eq:rstar-right}
\end{align}
Thus $r_\star$ already behaves as the desired soft energy comparator, with some room left for polynomial approximation error.

The remaining issue is that $r_\star$ is not itself a polynomial, whereas its Laurent-QSP construction, which we will present in the next section, requires a bounded polynomial of $y$. 
We therefore state the approximation result in the following lemma that we will use later.

\begin{lemma}[Polynomial approximation of $r_\star$]
\label{lem:soft-comparator-polynomial}
For every $\zeta\in(0,1/4)$, there exists a real polynomial $\widetilde r$ of degree
$O(\sfrac{1}{\Delta}\log\sfrac{1}{\zeta})$
such that $0\le\widetilde r(y)\le1$ for every $y\in[-1,1]$ and
\begin{equation}
\sup_{y\in[-1,1]} \left| \widetilde r(y)-r_\star(y) \right| \le \frac{\zeta}{32}. \label{eq:rtilde-monotone-close}
\end{equation}
\end{lemma}

\begin{proof}
    See~\Cref{app:soft-comparator-polynomial}.
\end{proof}

Combining \Cref{lem:soft-comparator-polynomial} with
\eqref{eq:rstar-left}--\eqref{eq:rstar-right}, we obtain
\begin{align}
E_a-E_b\le0 &\Longleftrightarrow \theta_{ab}\le-\frac{\gamma\tau}{2} \Longleftrightarrow y_{ab}\le-\Delta \Longrightarrow r_\star(y_{ab}) \le \frac{\zeta}{16}\Longrightarrow \widetilde r(y_{ab})\le\zeta, \label{eq:rtilde-left} \\
E_a-E_b\ge\gamma &\Longleftrightarrow \theta_{ab}\ge+\frac{\gamma\tau}{2} \Longleftrightarrow y_{ab}\ge+\Delta \Longrightarrow r_\star(y_{ab}) \ge 1-\frac{\zeta}{16}\Longrightarrow \widetilde r(y_{ab})\ge1-\zeta. \label{eq:rtilde-right}
\end{align}
Thus $\widetilde r(y)$ has exactly the behavior required of the witness probability.
Finally, using~\eqref{eq:Delta-lower}, the degree bound in~\Cref{lem:soft-comparator-polynomial} becomes $O((\alpha_H/\gamma+1)\log(1/\zeta))$.

The polynomial $\widetilde r$ specifies the witness probability we want, but we have not yet shown that a quantum circuit constructed using the SWAP echo can realize this probability. 
The next step is to connect $\widetilde r(y)$ to Laurent QSP.

\subsection{Realizing the Soft Energy Comparator with Laurent QSP}
\label{sec:laurent-qsp-synthesis}

Recall that for an energy pair $|\phi_a\rangle_A|\phi_b\rangle_B$, the shifted SWAP echo $\SWAPECHO_\gamma(\tau)$ supplies the phase $\theta_{ab}$ given in~\eqref{eq:theta-ab}.
We aim to construct a circuit that produces a witness outcome with probability $\widetilde{r}(y_{ab})$, where $y_{ab} = \sin(2\theta_{ab})$.

To connect this real variable to the Laurent-QSP signal, let $t\coloneqq\e^{\iu\theta}$ and $z\coloneqq t^2=\e^{2\iu\theta}.$
Then $\sin(2\theta) = \sfrac{(z-z^{-1})}{2\iu}.$
Consequently, the target probability polynomial $\widetilde{r}(y)$ can be mapped to a Laurent polynomial in $z$, which we define as
\begin{equation}
\mathcal R(z) \coloneqq \widetilde r\!\left( \frac{z-z^{-1}}{2\iu} \right). \label{eq:R-Laurent}
\end{equation}
Because $\widetilde{r}(y)$ is a polynomial of degree $n$, $\mathcal{R}(z)$ is a Laurent polynomial of degree at most $n$ in $z$. 
Moreover, from the range of $\widetilde r$ given in the statement of \Cref{lem:soft-comparator-polynomial}, we get $0\le\mathcal R(z)\le1$ for $|z|=1$.

We next relate this Laurent variable $z$ to the gap-shifted SWAP echo. 
On the fixed energy pair $|\phi_a\rangle_A|\phi_b\rangle_B$,~\eqref{eq:Wgamma} and~\eqref{eq:theta-ab} give
\begin{equation}
\SWAPECHO_\gamma(\tau)\big|_{a,b} = \e^{-\iu\theta_{ab}Z_q} = \begin{pmatrix} \e^{-\iu\theta_{ab}} & 0\\ 0 & \e^{+\iu\theta_{ab}} \end{pmatrix}.
\end{equation}
Defining $t_{ab}\coloneqq\e^{\iu\theta_{ab}}$, we can rewrite this as
\begin{equation}
\SWAPECHO_\gamma(\tau)\big|_{a,b} = \begin{pmatrix} t_{ab}^{-1} & 0\\ 0 & t_{ab} \end{pmatrix}. \label{eq:physical-Vt}
\end{equation}
This is precisely the signal form used by Laurent QSP: repeated applications of the SWAP echo generate powers of $t_{ab}$ and $t_{ab}^{-1}$, while interleaved single-qubit rotations acting on $q$ combine these powers into the desired Laurent polynomial.
Then since $z=t^2$, we have
$z_{ab} = t_{ab}^2 = \e^{2\iu\theta_{ab}},$
and hence
\begin{equation}
\mathcal R(z_{ab}) = \widetilde r\!\left( \frac{t_{ab}^2-t_{ab}^{-2}}{2\iu} \right) = \widetilde r\!\left( \sin(2\theta_{ab}) \right) = \widetilde r(y_{ab}). \label{eq:R-physical}
\end{equation}
Thus the probability polynomial constructed in the previous subsection is a Laurent polynomial of the phase variable supplied directly by the SWAP echo.

We formally state this required synthesis in the following lemma using the Laurent-QSP product-decomposition
theorem~\cite[Theorem~2]{Haah2019}.

\begin{lemma}[Laurent-QSP realization of $\widetilde{r}$]
\label{lem:laurent-qsp-comparator}
Let $\widetilde{r}(y)$ be the degree-$n$ polynomial defined in \Cref{lem:soft-comparator-polynomial}. 
There exists a unitary circuit $C_\zeta$, constructed using at most $2n$ queries to the shifted SWAP echo $\SWAPECHO_\gamma(\tau)$ interleaved with single-qubit gates on $q$, such that its action on any occupied energy pair is
\begin{equation}
C_\zeta |0\rangle_q |\phi_a\rangle_A |\phi_b\rangle_B = A(z_{ab}) |0\rangle_q |\phi_a\rangle_A |\phi_b\rangle_B + B(z_{ab}) |1\rangle_q |\phi_a\rangle_A |\phi_b\rangle_B, \label{eq:comparator-action}
\end{equation}
where $z_{ab} = e^{2i\theta_{ab}}$, and the amplitude coefficient explicitly satisfies
\begin{equation}
|B(z_{ab})|^2 = \mathcal{R}(z_{ab}) = \widetilde{r}(y_{ab}). \label{eq:qsp-witness-probability}
\end{equation}
Crucially, this circuit only queries $\SWAPECHO_\gamma(\tau)$ directly and does not require controlled access to it.
\end{lemma}

\begin{proof}
See~\Cref{app:laurent-qsp-synthesis}.
\end{proof}

We take $q=1$ to be the witness outcome and define its probability by
\begin{equation}
r_{ab} \coloneqq \Pr[q=1\mid a,b] = |B(z_{ab})|^2 = \widetilde r(y_{ab}). \label{eq:rab}
\end{equation}
The construction developed across the previous subsections can now be read as a direct chain from the energy difference to this witness probability:
\begin{align}
E_a-E_b\le0 &\Longleftrightarrow \theta_{ab}\le-\frac{\gamma\tau}{2} \Longleftrightarrow y_{ab}\le-\Delta \Longrightarrow r_{ab}\le\zeta, \label{eq:pair-pass} \\
E_a-E_b\ge\gamma &\Longleftrightarrow \theta_{ab}\ge+\frac{\gamma\tau}{2} \Longleftrightarrow y_{ab}\ge+\Delta \Longrightarrow r_{ab}\ge1-\zeta. \label{eq:pair-reject}
\end{align}
These are exactly the two comparison guarantees required by spectral descent.

We will also need the behavior of the comparator inside the transition region $0<E_a-E_b<\gamma$. 
Although no sharp decision is required there, the physical witness probability should remain close to a monotone function of the energy difference. 
Using~\eqref{eq:theta-ab}, we view the smooth target from~\eqref{eq:rstar} as
\begin{equation}
r_\star(E_a-E_b) \coloneqq r_\star\!\left( \sin\left[ 2\tau\left( E_a-E_b-\frac{\gamma}{2} \right) \right] \right). \label{eq:rstar-energy}
\end{equation}
Since all occupied phases lie in $[-\pi/4,\pi/4]$, the inner sine is monotone increasing in $E_a-E_b$, and $r_\star$ is itself monotone increasing. 
Thus $r_\star(E_a-E_b)$ is monotone throughout the occupied spectral range. 
Using the uniform approximation bound from~\eqref{eq:rtilde-monotone-close}, the physical witness probability tightly tracks this monotonic curve:
\begin{equation}
\left| r_{ab} - r_\star(E_a - E_b) \right| \le \frac{\zeta}{32}. \label{eq:rab-monotone-close}
\end{equation}
This property will later be used to show that conditioning on a witness biases the candidate state toward lower energies.

We summarize the energy comparator $C_\zeta$ and its cost below.

\begin{theorem}[Energy comparator cost]
\label{thm:qsp-filter}
For every $\zeta\in(0,1/4)$, there exists a unitary comparator circuit $C_\zeta$ whose witness probability $r_{ab}$ satisfies the strict comparison bounds in~\eqref{eq:pair-pass} and~\eqref{eq:pair-reject}, and the monotonic tracking bound in~\eqref{eq:rab-monotone-close} for all occupied energy pairs.
The comparator $C_\zeta$ uses
\begin{equation}
d = O\!\left( \left( \frac{\alpha_H}{\gamma}+1 \right) \log\frac{1}{\zeta} \right) \label{eq:qsp-degree}
\end{equation}
SWAP-echo queries, with sequential Hamiltonian-evolution time
\begin{equation}
T_{\rm cmp} = O\!\left( \frac{1}{\gamma} \log\frac{1}{\zeta} \right). \label{eq:qsp-comparator-time}
\end{equation}
The construction requires no ctrl-$\e^{-\iu Ht}$.
\end{theorem}

The scaling of the total evolution time $T_{\mathrm{cmp}}$ relies on the following argument. 
While the total number of SWAP-echo queries $d$ scales proportionally to $\alpha_H/\gamma$, each individual query applies time evolution for an elementary duration $\tau = \Theta(1/(\alpha_H + \gamma))$. 
Consequently, the dependence on $\alpha_H$ cancels out in the product $\Theta(d\tau)$, yielding a total evolution time proportional to $\gamma^{-1} \log(1/\zeta)$.

\section{Spectral Descent}
\label{sec:fresh-descent}

We now use the energy comparator $C_\zeta$ to turn the spectral descent principle of \Cref{sec:comparison-intuition} into an explicit ground-state preparation algorithm called Spectral Descent (SD).
This algorithm relies on a weak access model for the input state; that is, we assume that we have the ability to prepare fresh copies of the state $|\psi\rangle$ and do not have access to $U_\psi$.

SD maintains two system registers: the anchor register $A$ stores the current anchor state, while the candidate register $B$ is repeatedly prepared in the input state $|\psi\rangle$. 
One application of $C_\zeta$ to registers $A$, $B$, and the witness qubit $q$ is what we refer to as a \emph{comparison trial}. 
After each trial, SD measures the witness qubit $q$ in the computational basis $\{|0\rangle, |1\rangle\}$.
In what follows, the notation $q=0$ or $q=1$ refers to the corresponding classical measurement outcome, while $q$ by itself denotes the witness qubit.
If the outcome is $q=1$ (witness), the candidate and anchor registers are swapped. 
If the outcome is $q=0$ (no witness), the candidate register is discarded and the algorithm proceeds with the state remaining in the anchor register.

A sequence of at most $M$ comparison trials is what we refer to as a \emph{descent stage}. 
A stage ends as soon as the first witness outcome $q=1$ is observed, or after $M$ trials if no witness is found. 
The algorithm performs $R$ such stages.

\subsection{Spectral Descent Algorithm}
\label{sec:measured-algorithm}

The complete procedure is given in \Cref{alg:spectral-descent}. 
The parameter choices for $\zeta, M$, and $R$ appearing in Step 1 of the algorithm will be justified in \Cref{sec:convergence-sd}.

\begin{algorithm}[t]
\caption{Spectral Descent (SD)}
\label{alg:spectral-descent}
\begin{algorithmic}[1]
\Statex \hspace{-\algorithmicindent} \textbf{Input:} Fresh-copy access to the input state $|\psi\rangle$; known overlap lower bound $\eta$; known gap lower bound $\gamma$; target infidelity $\xi\in(0,1)$.

\Statex \hspace{-\algorithmicindent} \textbf{Output:} An anchor state $\rho_R$ satisfying $\langle\phi_0|\rho_R|\phi_0\rangle\ge 1-\xi$.

\State Set
\begin{equation}
\zeta\leftarrow\frac{\eta^2\xi}{256}, \qquad M\leftarrow \left\lceil \frac{\log 4}{\eta(1-\zeta)} \right\rceil, \qquad R\leftarrow \left\lceil \frac{\log\!\left(2/(\eta\xi)\right)} {\log(64/43)} \right\rceil \notag.
\end{equation}
\State Construct the energy comparator $C_\zeta$ from
\Cref{sec:qsp-comparator}.
\State Prepare the anchor register $A$ in the state $|\psi\rangle$.

\For{$t=1,\ldots,R$}
    \For{$m=1,\ldots,M$}
        \State Prepare the candidate register $B$ in the state
        $|\psi\rangle$.
        \State Initialize the witness qubit $q$ in $|0\rangle$.
        \State Apply $C_\zeta$ to joint register $qAB$.
        \State Measure $q$ in the computational basis $\{|0\rangle, |1\rangle\}$.
        \If{$q=1$}
            \State Apply $\operatorname{SWAP}_{AB}$, so that the
            accepted candidate becomes the new anchor.
            \State Discard $B$ and reset $q$.
            \State Terminate the current descent stage.
        \Else
            \State Discard $B$ and reset $q$.
        \EndIf
    \EndFor
\EndFor

\State \Return the state $\rho_R$ in the anchor register $A$.
\end{algorithmic}
\end{algorithm}

To illustrate, one descent stage, corresponding to the inner loop of \Cref{alg:spectral-descent}, is shown in \Cref{fig:measured-descent-stage}.

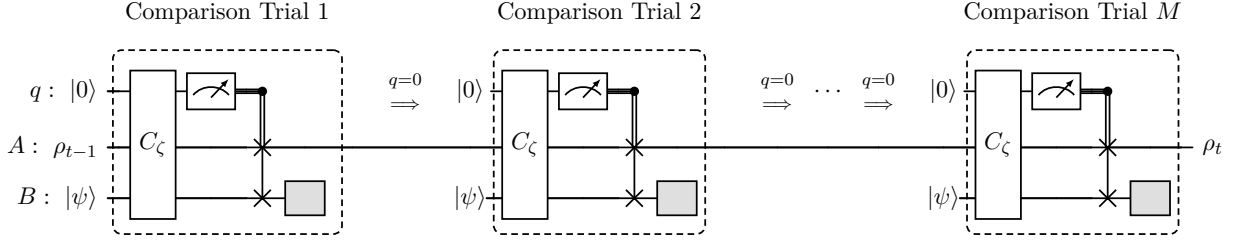
\begin{figure}[h]
\centering
\resizebox{\textwidth}{!}{
\begin{quantikz}[
    row sep=0.48cm,
    column sep=0.18cm
]
\lstick{$q:\ |0\rangle$}
&
\qw
&
\gate[3]{C_\zeta}
\gategroup[
    wires=3,
    steps=5,
    style={
        densely dashed,
        rounded corners,
        inner xsep=4pt,
        inner ysep=3pt
    },
    label style={
        label position=above,
        anchor=south,
        yshift=0.08cm
    }
]{Comparison Trial $1$}
&
\meter{}
&
\cw
&
\cwbend{1}
&
\setwiretype{n}
&
\setwiretype{n}\phantom{|\psi\rangle}
&
\gate[
    style={draw=none,minimum width=0.9cm}
]{\substack{q=0\\[1.3mm]\Longrightarrow}}
&
\inputD[
    draw=none,
    fill=none,
    inner sep=3pt
]{|0\rangle}
&
\gate[3]{C_\zeta}\setwiretype{q}
\gategroup[
    wires=3,
    steps=5,
    style={
        densely dashed,
        rounded corners,
        inner xsep=0pt,
        inner ysep=3pt
    },
    label style={
        label position=above,
        anchor=south,
        yshift=0.08cm
    }
]{Comparison Trial $2$}
&
\meter{}
&
\cw
&
\cwbend{1}
&
\setwiretype{n}
&
\setwiretype{n}\phantom{|\psi\rangle}
&
\gate[
    style={draw=none,minimum width=0.9cm}
]{\substack{q=0\\[1.3mm]\Longrightarrow}}
&
\setwiretype{n}\cdots
&
\gate[
    style={draw=none,minimum width=0.9cm}
]{\substack{q=0\\[1.3mm]\Longrightarrow}}
&
\inputD[
    draw=none,
    fill=none,
    inner sep=3pt
]{|0\rangle}
&
\gate[3]{C_\zeta}\setwiretype{q}
\gategroup[
    wires=3,
    steps=5,
    style={
        densely dashed,
        rounded corners,
        inner xsep=0pt,
        inner ysep=3pt
    },
    label style={
        label position=above,
        anchor=south,
        yshift=0.08cm
    }
]{Comparison Trial $M$}
&
\meter{}
&
\cw
&
\cwbend{1}
&
\setwiretype{n}
&
\setwiretype{n}
&
\setwiretype{n}
\\
\lstick{$A:\ \rho_{t-1}$}
&
\qw
&
\qw
&
\qw
&
\qw
&
\swap{1}
&
\qw
&
\qw
&
\qw
&
\qw
&
\qw
&
\qw
&
\qw
&
\swap{1}
&
\qw
&
\qw
&
\qw
&
\qw
&
\qw
&
\qw
&
\qw
&
\qw
&
\qw
&
\swap{1}
&
\qw
&
\qw
&
\rstick{$\rho_t$}
\\
\lstick{$B:\ |\psi\rangle$}
&
\qw
&
\qw
&
\qw
&
\qw
&
\targX{}
&
\gate[
    style={
        fill=gray!25,
        minimum width=0.48cm,
        minimum height=0.48cm,
        inner sep=0pt
    }
]{\phantom{X}}
&
\setwiretype{n}\phantom{|\psi\rangle}
&
\setwiretype{n}
&
\inputD[
    draw=none,
    fill=none,
    inner sep=0pt
]{|\psi\rangle}
&
\setwiretype{q}\qw
&
\qw
&
\qw
&
\targX{}
&
\gate[
    style={
        fill=gray!25,
        minimum width=0.48cm,
        minimum height=0.48cm,
        inner sep=0pt
    }
]{\phantom{X}}
&
\setwiretype{n}\phantom{|\psi\rangle}
&
\setwiretype{n}
&
\setwiretype{n}
&
\setwiretype{n}
&
\inputD[
    draw=none,
    fill=none,
    inner sep=0pt
]{|\psi\rangle}
&
\setwiretype{q}\qw
&
\qw
&
\qw
&
\targX{}
&
\gate[
    style={
        fill=gray!25,
        minimum width=0.48cm,
        minimum height=0.48cm,
        inner sep=0pt
    }
]{\phantom{X}}
&
\setwiretype{n}
&
\setwiretype{n}
\end{quantikz}
}
\caption{
The $t$-th descent stage of SD.
Each dashed box represents one comparison trial and contains the energy
comparator, the measurement-controlled anchor update, and the subsequent
discarding of the candidate register.
If the measured outcome is $q=0$, only the anchor register continues to the
next trial, while a new witness qubit is initialized in $|0\rangle$ and a
fresh candidate state $|\psi\rangle$ is prepared in $B$.
If the outcome is $q=1$, the candidate is swapped into $A$ and the descent
stage terminates.
The stage contains at most $M$ comparison trials.
}
\label{fig:measured-descent-stage}
\end{figure}

\subsection{Convergence to the Ground State}
\label{sec:convergence-sd}

We now prove that the output state $\rho_R$ of
\Cref{alg:spectral-descent} has ground-state fidelity at least $1-\xi$.
To prove this, we establish two properties of the descent.
The first is \emph{downward progress}: if the anchor is excited, successive descent stages drive it toward lower-energy levels. 
The second is \emph{ground-state stability}: once the anchor reaches the ground state, it remains there except for a small probability of an erroneous update.

We begin with an anchor state having a fixed energy and analyze one descent stage. 
Specifically, we determine the probability that the stage finds a witness ($q=1$) and the energy distribution of the candidate accepted upon finding one. 
These quantities describe the action of one stage, but they do not yet provide a convenient way to track its cumulative progress over multiple descent stages.

To quantify this progress, we introduce the notion of \emph{descent potential} that assigns zero potential to the ground-state energy and positive potential to every occupied excited-energy level. 
We show that one descent stage decreases this potential by a constant factor when the anchor is excited, while increasing it only with small probability when the anchor state is already the ground state. 
We then extend this fixed-energy analysis to the general anchor state produced by the algorithm and iterate the resulting potential bound over all $R$ descent stages.

For our analysis below, since the excited spectrum may contain degeneracies, we group together components of $|\psi\rangle$ having the same energy. 
Let $\mathcal E_0<\mathcal E_1<\cdots<\mathcal E_{s-1}$ denote the distinct occupied energies of $|\psi\rangle$ with $s \leq d$, and let
$\Pi_j$ denote the projector onto the eigenspace with energy $\mathcal E_j$. 
Define the corresponding spectral weights by $\pi_j \coloneqq \langle\psi|\Pi_j|\psi\rangle.$ 
These weights satisfy $\pi_j>0$ and $\sum_j\pi_j=1$. 
Since the ground state is nondegenerate and $|\psi\rangle$ has ground-state weight at least $\eta$, we have $\mathcal E_0=E_0$, $\Pi_0 = |\phi_0\rangle\langle\phi_0|$, and $\pi_0=p_0\ge\eta$.

\subsubsection{Probability of Finding a Witness in a Descent Stage}
\label{sec:measured-search}

For the moment, suppose that the anchor state is supported entirely in the eigenspace with energy $\mathcal E_j$, and that the candidate state lies in the eigenspace with energy $\mathcal E_k$. 
Let $r_{jk}$ denote the probability that the energy comparator $C_\zeta$ produces the witness outcome $q=1$ on this pair. 
By~\eqref{eq:theta-ab},~\eqref{eq:y-and-Delta}, and~\eqref{eq:rab}, we have $r_{jk}=\widetilde r(\sin[2\tau(\mathcal E_j-\mathcal E_k-\gamma/2)])$.
Since a fresh candidate state $|\psi\rangle$ has spectral weight $\pi_k$ at energy $\mathcal E_k$, the total witness probability in one comparison trial is simply
\begin{equation}
\lambda_j \coloneqq \sum_k \pi_k r_{jk}. \label{eq:lambda-j}
\end{equation}

We bound $\lambda_j$ for two distinct cases using the comparator guarantees established in~\eqref{eq:pair-pass} and~\eqref{eq:pair-reject}.
First, suppose the anchor state has already reached the ground state~($j=0$).
This implies that $\mathcal E_0-\mathcal E_k\le0$ for every $k$, and hence $r_{0k}\le\zeta$. 
The total witness probability is therefore strictly bounded by the comparator error:
\begin{equation}
\lambda_0 = \sum_k\pi_k r_{0k} \le \zeta \sum_k \pi_k =\zeta. \label{eq:lambda0-upper}
\end{equation}
Second, suppose the anchor state is an excited state~($j>0$). 
Its energy is separated from the ground-state component of the candidate state by at least $\gamma$, that is, $\mathcal E_j-\mathcal E_0\ge\gamma$, and therefore $r_{j0}\ge1-\zeta$. 
Keeping only the ground-state contribution in~\eqref{eq:lambda-j}, we get
\begin{equation}
\lambda_j \ge \pi_0 r_{j0} \ge \eta(1-\zeta), \qquad j>0. \label{eq:lambda-excited-lower}
\end{equation}
Thus
\begin{equation}
\lambda_0\le\zeta, \qquad \lambda_j\ge\eta(1-\zeta) \quad\text{for every }j>0. \label{eq:lambda-separation}
\end{equation}
A ground-state anchor state therefore produces a witness with probability that is bounded from above by the comparator error $\zeta$, whereas every excited anchor state produces one with probability at least of order $\eta$.

From~\eqref{eq:lambda-excited-lower}, observe that a single comparison trial may still have small success probability when $\eta$ is small. 
This is why each descent stage repeats the energy comparison with fresh candidate states while retaining the same anchor state. 
Since a no-witness outcome ($q=0$) leaves the anchor state unchanged, every trial within the same stage again has witness probability $\lambda_j$. 
If the stage contains at most $M$ comparison trials, its total witness probability is
\begin{equation}
s_j^{\rm SD} \coloneqq 1-(1-\lambda_j)^M. \label{eq:sj-meas}
\end{equation}

To boost this stage witness probability to a constant, we choose $M$ as follows:
\begin{equation}
M \coloneqq \left\lceil \frac{\log 4}{\eta(1-\zeta)} \right\rceil. \label{eq:M-choice}
\end{equation}
For any excited anchor state ($j>0$), substituting the bound $\lambda_j \ge \eta(1-\zeta)$ from~\eqref{eq:lambda-excited-lower} into the no-witness probability ($1-s_j^{\rm SD}$) gives:
\begin{equation}
1-s_j^{\rm SD} =(1-\lambda_j)^M \le\e^{-M\lambda_j} \le\e^{-M\eta(1-\zeta)} \le\frac14,
\end{equation}
where the first inequality uses $1-x\le\e^{-x}$, the second inequality uses~\eqref{eq:lambda-excited-lower}, and the last inequality uses~\eqref{eq:M-choice}.
Therefore, executing $M = \mathcal{O}(1/\eta)$ trials ensures a constant witness probability:
\begin{equation}
s_j^{\rm SD}\ge\frac34, \qquad j>0. \label{eq:excited-measured-success}
\end{equation}
Conversely, if the anchor state is already the ground state, each trial generates a witness with probability at most $\zeta$ (see~\eqref{eq:lambda0-upper}). 
Applying the standard union bound, the cumulative probability of a descent stage producing a false-witness is bounded by
\begin{equation}
s_0^{\rm SD} \leq M\zeta. \label{eq:delta-g-meas}
\end{equation}
This is an upper bound on the probability that the ground-state anchor is erroneously replaced during one stage.

Finding a witness does not guarantee that the successful candidate state is exactly the ground state. 
We must also determine which candidate energies are statistically favored upon finding a witness. 
For a single comparison trial with an anchor energy $\mathcal{E}_j$, the joint probability of drawing candidate energy $\mathcal{E}_k$ and measuring $q=1$ is $\pi_k r_{jk}$. By Bayes' theorem, the conditional distribution of the candidate energy, given that a witness was observed ($q=1$), is
\begin{equation}
q_j(k) \coloneqq \Pr[ k \mid j,\text{ witness}] = \frac{\pi_k r_{jk}}{\lambda_j}. \label{eq:qjk}
\end{equation}
Because a descent stage involves repeated trials, we must verify that this repetition does not change the above conditional energy distribution.

\begin{lemma}
\label{lem:measured-direction-preserved}
Fix an anchor energy level $\mathcal{E}_j$. Conditioned on a descent stage successfully finding a witness within its maximum of $M$ comparison trials, the energy of the first successful candidate state is exactly distributed according to $q_j(k)$ defined in~\eqref{eq:qjk}.
\end{lemma}

\begin{proof}
See~\Cref{app:measured-witness-direction}.
\end{proof}

\subsubsection{Descent Potential}
\label{sec:descent-potential}

The above analysis shows that an excited anchor finds a witness with constant probability. 
Finding a witness alone, however, does not
guarantee that the accepted candidate energy is lower than the anchor energy. 
Recall that the energy comparator $C_\zeta$ has a nonzero approximation error, and it also has a smooth transition region when two excited energies differ by less than $\gamma$. 
We therefore need to quantify the average spectral progress made by an accepted candidate after each descent stage.

For each occupied energy level $\mathcal{E}_j$, we define the \emph{descent potential} $W_j$ as the total spectral weight of the input state $|\psi\rangle$ lying strictly below $\mathcal{E}_j$:
\begin{equation}
W_j \coloneqq \sum_{k<j}\pi_k. \label{eq:Wj}
\end{equation}
 
By definition, the ground state has no occupied spectrum below it, yielding $W_0 = 0$.
While for every excited level $j>0$, the ground-state weight is included in the sum, so
\begin{equation}
W_j\ge\pi_0\ge\eta, \qquad j>0. \label{eq:W-excited-lower}
\end{equation}

Fix an excited anchor energy $\mathcal E_j$. 
From \Cref{lem:measured-direction-preserved}, we have that conditioned on observing a witness, the accepted candidate has energy $\mathcal E_k$ with probability $q_j(k)$. 
We define the potential of this accepted candidate state to be a random variable $W_{\rm cand}$, where $ W_{\rm cand} = W_k$ when the accepted candidate state has energy $\mathcal E_k$.
Therefore, the expected potential under $q_j(k)$ is
\begin{equation}
\E[W_{\rm cand}\mid j,\text{ witness}] = \sum_k q_j(k)W_k.
\end{equation}
Note that smaller $W_k$ corresponds to a candidate energy closer to the bottom of the occupied spectrum.

To understand the above expectation, first consider an idealized scenario where the QSP approximation error is zero, and we only consider genuinely lower candidate energy levels $k<j$. 
If these lower-energy components were sampled naively according to their original input weights, their normalized distribution would be $\mu_j(k) = \pi_k / W_j$. 
Now, calculating the expected potential under this naive distribution yields a contraction by a factor of roughly one-half:
\begin{equation}
\E_{k\sim\mu_j}[W_k] = \frac{W_j^2-\sum_{k<j}\pi_k^2}{2W_j} \le \frac{W_j}{2}. \label{eq:unweighted-half-main}
\end{equation}
We explicitly derive this inequality in \Cref{app:candidate-potential}.

The ideal target comparator can only decrease this expected potential. 
Recall the ideal target probability $r^\star_{jk} \coloneqq r_\star(\mathcal{E}_j - \mathcal{E}_k)$ introduced in \Cref{sec:qsp-comparator}. 
For energy levels ordered as $k < \ell < j$, the energy difference satisfies $\mathcal{E}_j - \mathcal{E}_k > \mathcal{E}_j - \mathcal{E}_\ell$. 
Because $r_\star$ is strictly monotonically increasing, it follows that $r^\star_{jk} \ge r^\star_{j\ell}$. 
Concurrently, the descent potentials obey the strict ordering $W_k \le W_\ell$. 
Thus the ideal comparator assigns higher witness probabilities to energies that possess lower descent potentials. 

Define the ideal conditional distribution as
\begin{equation}
\nu_j^\star(k) \coloneqq \frac{\pi_k r^\star_{jk}} {\sum_{\ell<j}\pi_\ell r^\star_{j\ell}}, \qquad k<j. \label{eq:v-j-star-define}
\end{equation}
Because $r^\star_{jk}$ decreases with $k$ while $W_k$ increases with $k$, these quantities are negatively correlated under the naive distribution $\mu_j$. 
This negative correlation guarantees that the expected potential under the ideal comparator is strictly less than the naive expectation in~\eqref{eq:unweighted-half-main}:
\begin{equation}
\E_{k\sim\nu_j^\star}[W_k] \le \E_{k\sim\mu_j}[W_k] \le \frac{W_j}{2}. \label{eq:ideal-candidate-half-main}
\end{equation}
We prove the first inequality explicitly in~\Cref{app:candidate-potential}.

The distribution $q_j(k)$ differs from this ideal distribution $ \nu_j^\star(k)$ in two ways. 
First, the QSP polynomial approximation introduces an error $|r_{jk} - r^\star_{jk}| \le \zeta/32$. 
Second, the energy comparator $C_\zeta$ may produce a small witness probability on candidate energies $k\ge j$. 
The following lemma bounds both these effects.

\begin{lemma}[Potential of an accepted candidate]
\label{lem:successful-candidate-potential}
Fix an excited anchor level $\mathcal{E}_j$ ($j>0$), and assume the comparator error $\zeta$ is sufficiently small such that $\zeta / (32\eta(1-\zeta)) \le 1/16$. 
Conditioned on one comparison trial producing a witness, the accepted candidate satisfies
\begin{equation}
\mathbb{E}[W_{\mathrm{cand}} \mid j, \mathrm{witness}] \le \frac{9}{16}W_j + \frac{\zeta}{\eta(1-\zeta)}. \label{eq:candidate-potential}
\end{equation}
By \Cref{lem:measured-direction-preserved}, this bound also holds for the first successful candidate state obtained across an entire descent stage.
\end{lemma}

\begin{proof}
See~\Cref{app:candidate-potential}.
\end{proof}

We can now bound the expected potential after one complete descent stage. 
For an excited anchor level $j>0$, a descent stage either leaves the anchor unchanged or replaces it by the first successful candidate state. 
Let $W_{\rm new}$ denote the descent potential of the anchor energy at the end of this stage. 
If the stage fails to find any witness (probability $1 - s_j^{\mathrm{SD}}$), the anchor state is retained, yielding $W_{\mathrm{new}} = W_j$. 
If the stage successfully finds a witness (probability $s_j^{\mathrm{SD}}$), the anchor state is replaced by the new accepted candidate state, which satisfies the conditional expectation from \Cref{lem:successful-candidate-potential}. 
Using the success probability bound $s_j^{\mathrm{SD}} \ge 3/4$ from~\eqref{eq:excited-measured-success}, the expected potential at the end of the stage evaluates to
\begin{align}
\mathbb{E}[W_{\mathrm{new}} \mid j] &\le (1 - s_j^{\mathrm{SD}})W_j + s_j^{\mathrm{SD}} \left( \frac{9}{16}W_j + \frac{\zeta}{\eta(1-\zeta)} \right) \\
&\le \frac{43}{64}W_j + \frac{\zeta}{\eta(1-\zeta)}. \label{eq:measured-stage-contraction}
\end{align}

If the anchor state is already the ground state ($j=0$), its initial potential is $W_0=0$. 
The potential can only increase if a false witness triggers an erroneous update. 
Because the probability of any false witness occurring during the entire stage is bounded as $s_0^{\rm SD} \leq M\zeta$, and the maximum possible potential is bounded trivially by $W_k \le 1$, we obtain the strict bound:
\begin{equation}
\mathbb{E}[W_{\mathrm{new}} \mid 0] \le M\zeta. \label{eq:measured-ground-potential}
\end{equation}
Equations~\eqref{eq:measured-stage-contraction} and~\eqref{eq:measured-ground-potential} formally capture the stability of the descent: an excited anchor tends to move toward lower-energy levels (downward progress), while the ground state anchor remains highly stable (ground-state stability).

We now extend the one-stage descent potential contraction bounds over repeated descent stages.
The analysis above treated an anchor state with a fixed energy $\mathcal E_j$, whereas the actual anchor state after $t$ decent stages, say $\rho_t$, need not be an energy eigenstate.
Because the anchor register is initialized in $|\psi\rangle$ at the beginning of the SD algorithm (see Step 3 of \Cref{alg:spectral-descent}) and every accepted candidate is obtained from a fresh copy of $|\psi\rangle$, every anchor state produced by SD remains supported on the occupied energy levels of $|\psi\rangle$.

Recall that $\Pi_j$ denotes the projector onto the eigenspace with energy $\mathcal E_j$.
This means that the quantity $\Tr(\Pi_j\rho_t)$ is the population of the anchor state at energy level $\mathcal E_j$.
Equivalently, it is the probability that a projective measurement of $H$ on $\rho_t$ returns the energy $\mathcal E_j$.
We define the mixed-state descent potential $V_t$ as the average of the potentials $W_j$ over these energy-level populations:
\begin{equation}
V_t \coloneqq \sum_j W_j \Tr(\Pi_j\rho_t). \label{eq:Vt}
\end{equation}
Because $W_0=0$, there is no ground-state population term in $V_t$.
Moreover, since $W_j\ge\eta$ for every $j>0$, we have
\begin{equation}
V_t \ge \eta \sum_{j>0}\Tr(\Pi_j\rho_t). \label{eq:Vt-excited-population}
\end{equation}
Thus, a small value of $V_t$ guarantees that the anchor state has little population outside the ground-state eigenspace.

The energy comparator $C_\zeta$ is block diagonal in the energy eigenspaces (see~\eqref{eq:comparator-action}), so the evolution of the energy-level populations after one decent stage can be analyzed by averaging the fixed-energy bounds derived above.
As shown in~\Cref{app:measured-potential-recursion} in detail, this gives the following recurrence relation:
\begin{equation}
V_{t+1} \le \frac{43}{64}V_t + M\zeta + \frac{\zeta}{\eta(1-\zeta)}. \label{eq:V-recursion-measured}
\end{equation}

To drive the final error below the target infidelity $\xi$, we choose the comparator $\zeta$ in the following way:
\begin{equation}
\zeta \coloneqq \frac{\eta^2\xi}{256}. \label{eq:zeta-measured}
\end{equation}
Since $\eta,\xi\le1$, this satisfies $\zeta<1/4$. 
Moreover, $1-\zeta\ge3/4$, and therefore
\begin{equation}
\frac{\zeta}{32\eta(1-\zeta)} \le \frac{\eta\xi}{6144} \le \frac1{16},
\end{equation}
so the assumption of \Cref{lem:successful-candidate-potential} is satisfied.

With $\zeta$ defined, we next bound the additive errors, that is, the second and third terms on the right-hand side of the inequality in~\eqref{eq:V-recursion-measured}.
From~\eqref{eq:M-choice}, we have $M\le \log4/[\eta(1-\zeta)]+1<3/\eta$. 
The cumulative false-witness probability (second term in~\eqref{eq:V-recursion-measured}) is therefore bounded by
\begin{equation}
M\zeta < \frac{3\eta\xi}{256}, \label{eq:delta-g-final}
\end{equation}
while the third term in~\eqref{eq:V-recursion-measured} is bounded by
\begin{equation}
\frac{\zeta}{\eta(1-\zeta)} \le \frac{\eta\xi}{192}. \label{eq:false-candidate-final}
\end{equation}

As mentioned before, the anchor register is initialized in the state $|\psi\rangle$. Because $W_j \le 1$ for all levels, the initial mixed-state potential is trivially bounded by $V_0 \le 1$. 
Iterating the recurrence relation in~\eqref{eq:V-recursion-measured} over $R$ consecutive descent stages yields
\begin{align}
V_R &\le \left(\frac{43}{64}\right)^R V_0 + \sum_{t=0}^{R-1} \left(\frac{43}{64}\right)^t \left( M\zeta + \frac{\zeta}{\eta(1-\zeta)} \right) \nonumber \\
&\le \left(\frac{43}{64}\right)^R + \frac{1}{ 1-\frac{43}{64} } \left( M\zeta + \frac{\zeta}{\eta(1-\zeta)} \right) \nonumber \\
&= \left(\frac{43}{64}\right)^R + \frac{64}{21} \left( M\zeta + \frac{\zeta}{\eta(1-\zeta)} \right). \label{eq:VR-bound}
\end{align}
Using the bounds in~\eqref{eq:delta-g-final} and~\eqref{eq:false-candidate-final}, the second term on the right-hand side of the above inequality is at most $13\eta\xi/252$.

We now choose the total number of descent stages $R$:
\begin{equation}
R \coloneqq \left\lceil \frac{ \log\!\left(2/(\eta\xi)\right) }{ \log(64/43) } \right\rceil. \label{eq:R-measured}
\end{equation}
This choice explicitly guarantees that $(43/64)^R \le \eta\xi/2$, and therefore
\begin{equation}
V_R \le \left( \frac12+\frac{13}{252} \right)\eta\xi = \frac{139}{252}\eta\xi < \eta\xi. \label{eq:VR-final}
\end{equation}

Finally, we map this bound on the potential directly to the target ground-state fidelity.
For every excited level $j>0$, recall that we have $W_j\ge\eta$ from~\eqref{eq:W-excited-lower}. 
Substituting this into the definition of $V_R$ gives:
\begin{equation}
V_R = \sum_{j>0}W_j\Tr(\Pi_j\rho_R) \ge \eta\sum_{j>0}\Tr(\Pi_j\rho_R).
\end{equation}
Together with~\eqref{eq:VR-final}, this implies $\sum_{j>0}\Tr(\Pi_j\rho_R)<\xi$. 
Since the ground state is nondegenerate,
\begin{equation}
\langle \phi_0|\rho_R|\phi_0\rangle \ge 1-\xi. \label{eq:measured-final-fidelity}
\end{equation}
This concludes the convergence analysis for \Cref{alg:spectral-descent}.

\subsection{Main Result}
\label{sec:measured-resources}

\begin{theorem}[Ground-State Preparation via Spectral Descent]
\label{thm:measured-gsp}
Consider an instance of the GSP problem (Problem~\ref{prob:gsp}), and suppose that fresh copies of the input state $|\psi\rangle$ can be prepared.
Choose the elementary evolution time $\tau$ as in \eqref{eq:tau} and define $U_H\coloneqq\e^{-\iu H\tau}$.
Then the SD algorithm (Algorithm~\ref{alg:spectral-descent}), with the parameter choices in
\eqref{eq:zeta-measured}, \eqref{eq:M-choice}, and \eqref{eq:R-measured},
outputs a state $\rho_R$ satisfying
\begin{equation}
\langle\phi_0|\rho_R|\phi_0\rangle\ge1-\xi.
\end{equation}
The SD algorithm uses
\begin{enumerate}[label=(\roman*),leftmargin=*]
\item
    $T_H^{\rm SD} = O\!\left( \frac1{\eta\gamma} \log^2\frac1{\eta\xi} \right) = \widetilde O\!\left( \frac1{\eta\gamma} \right)$
    total Hamiltonian-evolution time;
    \item
    $N_H^{\rm SD} = O\!\left( \frac{\alpha_H}{\eta \gamma} \log^2\frac1{\eta\xi} \right) = \tilde O\!\left( \frac{\alpha_H}{\eta \gamma} \right)$
    queries to the uncontrolled oracles $U_H$ and $U_H^\dagger$; and

    \item
    $N_\psi^{\rm SD} = O\!\left( \frac1\eta \log\frac{2}{\eta\xi} \right) = \tilde O\!\left( \frac1\eta \right)$
    fresh copies of $|\psi\rangle$.
\end{enumerate}
Crucially, it does not require access to $\mathrm{ctrl}\text{-}U_H$.
\end{theorem}

\begin{proof}
From~\eqref{eq:M-choice}, $M=O(1/\eta)$, while~\eqref{eq:R-measured} gives $R=O(\log(1/(\eta\xi)))$.
The algorithm therefore makes at most $MR = O(1/\eta \log (1/(\eta\xi)) )$ calls to the comparator $C_\zeta$.
Since $\zeta=\eta^2\xi/256$ (see~\eqref{eq:zeta-measured}), we have $\log(1/\zeta)=O(\log(1/(\eta\xi)))$. 
By \Cref{thm:qsp-filter}, one comparator call requires Hamiltonian-evolution time $O(1/\gamma \log 1/(\eta\xi)).$
Since we have $O(1/\eta \log (1/(\eta\xi)) )$ comparator calls in total, we get the total Hamiltonian-evolution time:
\begin{equation}
T_H^{\rm SD} = O\!\left( \frac{1}{\eta\gamma} \log^2\frac{1}{\eta\xi} \right) = \widetilde O\!\left( \frac{1}{\eta\gamma} \right), \label{eq:measured-H-time}
\end{equation}
which proves item~$(i)$ of the theorem statement

From~\eqref{eq:qsp-degree}, one comparator uses $d=O((\alpha_H/\gamma+1)\log(1/(\eta\xi)))$ queries to $\SWAPECHO_\gamma(\tau)$, and hence the same number, up to a constant
factor, of queries to $U_H$ and $U_H^\dagger$. Multiplying this cost
by the number of comparator calls gives
\begin{equation}
N_H^{\rm SD} = O\!\left( \frac{\alpha_H}{\eta \gamma} \log^2\frac1{\eta\xi} \right) = \tilde O\!\left( \frac{\alpha_H}{\eta \gamma} \right),
\end{equation}
which proves item~$(ii)$ of the theorem statement.

Finally, one copy of $|\psi\rangle$ initializes the anchor register,
and each comparison trial uses one fresh candidate state. Thus
\begin{equation}
N_\psi^{\rm SD} \le 1+MR = O\!\left( \frac1\eta \log\frac{2}{\eta\xi} \right),
\end{equation}
which proves item~$(iii)$ of the theorem statement.
\end{proof}

\section{Coherent Spectral Descent}
\label{sec:coherent-descent}

We now use the energy comparator $C_\zeta$ to construct our second GSP algorithm, called Coherent Spectral Descent (CSD).
We will show that CSD achieves $\widetilde O(1/(\gamma\sqrt{\eta}))$ total Hamiltonian-evolution time using only uncontrolled Hamiltonian evolution, answering the question posed in the introduction.
Recall that unlike SD, this algorithm assumes access to an input-state preparation unitary $U_\psi$ satisfying $U_\psi|0^n\rangle=|\psi\rangle$ and to its inverse $U_\psi^\dagger$.

CSD maintains the exact same register structure as SD: an anchor register $A$ that stores the current anchor state, a candidate register $B$ that stores the candidate state, and a single witness qubit $q$.
At the beginning of each descent stage, the anchor register $A$ contains the current anchor state. 
The candidate register $B$ and the witness qubit $q$ are initialized in the states $|0^n\rangle$ and $|0\rangle$, respectively. 
The algorithm then applies the unitary $U_\psi$ to prepare the candidate state $|\psi\rangle$ in register $B$, and subsequently applies the energy comparator $C_\zeta$ jointly to the register $qAB$.

Recall from~\eqref{eq:lambda-excited-lower} that for an excited anchor state with energy $\mathcal E_j$, this energy comparator produces the witness outcome $q=1$ with a probability bounded by $\lambda_j \ge \eta(1-\zeta)$.
SD amplifies this initial probability to a constant (at least $3/4$) by performing $O(1/\eta)$ independent comparison trials, discarding and preparing a fresh candidate state after every unsuccessful trial.
CSD instead keeps $B$ and $q$ coherent throughout the stage and applies fixed-point amplitude amplification (FPAA)~\cite{Yoder2014} to the witness subspace.
This amplifies the witness probability to at least $3/4$ using $O(1/\sqrt\eta)$ coherent applications of $C_\zeta$, $C_\zeta^\dagger$, $U_\psi$, and $U_\psi^\dagger$ as we will see in what follows.
After the amplification, CSD performs the same anchor update as SD.
The only difference between the two algorithms is that the stronger input-state access allows CSD to substitute repeated comparison trials with coherent amplitude amplification, yielding a quadratic improvement in the algorithmic dependence on the overlap $\eta$.

\subsection{Coherent Spectral Descent Algorithm}
\label{sec:coherent-algorithm}

The complete ground-state preparation procedure via CSD is given in~\Cref{alg:coherent-descent}.
To illustrate, one descent stage is shown in \Cref{fig:coherent-descent-stage}.

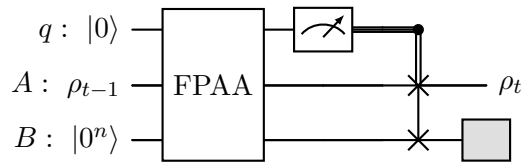
\begin{figure}[h]
\centering
\begin{quantikz}[row sep=0.45cm,column sep=0.40cm]
\lstick{$q:\ |0\rangle$}
& \gate[3]{\text{FPAA}}
& \meter{}
& \cw
& \cwbend{1}
& \setwiretype{n}
\\
\lstick{$A:\ \rho_{t-1}$}
& \qw
& \qw
& \qw
& \swap{1}
& \rstick{$\rho_t$}\qw
\\
\lstick{$B:\ |0^n\rangle$}
& \qw
& \qw
& \qw
& \targX{}
& \gate[
    style={
        fill=gray!25,
        minimum width=0.48cm,
        minimum height=0.48cm,
        inner sep=0pt
    }
]{\phantom{X}}
\end{quantikz}
\caption{
The $t$-th descent stage of CSD.
The Length-$L$ FPAA sequence described in \Cref{sec:coherent-witness-search} coherently amplifies the witness subspace, after which the witness qubit is measured.
If the measurement outcome is $q=1$, an ordinary SWAP places the accepted candidate state in the anchor register; if the outcome is $q=0$, no SWAP is applied.
The candidate register and witness qubit are then discarded or reset, while the anchor register is retained for the next stage.
}
\label{fig:coherent-descent-stage}
\end{figure}

The specific parameter choices defining the comparator error $\zeta$, the amplification sequence length $L$, and the total number of descent stages $R$ appearing in Step 1 of the algorithm will be derived and justified in Sections~\ref{sec:coherent-witness-search} and~\ref{sec:coherent-convergence}.

\begin{algorithm}[t]
\caption{Coherent Spectral Descent (CSD)}
\label{alg:coherent-descent}
\begin{algorithmic}[1]
\Statex \hspace{-\algorithmicindent} \textbf{Input:} Input-state preparation unitary $U_\psi$ and its inverse $U_\psi^\dagger$; known overlap lower bound $\eta$; known gap lower bound $\gamma$; target infidelity $\xi\in(0,1)$.

\Statex \hspace{-\algorithmicindent} \textbf{Output:} An anchor state $\rho_R$ satisfying $\langle\phi_0|\rho_R|\phi_0\rangle\ge1-\xi$.

\State Set
\begin{equation}
\zeta\leftarrow\frac{\eta^2\xi}{256}, \qquad L\leftarrow \text{smallest odd integer } \ge \frac{\operatorname{arcosh}2} {\sqrt{\eta(1-\zeta)}}, \qquad R\leftarrow \left\lceil \frac{\log\!\left(2/(\eta\xi)\right)} {\log(64/43)} \right\rceil \notag.
\end{equation}
\State Construct the energy comparator $C_\zeta$ from \Cref{sec:qsp-comparator}.
\State Prepare the anchor register $A\leftarrow U_\psi|0^n\rangle$.

\For{$t=1,\ldots,R$}
    \State Initialize $B\leftarrow|0^n\rangle$ and $q\leftarrow|0\rangle$.
    \State Apply $V=C_\zeta(I_q\otimes I_A\otimes U_\psi)$, then apply the $\ell=(L-1)/2$ FPAA iterates in~\eqref{eq:FPAA-sequence}.
    \State Measure $q$ in the computational basis $\{|0\rangle, |1\rangle\}$.
    \If{$q=1$}
        \State Apply $\operatorname{SWAP}_{AB}$, so that the accepted candidate becomes the new anchor.
    \EndIf
    \State Reset registers $B$ and $q$
\EndFor

\State \Return the state $\rho_R$ in the anchor register $A$.
\end{algorithmic}
\end{algorithm}

\subsection{Convergence to the Ground State}
\label{sec:coherent-convergence}

We now prove that the output state $\rho_R$ of \Cref{alg:coherent-descent} has ground-state fidelity of at least $1-\xi$.
The convergence analysis follows the same overall structure as the convergence analysis for SD established in Section~\ref{sec:fresh-descent}.

As before, we first suppose that the anchor state is supported entirely in an eigenspace with energy $\mathcal E_j$ and analyze one descent stage.
There are two CSD-specific questions that we will address in what follows.
First, we must determine how FPAA amplifies the probability of finding a witness.
Second, we must verify that amplifying the witness branch does not change the conditional distribution of the candidate energy once a witness is obtained.
After establishing these two facts, the descent-potential analysis from SD can be reused directly.

\subsubsection{Probability of Finding a Witness in a Descent Stage}
\label{sec:coherent-witness-search}

Fix an anchor state $|\varphi_j\rangle$ supported entirely in the eigenspace with energy $\mathcal E_j$.
At the beginning of a descent stage, the candidate register $B$ is initialized in $|0^n\rangle$ and the witness qubit $q$ is initialized in $|0\rangle$.
We first prepare the candidate state $|\psi\rangle_B$ by applying $U_\psi$ to register $B$, and then apply the comparator $C_\zeta$ to the joint system $qAB$.
Denote this preparation unitary by
\begin{equation}
V\coloneqq C_\zeta\bigl(I_q\otimes I_A\otimes U_\psi\bigr). \label{eq:A-search}
\end{equation}
The resulting state is
\begin{equation}
|\Phi_j\rangle \coloneqq V |0\rangle_q |\varphi_j\rangle_A |0^n\rangle_B . \label{eq:fpaa-prepared-state}
\end{equation}

The witness qubit $q$ decomposes $|\Phi_j\rangle$ into two orthogonal components according to whether we get $q=0$ or $q=1$ after measuring this qubit in the computational basis $\{|0\rangle, |1\rangle\}$.
We may therefore write the above state as follows:
\begin{equation}
|\Phi_j\rangle = \sqrt{1-\lambda_j} |\Phi_j^{(0)}\rangle + \sqrt{\lambda_j} |\Phi_j^{(1)}\rangle, \label{eq:fpaa-witness-decomposition}
\end{equation}
where $|\Phi_j^{(0)}\rangle$ is a quantum state supported on the no-witness subspace $q=0$, while $|\Phi_j^{(1)}\rangle$ is a quantum state supported on the witness subspace $q=1$.
So in the terminology of amplitude amplification, $|\Phi_j^{(1)}\rangle$ is a "good" state and $|\Phi_j^{(0)}\rangle$ is a "bad" state.
Consequently, measuring $q$ in the computational basis $\{|0\rangle, |1\rangle\}$ gives the witness outcome $q=1$ with probability $\lambda_j$.

Recall from~\eqref{eq:lambda-j} that
$\lambda_j=\sum_k \pi_k r_{jk}$.
Also recall the bounds on $\lambda_j$ derived in~\eqref{eq:lambda-separation} for two different cases: $\lambda_0\le\zeta$ for a ground-state anchor and $\lambda_j\ge\eta(1-\zeta)$ for every $j>0$.
Thus, if the anchor is excited ($j > 0$), the state $|\Phi_j\rangle$ already contains a nonzero witness component $|\Phi_j^{(1)}\rangle$ with squared amplitude at least $\eta(1-\zeta)$.
The purpose of FPAA is to increase the amplitude of this $q=1$ component, so that measuring $q$ in the computational basis $\{|0\rangle, |1\rangle\}$ gives $q=1$ with constant probability.
Importantly, the amplification procedure uses only the known lower bound $\eta(1-\zeta)$ and does not require knowledge of the actual value of $\lambda_j$.
This is why we use FPAA rather than the standard amplitude amplification.

We now describe the two generalized reflections required to construct the FPAA sequence.
Define the projector onto the initialized candidate and witness registers by
\begin{equation}
P_{\rm init} \coloneqq |0\rangle\langle0|_q\otimes I_A\otimes|0^n\rangle\langle0^n|_B. \label{eq:fpaa-initial-projector}
\end{equation}
Then the two generalized reflections are
\begin{align}
S_{\rm src}(\alpha) & \coloneqq V\bigl[I_{qAB}-(1-\e^{-\iu\alpha})P_{\rm init}\bigr]V^\dagger \label{eq:Ssrc} \\
S_{\rm wit}(\beta) & \coloneqq I_{qAB}-(1-\e^{\iu\beta}) \bigl( |1\rangle\langle1|_q \otimes I_{AB} \bigr), \label{eq:Swit}
\end{align}
The operator $S_{\rm src}(\alpha)$ applies the phase $\e^{-\iu\alpha}$ to the range of $VP_{\rm init}V^\dagger$ and acts as the identity on its orthogonal complement.
The identity on register $A$ in $P_{\rm init}$ ensures that this operation does not require knowledge of, or a reflection about, the anchor state.
For the fixed anchor state $|\varphi_j\rangle_A$, the comparator preserves this state and acts only on $qB$, so $V$ preserves the subspace in which the anchor is fixed to $|\varphi_j\rangle_A$.
Restricted to this subspace, $VP_{\rm init}V^\dagger$ is exactly $|\Phi_j\rangle\langle\Phi_j|$.
Thus $S_{\rm src}(\alpha)$ acts as the required source-state reflection within this subspace and preserves $\operatorname{span}\{|\Phi_j^{(0)}\rangle,|\Phi_j^{(1)}\rangle\}$.
The operator $S_{\rm wit}(\beta)$ applies the phase $\e^{\iu\beta}$ to the witness subspace $q=1$ and acts as the identity on the no-witness subspace $q=0$.
In particular, it applies the phase $\e^{\iu\beta}$ to the good state $|\Phi_j^{(1)}\rangle$ and leaves the bad state $|\Phi_j^{(0)}\rangle$ unchanged.
Note that implementing $S_{\rm wit}(\beta)$ is straightforward than since it is just a phase gate $P(\beta)$, as defined below, applied to the witness qubit $q$ and identity on the rest:
\begin{equation}
S_{\rm wit}(\beta) = \left(|0\rangle\langle0|_q+\e^{\iu\beta}|1\rangle\langle1|_q\right)\otimes I_{AB} \eqqcolon P(\beta)_q \otimes I_{AB}
\end{equation}

A single FPAA iterate is then
\begin{equation}
G(\alpha,\beta) \coloneqq -S_{\rm src}(\alpha)S_{\rm wit}(\beta). \label{eq:fpaa-grover-iterate}
\end{equation}
In Figure~\ref{fig:fpaa-iterate}, we show how to implement a single FPAA iterate using unitaries $C_\zeta$, $C_\zeta^\dagger$, $U_\psi$, and $U_\psi^\dagger$.
\begin{figure}[t]
    \centering
    \begin{quantikz}[row sep=0.45cm, column sep=0.32cm]
        \lstick{$q$}
        & \gate{P(\beta)}
          \gategroup[
              wires=4,
              steps=1,
              style={dashed,rounded corners,inner xsep=2pt,inner ysep=4pt},
              background,
              label style={label position=above,anchor=south,yshift=0.2cm}
          ]{{\small $S_{\rm wit}(\beta)$}}
        & \qw
        & \gate[3]{C_\zeta^\dagger}
          \gategroup[
              wires=4,
              steps=9,
              style={dashed,rounded corners,inner xsep=2pt,inner ysep=4pt},
              background,
              label style={label position=above,anchor=south,yshift=0.2cm}
          ]{{\small $S_{\rm src}(\alpha)$}}
        & \qw
        & \gate{X}
        & \ctrl{3}
        & \qw
        & \ctrl{3}
        & \gate{X}
        & \qw
        & \gate[3]{C_\zeta}
        & \qw
        \\
        \lstick{$A$}
        & \qw
        & \qw
        &
        & \qw
        & \qw
        & \qw
        & \qw
        & \qw
        & \qw
        & \qw
        &
        & \qw
        \\
        \lstick{$B$}
        & \qw
        & \qw
        &
        & \gate{U_\psi^\dagger}
        & \gate{X^{\otimes n}}
        & \ctrl{1}
        & \qw
        & \ctrl{1}
        & \gate{X^{\otimes n}}
        & \gate{U_\psi}
        &
        & \qw
        \\
        \lstick{$|0\rangle_a$}
        & \qw
        & \qw
        & \qw
        & \qw
        & \qw
        & \targ{}
        & \gate{P(-\alpha)}
        & \targ{}
        & \qw
        & \qw
        & \qw
        & \qw
    \end{quantikz}
    \caption{
    Circuit implementation of one generalized FPAA iterate
    $G(\alpha,\beta)=-S_{\rm src}(\alpha)S_{\rm wit}(\beta)$,
    up to an irrelevant global phase.
    The left dashed box implements the witness reflection
    $S_{\rm wit}(\beta)$ by applying the single-qubit phase gate $P(\beta)=|0\rangle\langle0|+\e^{\iu\beta}|1\rangle\langle1|$
    to the witness qubit $q$.
    The right dashed box implements the source reflection
    $S_{\rm src}(\alpha)$.
    The control shown on register $B$ is shorthand for controls on
    all $n$ qubits of $B$, and $X$ denotes the Pauli-$X$ gate.
    Thus, one FPAA iterate uses one application each of
    $C_\zeta$, $C_\zeta^\dagger$, $U_\psi$, and $U_\psi^\dagger$,
    together with one ancilla qubit.
    }
    \label{fig:fpaa-iterate}
\end{figure}
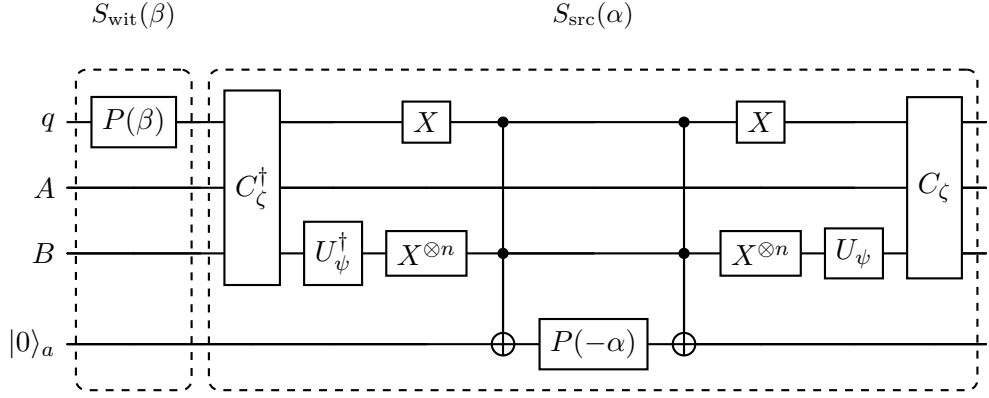

We now specify how many such iterates are used.
Following Ref.~\cite{Yoder2014}, the FPAA sequence is parameterized by an odd integer $L=2\ell+1$,
where $\ell$ is the number of FPAA iterates.
For a given $L$, the complete FPAA sequence applies
\begin{equation}
\label{eq:FPAA-sequence} G(\alpha_\ell,\beta_\ell) \cdots G(\alpha_2,\beta_2) G(\alpha_1,\beta_1)
\end{equation}
to the state $|\Phi_j\rangle$.
The phase angles $\{\alpha_m,\beta_m\}_{m=1}^{\ell}$ are chosen according to the FPAA construction of Ref.~\cite{Yoder2014}.

For our case, we choose $L$ to be the smallest odd integer satisfying
\begin{equation}
L \ge \frac{\operatorname{arcosh}2} {\sqrt{\eta(1-\zeta)}}. \label{eq:L-choice}
\end{equation}
.
With this choice, the following lemma gives the amplification guarantee.

\begin{lemma}
\label{lem:coherent-witness-amplification}
Let $L$ be the smallest odd integer satisfying~\eqref{eq:L-choice}, and let $s_j^{\rm CSD}$ denote the probability that measuring $q$ after the FPAA sequence gives $q=1$ when the anchor has energy $\mathcal E_j$.
Then
\begin{equation}
s_j^{\rm CSD} \ge \frac34, \qquad j>0. \label{eq:excited-coherent-success}
\end{equation}
If $L^2\zeta\le1/2$, then
\begin{equation}
s_0^{\rm CSD} \le 2L^2\zeta. \label{eq:coherent-ground-witness}
\end{equation}
\end{lemma}

\begin{proof}
See~\Cref{app:coherent-fixed-point-search}.
\end{proof}

For $\zeta<1/4$, the smallest odd number $L$ satisfying~\eqref{eq:L-choice} satisfies $L\le4/\sqrt\eta$.
Hence a single CSD stage uses $O(1/\sqrt{\eta})$ number of FPAA iterates, and therefore $O(1/\sqrt{\eta})$ applications of $C_\zeta$, $C_\zeta^\dagger$, $U_\psi$, and $U_\psi^\dagger$.
At the same time, the false-witness probability of a ground-state anchor grows from at most $\zeta$ to at most $O(L^2\zeta)=O(\zeta/\eta)$.
These are the coherent counterparts of the stage-level bounds in~\eqref{eq:excited-measured-success} and~\eqref{eq:delta-g-meas} for SD.

Finding a witness does not guarantee that the successful candidate state is exactly the ground state; we must also determine which candidate energies are statistically favored upon success.
Before amplification, the joint probability of candidate energy $\mathcal E_k$ and the witness outcome is $\pi_k r_{jk}$.
Therefore, conditioned on obtaining a witness, the candidate energy is distributed according to $q_j(k) = \pi_k r_{jk}/\lambda_j$,
which is the same distribution defined earlier in~\eqref{eq:qjk} for SD.
The following lemma shows that FPAA changes only the overall probability of reaching the witness subspace.
It does not change the conditional candidate-energy distribution within that subspace.

\begin{lemma}
\label{lem:coherent-direction-preserved}
Fix an anchor energy $\mathcal E_j$.
Conditioned on measuring $q=1$ after the FPAA sequence, the candidate energy $\mathcal E_k$ is distributed according to $q_j(k)$ in~\eqref{eq:qjk}.
\end{lemma}

\begin{proof}
See~\Cref{app:coherent-witness-direction}.
\end{proof}

\subsubsection{Descent Potential}
\label{sec:coherent-convergence-descent-potential}

From \Cref{lem:coherent-witness-amplification}, we have that for a fixed excited energy level $j>0$, the probability of finding a witness ($q=1$) after the FPAA sequence is $s_j^{\rm CSD}\ge3/4$.
Furthermore, from \Cref{lem:coherent-direction-preserved}, we have that conditioned on finding a witness, the candidate energy $\mathcal E_k$ is distributed according to $q_j(k)$ in~\eqref{eq:qjk}. 
Therefore the potential bound in \Cref{lem:successful-candidate-potential} that we derived for SD directly applies to CSD as well without modification:
\begin{equation}
\label{eq:successful-candidate-potential-csd} \E[ W_{\rm cand} \mid j,\mathrm{witness} ] \le \frac{9}{16}W_j + \frac{\zeta}{\eta(1-\zeta)}.
\end{equation}
We can now bound the expected potential after one complete descent stage.
Recall that, in addition to the FPAA sequence, a complete descent stage also includes the final SWAP operation that updates the state in the anchor register whenever a witness is obtained.
Let $W_{\rm new}$ denote the descent potential of the anchor state after the current descent stage. 
If the stage fails to find any witness (probability $1 - s_j^{\mathrm{CSD}}$), the anchor state is retained, yielding $W_{\mathrm{new}} = W_j$. 
If the stage successfully finds a witness (probability $s_j^{\mathrm{CSD}}$), the anchor state is replaced by the new candidate state, which satisfies the conditional expectation given in~\eqref{eq:successful-candidate-potential-csd}. 
Using the success probability bound $s_j^{\mathrm{CSD}} \ge 3/4$ from~\eqref{eq:excited-coherent-success}, the expected potential at the end of the stage evaluates to
\begin{align}
\E[W_{\rm new}\mid j] &\le (1-s_j^{\rm CSD})W_j + s_j^{\rm CSD} \left( \frac{9}{16}W_j + \frac{\zeta}{\eta(1-\zeta)} \right) \\
&\le \frac{43}{64}W_j + \frac{\zeta}{\eta(1-\zeta)}. \label{eq:coherent-excited-contraction}
\end{align}

On the other hand, for a ground-state anchor, the initial potential is $W_0=0$, and the potential can only increase if a false witness triggers an erroneous update.
By \Cref{lem:coherent-witness-amplification}, this event has probability at most $2L^2\zeta$.
Since $W_k\le1$ for every occupied energy level, we have
\begin{equation}
\E[W_{\rm new}\mid0] \le 2L^2\zeta. \label{eq:coherent-ground-error}
\end{equation}

The bounds~\eqref{eq:coherent-excited-contraction} and~\eqref{eq:coherent-ground-error} can be extended easily to an arbitrary anchor state. 
In \Cref{app:coherent-potential-recursion}, we rigourously derive the following recurrence relation for the mixed-state descent potential $V_t$ (see~\eqref{eq:Vt} for its definition):
\begin{equation}
V_{t+1} \le \frac{43}{64}V_t + 2L^2\zeta + \frac{\zeta}{\eta(1-\zeta)}. \label{eq:V-recursion-coherent}
\end{equation}

We use the same comparator error $\zeta$ as in SD, namely
\begin{equation}
\zeta = \frac{\eta^2\xi}{256}. \label{eq:zeta-coherent}
\end{equation}
Since $\zeta<1/4$ and $L\le4/\sqrt\eta$, we have
\begin{equation}
L^2\zeta \le \frac{16}{\eta} \cdot \frac{\eta^2\xi}{256} = \frac{\eta\xi}{16} \le \frac{1}{16}.
\end{equation}
In particular, the assumption $L^2\zeta\le1/2$ in \Cref{lem:coherent-witness-amplification} is satisfied.
The second and third terms on the right-hand side of~\eqref{eq:V-recursion-coherent} are bounded from above as follows:
\begin{equation}
2L^2\zeta \le \frac{\eta\xi}{8}, \qquad \frac{\zeta}{\eta(1-\zeta)} \le \frac{\eta\xi}{192}, \label{eq:coherent-additive-errors}
\end{equation}

The anchor register is initialized in the state $|\psi\rangle$. Because $W_j \le 1$ for all levels, the initial mixed-state potential is trivially bounded by $V_0 \le 1$. 
Iterating the recurrence relation in~\eqref{eq:V-recursion-coherent} over $R$ consecutive descent stages yields
\begin{align}
V_R &\le \left(\frac{43}{64}\right)^R V_0 + \sum_{t=0}^{R-1} \left(\frac{43}{64}\right)^t \left( 2L^2\zeta + \frac{\zeta}{\eta(1-\zeta)} \right) \nonumber \\
&\le \left(\frac{43}{64}\right)^R + \frac{1}{ 1-\frac{43}{64} } \left( 2L^2\zeta + \frac{\zeta}{\eta(1-\zeta)} \right) \nonumber \\
&= \left(\frac{43}{64}\right)^R + \frac{64}{21} \left( 2L^2\zeta + \frac{\zeta}{\eta(1-\zeta)} \right). \label{eq:VR-bound-coherent}
\end{align}
Using the bounds in~\eqref{eq:coherent-additive-errors}, the second term on the right-hand side of the above inequality is at most
\begin{equation}
\label{eq:bound-on-second-third-terms-of-VR} \frac{64}{21} \left( \frac18+\frac1{192} \right) \eta\xi = \frac{25}{63}\eta\xi < \frac{\eta\xi}{2}.
\end{equation}

We therefore use the same number of descent stages as SD,
\begin{equation}
\label{eq:R-coherent} R = \left\lceil \frac{\log(2/(\eta\xi))} {\log(64/43)} \right\rceil.
\end{equation}
This ensures that $(43/64)^R\le\eta\xi/2$, so~\eqref{eq:bound-on-second-third-terms-of-VR} gives $V_R<\eta\xi$.

Finally, we map this bound on the potential directly to the target ground-state fidelity.
For every excited level $j>0$, recall that we have $W_j\ge\eta$ from~\eqref{eq:W-excited-lower}. 
Substituting this into the definition of $V_R$ gives:
\begin{equation}
V_R = \sum_{j>0}W_j\Tr(\Pi_j\rho_R) \ge \eta\sum_{j>0}\Tr(\Pi_j\rho_R).
\end{equation}
Together with the inequality $V_R<\eta\xi$, this implies $\sum_{j>0}\Tr(\Pi_j\rho_R)<\xi$. 
Since the ground state is nondegenerate,
\begin{equation}
\langle \phi_0|\rho_R|\phi_0\rangle \ge 1-\xi. \label{eq:coherent-final-fidelity}
\end{equation}
This concludes the convergence analysis for \Cref{alg:coherent-descent}.

We summarize the result below.

\subsection{Main Result}
\label{sec:coherent-resources}

\begin{theorem}[Ground-State Preparation via Coherent Spectral Descent]
\label{thm:coherent-descent}
Consider an instance of the GSP problem (Problem~\ref{prob:gsp}), and 
suppose that we have coherent access to $U_\psi$ and $U_\psi^\dagger$, where $U_\psi|0^n\rangle=|\psi\rangle$.
Choose the elementary evolution time $\tau$ as in~\eqref{eq:tau} and define $U_H\coloneqq\e^{-\iu H\tau}$.
Then the CSD algorithm (Algorithm~\ref{alg:coherent-descent}), with the parameter choices in~\eqref{eq:zeta-coherent}, \eqref{eq:L-choice}, and~\eqref{eq:R-coherent}, outputs a state $\rho_R$ satisfying
\begin{equation}
\langle\phi_0|\rho_R|\phi_0\rangle \ge 1-\xi.
\end{equation}
The CSD algorithm uses
\begin{enumerate}[label=(\roman*),leftmargin=*]
\item
    $T_H^{\rm CSD} = O\!\left( \frac1{\gamma\sqrt\eta} \log^2\frac1{\eta\xi} \right) = \widetilde O\!\left( \frac1{\gamma\sqrt\eta} \right)$
    total Hamiltonian-evolution time;

\item
    $N_H^{\rm CSD} = O\!\left( \frac{\alpha_H}{\gamma\sqrt\eta} \log^2\frac1{\eta\xi} \right) = \widetilde O\!\left( \frac{\alpha_H}{\gamma\sqrt\eta} \right)$
    queries to the uncontrolled oracles $U_H$ and $U_H^\dagger$; and

\item
    $N_{U_\psi}^{\rm CSD} = O\!\left( \frac1{\sqrt\eta} \log\frac{2}{\eta\xi} \right) = \widetilde O\!\left( \frac1{\sqrt\eta} \right)$
    queries to $U_\psi$ and $U_\psi^\dagger$.
\end{enumerate}
Crucially, it does not require access to $\mathrm{ctrl}\text{-}U_H$.
\end{theorem}

\begin{proof}
The choices in~\eqref{eq:L-choice} and~\eqref{eq:R-coherent} give $L=O(1/\sqrt{\eta})$ and $R=O(\log(1/(\eta\xi)))$, respectively.
In a single descent stage, the FPAA sequence, as defined in~\eqref{eq:FPAA-sequence}, uses $O(1/\sqrt{\eta})$ number of FPAA iterates.
Each FPAA iterate makes one query to each of $C_\zeta$, $C_\zeta^\dagger$, $U_\psi$, and $U_\psi^\dagger$. 
Therefore, one descent stage makes $O(1/\sqrt{\eta})$ queries to each of these four unitaries.

The CSD algorithm (Algorithm~\ref{alg:coherent-descent}) consists of at most $R$ descent stages. Hence, the total number of queries to each of $C_\zeta$, $C_\zeta^\dagger$, $U_\psi$, and $U_\psi^\dagger$ is at most $LR = O(1/\sqrt{\eta} \log (1/(\eta\xi)) )$.

We first bound the total Hamiltonian-evolution time.
By \Cref{thm:qsp-filter}, one query to either $C_\zeta$ or $C_\zeta^\dagger$ requires Hamiltonian-evolution time $O(1/\gamma \log 1/(\eta\xi)).$
Since we have $O(1/\sqrt{\eta} \log (1/(\eta\xi)) )$ queries to each of $C_\zeta$ and $C_\zeta^\dagger$ in total as shown above, we get the total Hamiltonian-evolution time
\begin{equation}
T_H^{\rm CSD} = O\!\left( \frac{1}{\gamma\sqrt{\eta}} \log^2\frac{1}{\eta\xi} \right) = \widetilde O\!\left( \frac{1}{\gamma\sqrt{\eta}} \right). \label{eq:coherent-H-time}
\end{equation}
This proves item~$(i)$.
Note that the dependence on $\gamma$ is the same as in SD, but the dependence on $\eta$ improves from $1/\eta$ to $1/\sqrt{\eta}$.
Thus the scaling in~\eqref{eq:coherent-H-time} attains the known optimal joint gap-overlap dependence up to logarithmic factors, using only uncontrolled Hamiltonian evolution.

We next count the number of queries to the uncontrolled oracles $U_H$ and $U_H^\dagger$. 
From~\eqref{eq:qsp-degree}, one query to $C_\zeta$ uses $d=O((\alpha_H/\gamma+1)\log(1/(\eta\xi)))$ queries to $\SWAPECHO_\gamma(\tau)$, and hence the same number, up to a constant
factor, of queries to $U_H$ and $U_H^\dagger$. 
Multiplying this cost
by the number of queries to $C_\zeta$ and $C_\zeta^\dagger$ that we established above, we get
\begin{equation}
N_H^{\rm CSD} = O\!\left( \frac{\alpha_H}{\sqrt{\eta} \gamma} \log^2\frac1{\eta\xi} \right) = \tilde O\!\left( \frac{\alpha_H}{\sqrt{\eta} \gamma} \right)
\end{equation}
number of queries to each of $U_H$ and $U_H^\dagger$.
This proves item~$(ii)$.

Finally, as established above, the total number of queries to $U_\psi$ and $U_\psi^\dagger$ is $O(1/\sqrt{\eta} \log (1/(\eta\xi)) )$.
Including the single initial state preparation of the anchor register gives
\begin{equation}
N_{U_\psi}^{\rm CSD} = O\!\left( \frac1{\sqrt\eta} \log\frac{2}{\eta\xi} \right) = \widetilde O\!\left( \frac1{\sqrt\eta} \right)
\end{equation}
number of queries to each of $U_\psi$ and $U_\psi^\dagger$.
This proves item~$(iii)$ and also concludes the proof of the theorem.
\end{proof}

\section{Ground-State Property and Energy Estimation}
\label{sec:ground-state-estimation}

We now use SD and CSD to solve GSPE and GSEE problems defined formally in \Cref{sec:ground-state-tasks}.
For GSPE, we first prepare an approximate ground state and then estimate the desired expectation value by measuring independently prepared outputs.
We must account for both the error in the prepared state and the statistical error in the measurements.
For GSEE, we can relax the preparation requirement: it is enough to concentrate the output on energies close to $E_0$.
We will show that replacing the comparison scale $\gamma$ by a scale proportional to the target energy precision gives this weaker guarantee without a gap promise.

\subsection{Expectation Values from an Approximate Ground State}
\label{sec:approximate-ground-expectation}

Let $P_0\coloneqq|\phi_0\rangle\langle\phi_0|$ and $Q_0\coloneqq I-P_0$.
Suppose that $\rho$ is an output state of SD or CSD satisfying $\Tr(P_0\rho)\ge1-\xi$, and let $O$ be an operator with $\|O\|\le1$.
Our goal is to estimate $\mu\coloneqq\langle\phi_0|O|\phi_0\rangle$, whereas a measurement on $\rho$ gives information about $\Tr(O\rho)$.
The following lemma bounds the difference between these two quantities.
It applies to both Hermitian observables and unitary operators, where in the latter case, $\mu$ may be complex.

\begin{lemma}[Expectation-value bias]
\label{lem:ground-property-bias}
Let $\rho$ satisfy $\langle\phi_0|\rho|\phi_0\rangle\ge1-\xi$, and let $O$ satisfy $\|O\|\le1$.
Then
\begin{equation}
\left|\Tr(O\rho)-\langle\phi_0|O|\phi_0\rangle\right| \le 2\sqrt{\xi}. \label{eq:general-property-bias}
\end{equation}
If, in addition, $[O,H]=0$, then
\begin{equation}
\left|\Tr(O\rho)-\langle\phi_0|O|\phi_0\rangle\right| \le 2\xi. \label{eq:commuting-property-bias}
\end{equation}
\end{lemma}

\begin{proof}
By the Fuchs-van de Graaf inequality and the assumption $\Tr(P_0\rho)\ge 1-\xi$,
\begin{equation}
\left \|\rho-P_0 \right\|_1 \le 2\sqrt{1-\Tr(P_0\rho)} \le 2\sqrt{\xi}.
\end{equation}
Therefore, by Hölder's inequality,
\begin{equation}
\left |\Tr(O\rho)-\mu \right| = \left|\Tr[O(\rho-P_0)]\right| \le \left\|O\right\| \left\|\rho-P_0\right \|_1 \le 2\sqrt{\xi},
\end{equation}
which proves~\eqref{eq:general-property-bias}.

Now suppose that $[O,H]=0$.
For this case, we have $P_0OQ_0=Q_0OP_0=0$ and $P_0OP_0=\mu P_0$.
Writing $r\coloneqq\Tr(Q_0\rho)\le\xi$, we obtain $\Tr(O\rho)-\mu=-\mu r+\Tr(OQ_0\rho Q_0)$.
Since $Q_0\rho Q_0$ is positive semidefinite with trace $r$, we have $\left |\Tr(OQ_0\rho Q_0)\right |\le\|O\|r\le r$.
Together with $|\mu|\le1$, this gives $\left |\Tr(O\rho)-\mu \right |\le2r\le2\xi$.
\end{proof}

To allocate at most $\varepsilon/2$ of the estimation error to state preparation, we therefore choose
\begin{equation}
\xi_O\coloneqq \begin{cases} \varepsilon/4, & [O,H]=0,\\[1mm] \varepsilon^2/16, & \text{otherwise}. \end{cases} \label{eq:property-xi-choice}
\end{equation}
Indeed, the first choice gives $2\xi_O=\varepsilon/2$, while the second gives $2\sqrt{\xi_O}=\varepsilon/2$.

\subsection{Ground-State Property Estimation}
\label{sec:gspe-from-gsp}

We next specify how $O$ is measured.
This access assumption is separate from access to the Hamiltonian dynamics and is needed to state the resource bounds precisely.
Throughout this subsection, let $\varepsilon\in(0,1)$ be the target additive error and let $\delta\in(0,1)$ be the allowed failure probability.

\subsubsection{Hermitian observables}
\label{sec:gspe-hermitian}

Suppose that $O=O^\dagger$ and $\|O\|\le1$.
A projective measurement of $O$ on $\rho$ returns an eigenvalue $X_O\in[-1,1]$ with mean $\E[X_O]=\Tr(O\rho)$.
The same argument applies to any available measurement procedure that returns an unbiased estimator in $[-1,1]$.
We assume access to such a procedure, and count its uses separately from the preparation resources.

Prepare $N$ independent outputs of SD or CSD, each with target infidelity $\xi_O$, and let $\widehat\mu$ be the average of their measurement outcomes.
The standard Hoeffding's inequality gives $\Pr[|\widehat\mu-\Tr(O\rho)|>\varepsilon/2]\le2\exp(-N\varepsilon^2/8)$.
Thus $N=\lceil8\varepsilon^{-2}\log(2/\delta)\rceil$ is sufficient to bound the statistical error by $\varepsilon/2$ with probability at least $1-\delta$.
On this event, \Cref{lem:ground-property-bias} and the triangle inequality give $|\widehat\mu-\mu|\le|\widehat\mu-\Tr(O\rho)|+|\Tr(O\rho)-\mu|\le\varepsilon$.

\subsubsection{Unitary operators}
\label{sec:gspe-unitary}

Suppose instead that $O$ is unitary and controlled-$O$ is available.
We estimate $\Tr(O\rho)$ using the standard Hadamard test.
Applying the Hadamard test with an $X$- or $Y$-basis measurement on the readout qubit gives a random variable in $\{-1,1\}$ whose expectation is $\operatorname{Re}\Tr(O\rho)$ or $\operatorname{Im}\Tr(O\rho)$, respectively.

Use independent preparations for the two components, with $N_0=\lceil16\varepsilon^{-2}\log(4/\delta)\rceil$ samples for each component.
Hoeffding's inequality then bounds each component error by $\varepsilon/(2\sqrt2)$ except with probability $\delta/2$.
By a union bound, the resulting complex estimate differs from $\Tr(O\rho)$ in modulus by at most $\sqrt{2(\varepsilon/(2\sqrt2))^2}=\varepsilon/2$ with probability at least $1-\delta$.
Since $\|O\|=1$, the preparation bias is again at most $\varepsilon/2$ under~\eqref{eq:property-xi-choice}, so the final error is at most $\varepsilon$.
If $O$ is also Hermitian, its expectation value is real and only the $X$-basis measurement is needed.

\subsubsection{Main results}
\label{sec:gspe-resource-bounds}

The following theorems state the resulting bounds in the two input-state access models.

\begin{theorem}[Ground-State Property Estimation via Spectral Descent]
\label{thm:gspe-sd}
Consider an instance of the GSPE problem (Problem~\ref{prob:gspe}), and suppose that fresh copies of the input state $|\psi\rangle$ can be prepared.
Suppose further that $O$ is either a Hermitian observable with measurement access as in \Cref{sec:gspe-hermitian}, or unitary with controlled-$O$ access as in \Cref{sec:gspe-unitary}.
Choose $\xi_O$ as in~\eqref{eq:property-xi-choice}, choose $\tau$ as in~\eqref{eq:tau}, and define $U_H\coloneqq\e^{-\iu H\tau}$.
Then independent applications of the SD algorithm with target infidelity $\xi_O$, followed by the corresponding measurements of $O$, produce an estimate $\widehat\mu$ satisfying $|\widehat\mu-\langle\phi_0|O|\phi_0\rangle|\le\varepsilon$ with probability at least $1-\delta$.
The procedure uses
\begin{enumerate}[label=(\roman*),leftmargin=*]
    \item $T_H^{\rm SD\text{-}GSPE}=O\!\left(\frac{1}{\eta\gamma\varepsilon^2}\log\frac{2}{\delta}\,\log^2\frac{2}{\eta\xi_O}\right)$ total Hamiltonian-evolution time;
    \item $N_H^{\rm SD\text{-}GSPE}=O\!\left(\frac{\alpha_H}{\eta\gamma\varepsilon^2}\log\frac{2}{\delta}\,\log^2\frac{2}{\eta\xi_O}\right)$ queries to the uncontrolled oracles $U_H$ and $U_H^\dagger$; and
    \item $N_\psi^{\rm SD\text{-}GSPE}=O\!\left(\frac{1}{\eta\varepsilon^2}\log\frac{2}{\delta}\,\log\frac{2}{\eta\xi_O}\right)$ fresh copies of $|\psi\rangle$.
\end{enumerate}
In addition, the procedure uses $O(\varepsilon^{-2}\log(2/\delta))$ measurements of the specified type.
Crucially, it does not require access to $\mathrm{ctrl}\text{-}U_H$.
\end{theorem}

\begin{proof}
By \Cref{thm:measured-gsp}, one SD run with target infidelity $\xi_O$ produces a state $\rho$ satisfying $\langle\phi_0|\rho|\phi_0\rangle\ge1-\xi_O$.
The measurement analysis above then proves the claimed accuracy using $N=O(\varepsilon^{-2}\log(2/\delta))$ independent runs.
Each run uses $O((\eta\gamma)^{-1}\log^2(2/(\eta\xi_O)))$ evolution time, $O(\alpha_H(\eta\gamma)^{-1}\log^2(2/(\eta\xi_O)))$ queries to $U_H$ and $U_H^\dagger$, and $O(\eta^{-1}\log(2/(\eta\xi_O)))$ fresh input copies.
Multiplying these three bounds by $N$ proves items~$(i)$--$(iii)$.
\end{proof}

\begin{theorem}[Ground-State Property Estimation via Coherent Spectral Descent]
\label{thm:gspe-csd}
Consider an instance of the GSPE problem (Problem~\ref{prob:gspe}), and suppose that we have access to $U_\psi$ and $U_\psi^\dagger$, where $U_\psi|0^n\rangle=|\psi\rangle$.
Let $O$ have either of the two measurement access types specified in \Cref{thm:gspe-sd}.
Choose $\xi_O$ as in~\eqref{eq:property-xi-choice}, choose $\tau$ as in~\eqref{eq:tau}, and define $U_H\coloneqq\e^{-\iu H\tau}$.
Then independent applications of the CSD algorithm with target infidelity $\xi_O$, followed by the corresponding measurements of $O$, produce an estimate $\widehat\mu$ satisfying $|\widehat\mu-\langle\phi_0|O|\phi_0\rangle|\le\varepsilon$ with probability at least $1-\delta$.
The procedure uses
\begin{enumerate}[label=(\roman*),leftmargin=*]
    \item $T_H^{\rm CSD\text{-}GSPE}=O\!\left(\frac{1}{\gamma\sqrt\eta\,\varepsilon^2}\log\frac{2}{\delta}\,\log^2\frac{2}{\eta\xi_O}\right)$ total Hamiltonian-evolution time;
    \item $N_H^{\rm CSD\text{-}GSPE}=O\!\left(\frac{\alpha_H}{\gamma\sqrt\eta\,\varepsilon^2}\log\frac{2}{\delta}\,\log^2\frac{2}{\eta\xi_O}\right)$ queries to the uncontrolled oracles $U_H$ and $U_H^\dagger$; and
    \item $N_{U_\psi}^{\rm CSD\text{-}GSPE}=O\!\left(\frac{1}{\sqrt\eta\,\varepsilon^2}\log\frac{2}{\delta}\,\log\frac{2}{\eta\xi_O}\right)$ queries to $U_\psi$ and $U_\psi^\dagger$.
\end{enumerate}
In addition, the procedure uses $O(\varepsilon^{-2}\log(2/\delta))$ measurements of the specified type.
Crucially, it does not require access to $\mathrm{ctrl}\text{-}U_H$.
\end{theorem}

\begin{proof}
The fidelity guarantee in \Cref{thm:coherent-descent} gives the same preparation bias as in \Cref{thm:gspe-sd}, and the measurement procedure is unchanged.
Thus $N=O(\varepsilon^{-2}\log(2/\delta))$ independent preparations again suffice.
One CSD run uses $O((\gamma\sqrt\eta)^{-1}\log^2(2/(\eta\xi_O)))$ evolution time, $O(\alpha_H(\gamma\sqrt\eta)^{-1}\log^2(2/(\eta\xi_O)))$ Hamiltonian queries, and $O(\eta^{-1/2}\log(2/(\eta\xi_O)))$ input-state oracle queries.
Multiplying each cost by $N$ proves the stated bounds.
\end{proof}

The same bias analysis also applies when a Hermitian observable is given through an $\alpha_O$-block encoding $U_O$, with $\alpha_O\ge1$ and controlled-$U_O$ available.
A Hadamard test on $|0\rangle\langle0|\otimes\rho$, where $|0\rangle$ initializes the block-encoding ancillas, then estimates $\Tr(O\rho)/\alpha_O$.
Rescaling the outcome by $\alpha_O$ gives an unbiased estimator in $[-\alpha_O,\alpha_O]$.
The required sample count is then $O(\alpha_O^2\varepsilon^{-2}\log(2/\delta))$.
This factor multiplies both the measurement count and the total preparation resources, while the choice of $\xi_O$ remains unchanged.

\paragraph{Going from $1/\varepsilon^2$ to $1/\varepsilon$ with coherent observable access.}
The above GSPE algorithms prepare a new state for each measurement,
which leads to the $1/\varepsilon^2$ dependence.
For CSD, amplitude estimation primitive can improve this to $1/\varepsilon$, up to
logarithmic factors, under the following additional access assumption.

For a Hermitian observable $O$ with $\|O\|\le1$, assume access to a unitary $V_O$ and its inverse.
Starting with the system register in any state $\sigma$ and the ancillary register in $|0\cdots0\rangle$, we apply $V_O$ jointly to both registers.
Measuring one particular ancilla qubit then gives outcome $1$ with probability
\begin{equation}
p_\sigma=\frac{1+\Tr(O\sigma)}{2}.
\end{equation}
Thus, estimating $p_\sigma$ also estimates $\Tr(O\sigma)=2p_\sigma-1$.
We will use this circuit without performing its final measurement.

We must also make the entire CSD preparation reversible.
Instead of measuring each witness qubit, we use it directly to control the SWAP, and instead of resetting or discarding registers, we retain them and use fresh registers for the next stage.
These changes give a unitary circuit $W$ whose output on the anchor register, after tracing out the retained registers, is the same state $\rho_R$ as in CSD.
Furthermore, reversing the circuit implements $W^\dagger$.
Note that both these circuits use only $U_H$ and $U_H^\dagger$ for Hamiltonian access.

Now choose the CSD infidelity to be $\xi=\varepsilon^2/16$.
The preparation bias then satisfies
\begin{equation}
\left|\Tr(O\rho_R)-\langle\phi_0|O|\phi_0\rangle\right| \le 2\sqrt{\xi}=\frac{\varepsilon}{2}.
\end{equation}
Now apply $W$, followed by $V_O$ on the anchor and fresh ancillas.
The readout probability is $p=(1+\Tr(O\rho_R))/2$.
Iterative amplitude estimation uses this combined circuit and its inverse to obtain an estimate $\widehat p$ satisfying $|\widehat p-p|\le\varepsilon/4$ with probability at least $1-\delta$, using $\widetilde O(1/\varepsilon)$ circuit applications~\cite{Grinko2021}.
Consequently, $\widehat\mu=2\widehat p-1$ satisfies
\begin{equation}
|\widehat\mu-\langle\phi_0|O|\phi_0\rangle| \le 2|\widehat p-p|+2\sqrt{\xi}\le\varepsilon.
\end{equation}
Multiplying the cost of one CSD run by the number of circuit applications gives
\begin{equation}
N_H^{\rm CSD\text{-}GSPE}=\widetilde O\!\left( \frac{\alpha_H}{\gamma\sqrt{\eta}\,\varepsilon} \right), \qquad T_H^{\rm CSD\text{-}GSPE}=\widetilde O\!\left( \frac{1}{\gamma\sqrt{\eta}\,\varepsilon} \right).
\end{equation}
The procedure also uses $\widetilde O(1/(\sqrt{\eta}\,\varepsilon))$ queries to
$U_\psi$ and $U_\psi^\dagger$ and $\widetilde O(1/\varepsilon)$ queries to $V_O$ and $V_O^\dagger$.
The suppressed factors include logarithmic dependence on $1/\delta$.

Iterative amplitude estimation does not require controlled applications of the combined circuit.
Hence this improvement introduces no controlled Hamiltonian evolution.
It does require keeping the additional registers and being able to reverse the preparation, which fresh-copy assumption of SD alone does not provide.

\subsection{Ground-State Energy Estimation}
\label{sec:gsee-from-gsp}

For GSEE, an output supported near the ground-state energy is sufficient, even if it has little overlap with the ground state itself.
We therefore replace $\gamma$ in the comparator construction by an energy resolution $h>0$.
The resulting comparator rejects candidates whose energy is at least the anchor energy and accepts candidates whose energy is lower by at least $h$, with the same error $\zeta$ as before.
No spectral-gap assumption is used in constructing this comparator: $h$ specifies its transition width.
The SD and CSD circuits, including their numbers of trials or FPAA iterates per stage, can otherwise be used unchanged.

There is a separate access requirement for reporting an absolute energy.
Replacing $H$ by $H+cI$ multiplies every uncontrolled evolution by a global phase and leaves every SWAP echo unchanged.
Thus these operations alone cannot determine the additive energy offset.
As in the property-estimation results, we must specify a measurement procedure that supplies this information.
We first prove the gap-free preparation guarantee, then give GSEE bounds under the bounded energy-measurement assumption.

\subsubsection{Descent into a low-energy interval}

For $h>0$, let $P_{\le E_0+h}\coloneqq\sum_{j:\mathcal E_j\le E_0+h}\Pi_j$ denote the projector onto the occupied energy levels at most $h$ above $E_0$.
This projector is used only in the analysis; neither algorithm needs to know $E_0$.

\begin{lemma}[Low-energy preparation without a gap promise]
\label{lem:gap-free-descent}
Let $H$ satisfy $\|H\|\le\alpha_H$ and the overlap promise, and let $h\in(0,2\alpha_H]$ and $\xi\in(0,1)$.
Run SD or CSD with $\gamma$ replaced by $h$ in the comparator, with $\zeta=\eta^2\xi/256$, and with the same choices of $M$, $L$, and $R$ as in~\eqref{eq:M-choice}, \eqref{eq:L-choice}, and~\eqref{eq:R-measured}, respectively, using $M$ for SD and $L$ for CSD.
Then the output satisfies
\begin{equation}
\Tr(P_{\le E_0+h}\rho_R)\ge1-\xi. \label{eq:gap-free-low-energy}
\end{equation}
The resource bounds in \Cref{thm:measured-gsp,thm:coherent-descent} hold with $\gamma$ replaced by $h$.
No lower bound on the spectral gap is required.
\end{lemma}

\begin{proof}
Choose $\tau=\pi/[4(2\alpha_H+h/2)]$ and replace the shift $-\gamma/2$ by $-h/2$ in the SWAP echo.
The comparator construction in \Cref{sec:qsp-comparator} then gives $r_{jk}\le\zeta$ when $\mathcal E_k\ge\mathcal E_j$, $r_{jk}\ge1-\zeta$ when $\mathcal E_j-\mathcal E_k\ge h$, and the same monotonic tracking guarantee as before.

The original potential $W_j$ distinguishes the ground state from every excited energy.
Here we instead assign zero potential to the entire desired interval, by defining
\begin{equation}
W_j^{(h)}\coloneqq \begin{cases} 0, & \mathcal E_j\le E_0+h,\\ W_j, & \mathcal E_j>E_0+h, \end{cases} \qquad V_t^{(h)}\coloneqq\sum_jW_j^{(h)}\Tr(\Pi_j\rho_t). \label{eq:band-potential}
\end{equation}
Since $W_j\ge\eta$ for every $j>0$, we have $V_t^{(h)}\ge\eta\Tr[(I-P_{\le E_0+h})\rho_t]$ on the occupied support.
We will show that this potential obeys the same recurrence as in the ground-state analysis.

First fix an anchor energy $\mathcal E_j>E_0+h$.
The ground-state candidate is lower by more than $h$, so $\lambda_j\ge\pi_0r_{j0}\ge\eta(1-\zeta)$.
Consequently, a descent stage finds a witness with probability at least $3/4$ for both SD and CSD.
The proof of \Cref{lem:successful-candidate-potential} still applies to this anchor: its only use of the ground-state gap was to obtain the lower bound $\sum_{k<j}\pi_kr_{jk}\ge\eta(1-\zeta)$, which we have just verified.
The ordering of the ideal comparator and its approximation error are unchanged.
Since $W_k^{(h)}\le W_k$ for every candidate and $W_j^{(h)}=W_j$, the conditional candidate bound gives
\begin{equation}
\E[W_{\rm new}^{(h)}\mid j] \le\frac{43}{64}W_j^{(h)} +\frac{\zeta}{\eta(1-\zeta)}, \qquad \mathcal E_j>E_0+h. \label{eq:band-excited-contraction}
\end{equation}
Here, as before, both repetition and FPAA preserve the conditional candidate distribution $q_j(k)=\pi_kr_{jk}/\lambda_j$.

Now fix $\mathcal E_j\le E_0+h$.
Any candidate outside the desired interval has $\mathcal E_k>E_0+h\ge\mathcal E_j$, so its witness probability is at most $\zeta$.
The total one-trial witness weight on such candidates is therefore $\sum_{k:\mathcal E_k>E_0+h}\pi_kr_{jk}\le\zeta$.
For SD, a union bound over at most $M$ trials bounds the probability of leaving the interval by $M\zeta$.
This argument allows witnesses that move the anchor to other energies inside the interval; those outcomes still have zero potential.

For CSD, we need to control amplification of the witness weight outside the interval even when $\lambda_j$ itself is large.
The FPAA success probability satisfies $P_L(\lambda)\le2L^2\lambda$ for every $\lambda\in[0,1]$.
Indeed, \eqref{eq:CSD-fixed-point-low-side} proves this when $L^2\lambda\le1/2$, while for $L^2\lambda>1/2$ it follows from $P_L(\lambda)\le1<2L^2\lambda$.
When $\lambda_j>0$, preservation of $q_j(k)$ therefore gives
\begin{equation}
\Pr[\text{leave the interval}\mid j] =\frac{P_L(\lambda_j)}{\lambda_j} \sum_{k:\mathcal E_k>E_0+h}\pi_kr_{jk} \le2L^2\zeta. \label{eq:band-coherent-leakage}
\end{equation}
If $\lambda_j=0$, no witness is produced and the same bound holds.
Since every potential is at most one, the expected potential after a stage starting inside the interval is at most $M\zeta$ for SD and $2L^2\zeta$ for CSD.

These fixed-energy bounds also apply to an arbitrary anchor state.
Before the conditional SWAP, the comparison and coherent-search operations preserve the anchor-energy sectors.
On a no-witness branch, an off-diagonal anchor-energy block contributes zero to every output energy population; on a witness branch, it vanishes when the old anchor is discarded.
Thus the population argument in \Cref{app:measured-potential-recursion,app:coherent-potential-recursion} applies with $W_j^{(h)}$ in place of $W_j$.
Averaging the fixed-energy bounds yields
\begin{equation}
V_{t+1}^{(h)}\le\frac{43}{64}V_t^{(h)} +\frac{\zeta}{\eta(1-\zeta)} +\begin{cases} M\zeta, & \mathrm{SD},\\ 2L^2\zeta, & \mathrm{CSD}. \end{cases} \label{eq:band-potential-recursion}
\end{equation}
The parameters satisfy $M<3/\eta$, $L\le4/\sqrt\eta$, and $1-\zeta\ge3/4$.
Also, $\zeta=\eta^2\xi/256$ satisfies the small-error condition in \Cref{lem:successful-candidate-potential}.
Since $V_0^{(h)}\le1$, iterating~\eqref{eq:band-potential-recursion} gives, for either algorithm, $V_R^{(h)}\le(43/64)^R+(25/63)\eta\xi$.
The stated choice of $R$ makes the first term at most $\eta\xi/2$, so $V_R^{(h)}\le(113/126)\eta\xi<\eta\xi$.
Dividing by $\eta$ proves~\eqref{eq:gap-free-low-energy}.

Finally, one comparator uses $O(h^{-1}\log(1/\zeta))$ evolution time and $O((\alpha_H/h+1)\log(1/\zeta))$ elementary evolution queries.
The numbers of stages, trials, and FPAA iterates have not changed, and $h\le2\alpha_H$ allows the additive one in the query bound to be absorbed.
The resource counts therefore follow exactly as in \Cref{thm:measured-gsp,thm:coherent-descent}, with $h$ in place of $\gamma$.
\end{proof}

\subsubsection{Estimating the energy from bounded measurement outcomes}

We first retain the measurement model in which one output state supplies a real random variable $X_H$ satisfying $\E[X_H]=\Tr(H\rho)$ and $|X_H|\le\beta_H$.
We may take the known bound $\beta_H\ge\alpha_H$ by enlarging it if necessary.
This model includes, for example, randomized measurements of terms in a known decomposition of $H$.
If $H=\sum_a h_aP_a$ with Hermitian unitaries $P_a$ and real coefficients $h_a$, selecting $a$ with probability $|h_a|/\sum_b|h_b|$, measuring $P_a$, and multiplying the outcome by $\operatorname{sgn}(h_a)\sum_b|h_b|$ gives such an estimator with $\beta_H=\sum_a|h_a|$.

Recall that both algorithms preserve the occupied spectral support of $|\psi\rangle$.
For SD this follows from the comparator action and the use of fresh candidates.
For CSD, the normalized witness component before and after FPAA is the same, while a no-witness outcome retains the original anchor-energy sector, so the same support statement holds.
Define the occupied spectral width by
\begin{equation}
D_H\coloneqq\max_{j:c_j\ne0}(E_j-E_0)\le2\alpha_H. \label{eq:Dham}
\end{equation}
For an output satisfying~\eqref{eq:gap-free-low-energy}, splitting the energy expectation into levels inside and outside the interval gives
\begin{align}
0\le\Tr(H\rho_R)-E_0 &=\sum_j(\mathcal E_j-E_0)\Tr(\Pi_j\rho_R) \nonumber \\
&\le h\Tr(P_{\le E_0+h}\rho_R) +D_H\Tr[(I-P_{\le E_0+h})\rho_R] \le h+2\alpha_H\xi. \label{eq:gsee-bias}
\end{align}
Thus, for $0<\varepsilon<\alpha_H$, choosing $h=\varepsilon/4$ and $\xi=\varepsilon/(8\alpha_H)$ bounds the preparation bias by $\varepsilon/2$.
These choices depend only on known parameters, not on $E_0$ or $D_H$.
Hoeffding's inequality then shows that $N_E=\lceil8\beta_H^2\varepsilon^{-2}\log(2/\delta)\rceil$ independent preparations and measurements estimate $\Tr(H\rho_R)$ to error at most $\varepsilon/2$, except with probability $\delta$.

\begin{theorem}[Ground-State Energy Estimation via Spectral Descent]
\label{thm:gsee-spectral-descent}
Consider an instance of the GSEE problem (Problem~\ref{prob:gsee}) with $0<\varepsilon<\alpha_H$, and suppose that fresh copies of the input state $|\psi\rangle$ can be prepared.
Suppose further that the bounded energy measurement described above is available with $\beta_H\ge\alpha_H$.
Choose $h=\varepsilon/4$, $\xi=\varepsilon/(8\alpha_H)$, $\tau=\pi/[4(2\alpha_H+h/2)]$, and $U_H\coloneqq\e^{-\iu H\tau}$.
Then independent applications of the SD algorithm with the parameters in \Cref{lem:gap-free-descent}, followed by the energy measurements, produce an estimate $\widehat E_0$ satisfying
$|\widehat E_0-E_0|\le\varepsilon$ with probability at least $1-\delta$.
The procedure uses
\begin{enumerate}[label=(\roman*),leftmargin=*]
    \item $T_H^{\rm SD\text{-}GSEE}=O\!\left(\frac{\beta_H^2}{\eta\varepsilon^3}\log\frac{2}{\delta}\,\log^2\frac{16\alpha_H}{\eta\varepsilon}\right)$ total Hamiltonian-evolution time;
    \item $N_H^{\rm SD\text{-}GSEE}=O\!\left(\frac{\alpha_H\beta_H^2}{\eta\varepsilon^3}\log\frac{2}{\delta}\,\log^2\frac{16\alpha_H}{\eta\varepsilon}\right)$ queries to the uncontrolled oracles $U_H$ and $U_H^\dagger$; and
    \item $N_\psi^{\rm SD\text{-}GSEE}=O\!\left(\frac{\beta_H^2}{\eta\varepsilon^2}\log\frac{2}{\delta}\,\log\frac{16\alpha_H}{\eta\varepsilon}\right)$ fresh copies of $|\psi\rangle$.
\end{enumerate}
In addition, the procedure uses $O(\beta_H^2\varepsilon^{-2}\log(2/\delta))$ energy measurements.
No gap promise is required, and the state preparations do not require access to $\mathrm{ctrl}\text{-}U_H$.
\end{theorem}

\begin{proof}
By \Cref{lem:gap-free-descent} and~\eqref{eq:gsee-bias}, each SD output has energy expectation in $[E_0,E_0+\varepsilon/2]$.
For the average $\widehat E_0$ of $N_E$ independent measurement outcomes, Hoeffding's inequality gives $\Pr[|\widehat E_0-\Tr(H\rho_R)|>\varepsilon/2]\le2\exp[-N_E\varepsilon^2/(8\beta_H^2)]\le\delta$.
On the complementary event, the triangle inequality gives $|\widehat E_0-E_0|\le\varepsilon$.
One preparation uses $O((\eta\varepsilon)^{-1}\log^2(16\alpha_H/(\eta\varepsilon)))$ evolution time, $O(\alpha_H(\eta\varepsilon)^{-1}\log^2(16\alpha_H/(\eta\varepsilon)))$ Hamiltonian queries, and $O(\eta^{-1}\log(16\alpha_H/(\eta\varepsilon)))$ fresh input copies.
Multiplying by $N_E=O(\beta_H^2\varepsilon^{-2}\log(2/\delta))$ gives items~$(i)$--$(iii)$.
\end{proof}

\begin{theorem}[Ground-State Energy Estimation via Coherent Spectral Descent]
\label{thm:gsee-coherent-descent}
Consider an instance of the GSEE problem (Problem~\ref{prob:gsee}) with $0<\varepsilon<\alpha_H$, and suppose that we have coherent access to $U_\psi$ and $U_\psi^\dagger$, where $U_\psi|0^n\rangle=|\psi\rangle$.
Assume the same bounded energy-measurement access as in \Cref{thm:gsee-spectral-descent}.
Choose $h=\varepsilon/4$, $\xi=\varepsilon/(8\alpha_H)$, $\tau=\pi/[4(2\alpha_H+h/2)]$, and $U_H\coloneqq\e^{-\iu H\tau}$.
Then independent applications of the CSD algorithm with the parameters in \Cref{lem:gap-free-descent}, followed by the energy measurements, produce an estimate $\widehat E_0$ satisfying $|\widehat E_0-E_0|\le\varepsilon$ with probability at least $1-\delta$.
The procedure uses
\begin{enumerate}[label=(\roman*),leftmargin=*]
    \item $T_H^{\rm CSD\text{-}GSEE}=O\!\left(\frac{\beta_H^2}{\sqrt\eta\,\varepsilon^3}\log\frac{2}{\delta}\,\log^2\frac{16\alpha_H}{\eta\varepsilon}\right)$ total Hamiltonian-evolution time;
    \item $N_H^{\rm CSD\text{-}GSEE}=O\!\left(\frac{\alpha_H\beta_H^2}{\sqrt\eta\,\varepsilon^3}\log\frac{2}{\delta}\,\log^2\frac{16\alpha_H}{\eta\varepsilon}\right)$ queries to the uncontrolled oracles $U_H$ and $U_H^\dagger$; and
    \item $N_{U_\psi}^{\rm CSD\text{-}GSEE}=O\!\left(\frac{\beta_H^2}{\sqrt\eta\,\varepsilon^2}\log\frac{2}{\delta}\,\log\frac{16\alpha_H}{\eta\varepsilon}\right)$ queries to $U_\psi$ and $U_\psi^\dagger$.
\end{enumerate}
In addition, the procedure uses $O(\beta_H^2\varepsilon^{-2}\log(2/\delta))$ energy measurements.
No gap promise is required, and the state preparations do not require access to $\mathrm{ctrl}\text{-}U_H$.
\end{theorem}

\begin{proof}
The low-energy guarantee in \Cref{lem:gap-free-descent} is the same for CSD as for SD, so~\eqref{eq:gsee-bias} again bounds the preparation bias by $\varepsilon/2$.
The measurement analysis in \Cref{thm:gsee-spectral-descent} therefore proves the claimed accuracy with the same number $N_E$ of independent outputs.
For each output, the CSD resource bounds are $O((\varepsilon\sqrt\eta)^{-1}\log^2(16\alpha_H/(\eta\varepsilon)))$ evolution time, $O(\alpha_H(\varepsilon\sqrt\eta)^{-1}\log^2(16\alpha_H/(\eta\varepsilon)))$ Hamiltonian queries, and $O(\eta^{-1/2}\log(16\alpha_H/(\eta\varepsilon)))$ input-state oracle queries.
Multiplying by $N_E$ gives the three stated bounds.
\end{proof}

\section{Conclusions and Discussion}
\label{sec:conclusions}

We introduced spectral descent as a general approach to ground-state preparation using only uncontrolled Hamiltonian evolution.
Our two algorithms, SD and CSD, use fresh copies of an input state and coherent access to its preparation unitary, respectively.
CSD matches the known GSP lower bound up to logarithmic factors, establishing that controlled Hamiltonian evolution is unnecessary for attaining the near-optimal scaling.
The same framework also gives ground-state property and energy-estimation algorithms under the measurement-access assumptions specified in \Cref{sec:ground-state-estimation}.

For GSEE, allowing controlled Hamiltonian evolution only after descent gives a near-optimal algorithm.
By \Cref{lem:gap-free-descent}, running SD or CSD with $h=\varepsilon/2$ and $\xi=1/16$ produces a state with weight at least $15/16$ on energies in $[E_0,E_0+\varepsilon/2]$, without a spectral-gap promise.
Phase estimation with additive error $\varepsilon/2$ and failure probability at most $1/16$ then estimates $E_0$ to error $\varepsilon$ with probability at least $7/8$, using $O(1/\varepsilon)$ additional evolution time.
Repeating the entire procedure $O(\log(2/\delta))$ times and taking the median gives
\begin{equation}
T_H^{\rm SD+PE} = O\!\left( \frac{\log^2(2/\eta)\log(2/\delta)} {\eta\varepsilon} \right), \qquad T_H^{\rm CSD+PE} = O\!\left( \frac{\log^2(2/\eta)\log(2/\delta)} {\sqrt{\eta}\,\varepsilon} \right).
\end{equation}
The CSD bound matches the known lower bound
$\Omega(\log(1/\delta)/(\sqrt{\eta}\,\varepsilon))$
up to logarithmic factors in $\eta$, in the applicable parameter and dimension regime~\cite{Somma2026}.
Only the phase-estimation step requires controlled evolution, with total controlled evolution time $O(\varepsilon^{-1}\log(2/\delta))$.
This replaces the bounded energy-measurement assumption with controlled Hamiltonian access.

For GSPE, the coherent estimation extension in \Cref{sec:gspe-from-gsp} improves the precision dependence from $1/\varepsilon^2$ to $1/\varepsilon$, up to logarithmic factors, while retaining uncontrolled Hamiltonian access under the stated observable-access assumption.
The resulting total evolution time is $\widetilde O(1/(\gamma\sqrt{\eta}\,\varepsilon))$, but we have not established whether this joint dependence is optimal.
In particular, it remains open whether the precision cost $1/\varepsilon$ must multiply the full preparation cost $1/(\gamma\sqrt{\eta})$, or whether an algorithm can estimate the property with a better joint dependence. Resolving this question would clarify the optimal complexity of GSPE under these access assumptions.

\section{Acknowledgments}
The authors acknowledge the use of ChatGPT 5.6 and 6 for help with writing and reviewing the manuscript, including identifying technical errors and inconsistencies in earlier versions.
This work was supported by the U.S. Department of Energy, Office of Science, National Quantum Information Science Research Centers, Co-design Center for Quantum Advantage (C2QA) under contract number DE-SC0012704, and by the U.S. National Science Foundation National Quantum Virtual Laboratory (NQVL), Erasure Qubits and Dynamic Circuits for Quantum Advantage (ERASE), under grant number OSI-2435244.

\appendix

\section{The SWAP-Echo Identity}
\label{sec:proof-of-exact-signal-lemma}

Expanding the two controlled-SWAPs in~\eqref{eq:W} gives
\begin{equation}
\SWAPECHO(\tau) =|0\rangle\langle0|_q\otimes \e^{-\iu D_{AB}\tau} +|1\rangle\langle1|_q\otimes \operatorname{SWAP}_{AB}\e^{-\iu D_{AB}\tau} \operatorname{SWAP}_{AB}. \label{eq:W-branch-expansion}
\end{equation}
The SWAP exchanges the two tensor factors, so
\begin{align}
\operatorname{SWAP}_{AB}\e^{-\iu D_{AB}\tau} \operatorname{SWAP}_{AB} &=\operatorname{SWAP}_{AB} \left(\e^{-\iu H_A\tau}\otimes\e^{+\iu H_B\tau}\right) \operatorname{SWAP}_{AB} \\
&=\e^{+\iu H_A\tau}\otimes\e^{-\iu H_B\tau} \\
&=\e^{+\iu D_{AB}\tau}.
\end{align}
Substituting this into~\eqref{eq:W-branch-expansion} gives the first equality in~\eqref{eq:W-exact}.  Since $Z_q|0\rangle=|0\rangle$ and $Z_q|1\rangle=-|1\rangle$, the spectral decomposition of $Z_q$ gives the second equality.

Finally, $Z_q\otimes D_{AB}$ commutes with $Z_q\otimes sI$.  Hence
\begin{align}
\e^{-\iu s\tau Z_q}\SWAPECHO(\tau) &=\e^{-\iu Z_q\otimes sI\,\tau} \e^{-\iu Z_q\otimes D_{AB}\tau} \\
&=\e^{-\iu Z_q\otimes(D_{AB}+sI)\tau},
\end{align}
which proves~\eqref{eq:Ws}.

\section{Polynomial Construction of the Energy Comparator}
\label{app:soft-comparator-polynomial}

In this appendix, we prove \Cref{lem:soft-comparator-polynomial}. 
Recall from~\eqref{eq:rstar} that
\begin{equation}
r_\star(y) = \frac{1+\operatorname{erf}(\kappa y)}{2}, \qquad \kappa = \frac{1}{\Delta} \sqrt{\log\frac{8}{\zeta}}.
\end{equation}
We will first approximate the error function $\operatorname{erf}(\cdot)$ and then rescale the resulting polynomial so that the final comparison probability remains in the interval $[0,1]$.

\begin{proof}[Proof of \Cref{lem:soft-comparator-polynomial}]

Set $\epsilon_p\coloneqq \frac{\zeta}{32}$.
The polynomial approximation to the error function given by~\cite[Corollary~4]{Low2017} implies that there exists a real odd polynomial $p_n$ satisfying
\begin{equation}
\sup_{y\in[-1,1]} \left| p_n(y)-\operatorname{erf}(\kappa y) \right| \le \epsilon_p, \label{eq:erf-poly-error}
\end{equation}
with degree
\begin{equation}
d = O\!\left( \sqrt{ \left( \kappa^2+\log\frac{1}{\epsilon_p} \right) \log\frac{1}{\epsilon_p} } \right). \label{eq:erf-degree}
\end{equation}

The polynomial $p_n$ approximates the desired function uniformly, but it is not necessarily bounded by $1$. 
We therefore first enforce this boundedness.
Since $\left| \operatorname{erf}(\kappa y) \right| \le1$.
for every real $y$, \eqref{eq:erf-poly-error} gives
\begin{equation}
|p_n(y)| \le \left| p_n(y)-\operatorname{erf}(\kappa y) \right| + \left| \operatorname{erf}(\kappa y) \right| \le \epsilon_p+1.
\end{equation}
Define the rescaled polynomial
\begin{equation}
\bar p_n(y) \coloneqq \frac{p_n(y)}{1+\epsilon_p}. \label{eq:pbar}
\end{equation}
It follows immediately that $|\bar p_n(y)| \le1$ for every $y\in[-1,1]$.

We next verify that this rescaling does not significantly worsen the approximation error. 
Using~\eqref{eq:pbar}, we have
\begin{align}
\left| \bar p_n(y) - \operatorname{erf}(\kappa y) \right| &= \left| \frac{ p_n(y) - (1+\epsilon_p) \operatorname{erf}(\kappa y) } {1+\epsilon_p} \right| \\
&\le \frac{ \left| p_n(y)-\operatorname{erf}(\kappa y) \right| + \epsilon_p \left| \operatorname{erf}(\kappa y) \right| } {1+\epsilon_p} \\
&\le \frac{2\epsilon_p}{1+\epsilon_p} \\
&\le 2\epsilon_p \\
&= \frac{\zeta}{16}. \label{eq:pbar-error}
\end{align}

We now define the polynomial that will represent the witness probability:
\begin{equation}
\widetilde r(y) \coloneqq \frac{1+\bar p_n(y)}{2}. \label{eq:rtilde-poly}
\end{equation}
Since $|\bar p_n(y)|\le1$, we have $0\le \widetilde r(y) \le 1$ for every $y\in[-1,1],$.

It remains now to compare $\widetilde r$ with $r_\star$. 
Using~\eqref{eq:rstar} and~\eqref{eq:rtilde-poly}, we have
\begin{align}
\left| \widetilde r(y)-r_\star(y) \right| &= \frac12 \left| \bar p_n(y) - \operatorname{erf}(\kappa y) \right| \\
&\le \frac12 \cdot \frac{\zeta}{16} \\
&= \frac{\zeta}{32},
\end{align}
where the inequality uses~\eqref{eq:pbar-error}. 
This proves~\eqref{eq:rtilde-monotone-close}.

We finally verify the degree claimed in the lemma statement. 
From the definition of $\kappa$, we have $\kappa^2 = \frac{1}{\Delta^2} \log\frac{8}{\zeta}$,
while $\log\frac{1}{\epsilon_p} = \log\frac{32}{\zeta}$.
Since $\zeta\in(0,1/4)$, both
$\log(8/\zeta)$ and $\log(32/\zeta)$ are
$O(\log(1/\zeta))$. 
Moreover, $0<\Delta\le1$, so $1/\Delta^2\ge1$. 
Therefore
\begin{align}
\kappa^2+\log\frac{1}{\epsilon_p} &= O\!\left( \frac{1}{\Delta^2} \log\frac{1}{\zeta} \right).
\end{align}
Substituting this bound into~\eqref{eq:erf-degree},
\begin{equation}
d = O\!\left( \sqrt{ \frac{1}{\Delta^2} \log\frac{1}{\zeta} \log\frac{1}{\zeta} } \right) = O\!\left( \frac{1}{\Delta} \log\frac{1}{\zeta} \right).
\end{equation}
This is exactly the degree bound in the lemma statement.
\end{proof}

\section{Laurent-QSP Implementation of the Energy Comparator}
\label{app:laurent-qsp-synthesis}

In this appendix, we prove \Cref{lem:laurent-qsp-comparator}. 
Recall from~\eqref{eq:R-Laurent} that
\begin{equation}
\mathcal R(z) = \widetilde r\!\left( \frac{z-z^{-1}}{2\iu} \right).
\end{equation}
Since $\widetilde r$ is a real polynomial of degree $n$, $\mathcal R$ is a Laurent polynomial of degree at most $n$. 
For $|z|=1$, we may write $z=\e^{2\iu\theta}$, in which case
$\frac{z-z^{-1}}{2\iu} = \sin(2\theta)\in[-1,1].$
Also, from~\Cref{lem:soft-comparator-polynomial}, we have $0\le\widetilde r(y)\le1$ for every $y\in[-1,1]$.
This implies that $0\le\mathcal R(z)\le1$ for $|z|=1$.

\begin{proof}[Proof of \Cref{lem:laurent-qsp-comparator}]

Because $\mathcal R(z)$ and $1-\mathcal R(z)$ are nonnegative Laurent polynomials on the unit circle, the Fej\'er--Riesz theorem gives Laurent-polynomial spectral factors $A(z)$ and $B(z)$, each of degree at most $n$, satisfying
\begin{equation}
A(z)A(z)^* = 1-\mathcal R(z), \qquad B(z)B(z)^* = \mathcal R(z), \label{eq:FR-factors}
\end{equation}
on $|z|=1$, where $f(z)^* \coloneqq \overline{f(1/\bar z)}$.
On the unit circle, $f(z)^*=\overline{f(z)}$, so~\eqref{eq:FR-factors} gives
$|A(z)|^2+|B(z)|^2=1.$

Next, define the matrix Laurent polynomial as follows:
\begin{equation}
F(z) \coloneqq \begin{pmatrix} A(z) & -B(z)^*\\ B(z) & A(z)^* \end{pmatrix}. \label{eq:F-Laurent}
\end{equation}
Using~\eqref{eq:FR-factors}, we have on $|z|=1$
$F(z)^\dagger F(z) = I,$
and
\begin{align}
\det F(z) &= A(z)A(z)^* + B(z)B(z)^* \\
&= 1-\mathcal R(z)+\mathcal R(z) \\
&= 1.
\end{align}
Hence
$F(z)\in SU(2)$
for $|z|=1.$

We now express this matrix polynomial in terms of the phase variable. 
Recall our notation
$t=\e^{\iu\theta}$ 
and
$z=t^2.$
Since $F(z)$ has Laurent degree at most $n$ in $z$, the matrix Laurent polynomial $F(t^2)$ has degree at most $2n$ in $t$.
Moreover, because the substitution $z=t^2$ doubles every exponent, all Laurent powers appearing in $F(t^2)$ are even.
Thus $F(t^2)$ satisfies the parity condition required by Haah's Laurent-QSP factorization theorem~\cite[Theorem~2]{Haah2019}.
The theorem therefore decomposes $F(t^2)$ into a product of at most $2n$ primitive signal factors, together with a residual single-qubit unitary.

Each primitive factor has the form
\begin{equation}
E_P(t) = tP+t^{-1}(I-P), \label{eq:primitive-factor}
\end{equation}
where $P$ is a rank-one projector on the signal qubit.

Define
\begin{equation}
V(t) \coloneqq \begin{pmatrix} t^{-1} & 0\\ 0 & t \end{pmatrix}.
\end{equation}
For any rank-one projector $P$, choose a single-qubit unitary $R_P$ satisfying $R_P|1\rangle\langle1|R_P^\dagger=P$. 
Then
\begin{equation}
E_P(t) = R_PV(t)R_P^\dagger. \label{eq:primitive-conjugation}
\end{equation}
Thus every primitive factor can be implemented using one call to $V(t)$, surrounded only by single-qubit rotations.

We now connect this Laurent-QSP factorization to the physical signal used in the main text. 
Equation~\eqref{eq:physical-Vt} already shows that, on every occupied energy pair $(a,b)$,
\begin{equation}
\SWAPECHO_\gamma(\tau)\big|_{a,b} = V(t_{ab}), \qquad t_{ab}=\e^{\iu\theta_{ab}}.
\end{equation}
Therefore each occurrence of $V(t)$ in the Laurent-QSP factorization is implemented by one direct call to $\SWAPECHO_\gamma(\tau)$, surrounded only by single-qubit rotations. 
Since $F(t^2)$ has Laurent degree at most $2n$ in $t$, the resulting sequence uses at most $2n$ calls to $\SWAPECHO_\gamma(\tau)$.

Let $C_\zeta$ denote this sequence. 
For
$z_{ab}=t_{ab}^2=\e^{2\iu\theta_{ab}},$
the action of $C_\zeta$ on an initial signal state $|0\rangle_q$ is given by the first column of $F(z_{ab})$. 
Hence
\begin{equation}
C_\zeta |0\rangle_q|E_aE_b\rangle = A(z_{ab}) |0\rangle_q|E_aE_b\rangle + B(z_{ab}) |1\rangle_q|E_aE_b\rangle,
\end{equation}
which is~\eqref{eq:comparator-action}.

Finally, the Fej\'er--Riesz factorization gives $|B(z)|^2=\mathcal R(z)$ for $|z|=1$.
Evaluating this identity at $z=z_{ab}$ and using~\eqref{eq:R-Laurent}, we get
\begin{align}
|B(z_{ab})|^2 &= \mathcal R(z_{ab}) \\
&= \widetilde r\!\left( \frac{z_{ab}-z_{ab}^{-1}}{2\iu} \right) \\
&= \widetilde r\!\left( \sin(2\theta_{ab}) \right) \\
&= \widetilde r(y_{ab}),
\end{align}
which proves~\eqref{eq:qsp-witness-probability}.

The Laurent-QSP sequence calls $\SWAPECHO_\gamma(\tau)$ directly; it does not require controlled access to the signal unitary. 
By \Cref{sec:swap-echo}, each call to $\SWAPECHO_\gamma(\tau)$ uses only uncontrolled forward and backward Hamiltonian evolution, controlled-SWAP gates, and single-qubit phases.
\end{proof}

\section{Candidate Distribution Conditioned on a Witness in SD}
\label{app:measured-witness-direction}

\begin{proof}[Proof of \Cref{lem:measured-direction-preserved}]
Fix an anchor energy level $\mathcal E_j$, and consider one descent stage containing at most $M$ comparison trials.

Let $\ell\in\{1,\ldots,M\}$ denote the index of the first comparison trial that produces the witness outcome $q=1$. 
For the first witness to occur on trial $\ell$, the preceding $\ell-1$ comparison trials must all produce $q=0$.

Because a no-witness outcome leaves the anchor unchanged, the anchor remains at energy $\mathcal E_j$ throughout these preceding trials. 
Therefore each of them has the same no-witness probability $1-\lambda_j$.

On comparison trial $\ell$, the probability that the fresh candidate register occupies energy level $\mathcal E_k$ and produces a witness is $\pi_k r_{jk}$. 
Hence the joint probability that the first successful comparison trial is trial $\ell$ and that its candidate register occupies level $\mathcal E_k$ is
\begin{equation}
\Pr[\text{first witness at } \ell, \mathcal E_k] = (1-\lambda_j)^{\ell-1} \pi_k r_{jk}. \label{eq:first-witness-joint}
\end{equation}

We now sum over the possible locations of the first successful comparison trial. 
The probability that the descent stage succeeds (finds a witness) and that the successful candidate register occupies level $\mathcal E_k$ is
\begin{equation}
\Pr[ \mathcal E_k,\text{ stage succeeds} ] = \sum_{\ell=1}^{M} (1-\lambda_j)^{\ell-1} \pi_k r_{jk} = \pi_k r_{jk} \sum_{\ell=1}^{M} (1-\lambda_j)^{\ell-1}. \label{eq:witness-k-summed}
\end{equation}

Similarly, summing over all possible cadidate energies gives the total probability that the descent stage succeeds:
\begin{align}
\Pr[ \text{stage succeeds} ] &= \sum_{\ell=1}^{M} (1-\lambda_j)^{\ell-1} \sum_k \pi_k r_{jk} \\
&= \lambda_j \sum_{\ell=1}^{M} (1-\lambda_j)^{\ell-1}, \label{eq:witness-total-summed}
\end{align}
where the definition of $\lambda_j$ in~\eqref{eq:lambda-j} was used in the second line.

Conditioning on the event that the descent stage succeeds, we therefore obtain
\begin{align}
\Pr[ k \mid \text{stage succeeds} ] &= \frac{ \pi_k r_{jk} \sum_{\ell=1}^{M} (1-\lambda_j)^{\ell-1} }{ \lambda_j \sum_{\ell=1}^{M} (1-\lambda_j)^{\ell-1} } \\
&= \frac{\pi_k r_{jk}}{\lambda_j} \\
&= q_j(k),
\end{align}
where the final equality uses~\eqref{eq:qjk}.

Thus the first successful candidate register in a descent stage has exactly the same conditional energy distribution as the candidate register in a single successful comparison trial. 
The additional comparison trials change only the probability that the descent stage finds a witness; they do not change the candidate-energy distribution conditioned on success.
\end{proof}

\section{Expected Candidate Potential Conditioned on a Witness}
\label{app:candidate-potential}

In this appendix, we prove \Cref{lem:successful-candidate-potential}, which bounds the expected potential of an accepted candidate state. 
Fix an excited anchor level $\mathcal E_j$, with $j>0$.
We proceed in three steps that mirror the narrative in~\Cref{sec:descent-potential}.
First, we compute the expected potential of a candidate sampled from the naive input distribution $\mu_j$, proving the factor of $1/2$ claimed in~\eqref{eq:unweighted-half-main}. 
Second, we show that weighting these candidates by the ideal monotone comparator $\nu_j^\star$ (defined in~\eqref{eq:v-j-star-define}) can only improve (decrease) this expected potential due to negative correlation, proving~\eqref{eq:ideal-candidate-half-main}.
Finally, we introduce the two physical sources of error, that is, QSP polynomial approximations and false witnesses (candidates with $k \ge j$), and bound their impact to arrive at our final result.

\begin{proof}[Proof of \Cref{lem:successful-candidate-potential}]

Recall from~\eqref{eq:Wj} the definition of the descent potential $W_j$ and also recall that the naive distribution $\mu_j$ is defined as $\mu_j(k) = \pi_k/W_j$ for $k<j$.
Since $j>0$, \eqref{eq:W-excited-lower} gives $W_j\ge\eta>0$, so this distribution is well defined.

We first calculate the expected potential under $\mu_j$. 
Using the definition $W_k=\sum_{\ell<k}\pi_\ell$, we get
\begin{equation}
\sum_{k<j}\pi_k W_k = \sum_{k<j} \pi_k \sum_{\ell<k}\pi_\ell = \sum_{\ell<k<j} \pi_\ell\pi_k. \label{eq:pair-sum-app}
\end{equation}
On the other hand, we have
\begin{equation}
W_j^2 = \left( \sum_{k<j}\pi_k \right)^2 = \sum_{k<j}\pi_k^2 + 2 \sum_{\ell<k<j} \pi_\ell\pi_k. \label{eq:W-square-app}
\end{equation}
Solving~\eqref{eq:W-square-app} for the pair sum and substituting it back into~\eqref{eq:pair-sum-app} yields the following identity:
\begin{equation}
\sum_{k<j}\pi_kW_k = \frac12 \left( W_j^2-\sum_{k<j}\pi_k^2 \right). \label{eq:pair-sum-evaluated-app}
\end{equation}
Dividing by $W_j$ to form the expectation gives us the strict contraction bound claimed in~\eqref{eq:unweighted-half-main}:
\begin{equation}
\E_{k\sim\mu_j}[W_k] = \sum_{k<j} \frac{\pi_k}{W_j}W_k = \frac{ W_j^2-\sum_{k<j}\pi_k^2 }{ 2W_j }\le \frac{W_j}{2}. \label{eq:unweighted-half-app}
\end{equation}

We next prove the first inequality in~\eqref{eq:ideal-candidate-half-main}. 
Recall that $\nu_j^\star$ (defined in~\eqref{eq:v-j-star-define}) is obtained from $\mu_j$ by weighting each lower-energy candidate $k<j$ by $r^\star_{jk}$. 
Because $r^\star_{jk}$ is monotonically decreasing with $k$, while the potential $W_k$ is monotonically increasing with $k$, they are negatively correlated.

We can make this negative correlation explicit.
Writing expectations with respect to $\mu_j$, we have
\begin{equation}
\E_{\mu_j}[r^\star] \E_{\mu_j}[W] - \E_{\mu_j}[r^\star W] = \frac12 \sum_{k,\ell<j} \mu_j(k)\mu_j(\ell) \left( r^\star_{jk}-r^\star_{j\ell} \right) \left( W_\ell-W_k \right). \label{eq:negative-correlation-identity-app}
\end{equation}
Note that every term on the right-hand side is nonnegative. 
Indeed, if $k<\ell$, then
$r^\star_{jk}\ge r^\star_{j\ell}$ and $W_k\le W_\ell$, while for $k>\ell$ both differences change sign. 
Therefore, the overall sum is nonnegative, giving:
\begin{equation}
\E_{\mu_j}[r^\star W] \le \E_{\mu_j}[r^\star] \E_{\mu_j}[W]. \label{eq:negative-correlation-app}
\end{equation}
By the definition of $\nu_j^\star$, its expected potential is exactly the ratio of these expectations, and then using the inequality above, we get
\begin{equation}
\E_{k\sim\nu_j^\star}[W_k] = \frac{ \E_{\mu_j}[r^\star W] }{ \E_{\mu_j}[r^\star] }\le \E_{\mu_j}[W_k] \le \frac{W_j}{2}, \label{eq:ideal-weighted-half-app}
\end{equation}
where the final inequality follows from~\eqref{eq:unweighted-half-app}. 
This proves~\eqref{eq:ideal-candidate-half-main}.
In particular, the factor $1/2$ follows from the cumulative form of $W_k$, while monotonicity of $r^\star_{jk}$ ensures that the comparator weighting cannot make this bound worse.

We now replace the ideal probabilities $r^\star_{jk}$ in the definition of $\nu_j^\star$ (see~\eqref{eq:v-j-star-define}) by the probabilities $r_{jk}$, and define the following distribution:
\begin{equation}
\widetilde\nu_j(k) \coloneqq \frac{\pi_k r_{jk}}{\sum_{\ell<j}\pi_\ell r_{j\ell}}, \qquad k<j. \label{eq:nu-physical-app}
\end{equation}
With that in place, we now compare $\widetilde\nu_j$ with the ideal distribution $\nu_j^\star$.
Consider
\begin{equation}
\sum_{k<j} \pi_k |r_{jk}-r^\star_{jk}| \le \frac{\zeta}{32} \sum_{k<j}\pi_k = \frac{\zeta}{32}W_j. \label{eq:unnormalized-error-app}
\end{equation}
Also, the total witness weight coming from these lower-energy candidates satisfies
\begin{equation}
\sum_{k<j}\pi_k r_{jk} \ge \pi_0 r_{j0} \ge \eta(1-\zeta), \label{eq:lower-witness-weight-app}
\end{equation}
where we used $j>0$, the overlap promise $\pi_0 \geq \eta$, and~\eqref{eq:pair-reject}.

We can now bound the total variation distance between these two distributions. 
Using their definitions, consider
\begin{align}
2\operatorname{TV} \left( \widetilde\nu_j, \nu_j^\star \right) &= \sum_{k<j} \left| \frac{\pi_k r_{jk}} {\sum_{\ell<j}\pi_\ell r_{j\ell}} - \frac{\pi_k r^\star_{jk}} {\sum_{\ell<j}\pi_\ell r^\star_{j\ell}} \right| \\
&\le \sum_{k<j} \left| \frac{\pi_k r_{jk}} {\sum_{\ell<j}\pi_\ell r_{j\ell}} - \frac{\pi_k r^\star_{jk}} {\sum_{\ell<j}\pi_\ell r_{j\ell}} \right| + \sum_{k<j} \left| \frac{\pi_k r^\star_{jk}} {\sum_{\ell<j}\pi_\ell r_{j\ell}} - \frac{\pi_k r^\star_{jk}} {\sum_{\ell<j}\pi_\ell r^\star_{j\ell}} \right|. \label{eq:TV-split-app}
\end{align}
The first term in~\eqref{eq:TV-split-app} is
\begin{equation}
\frac{ \sum_{k<j}\pi_k|r_{jk}-r^\star_{jk}| }{ \sum_{\ell<j}\pi_\ell r_{j\ell} }. \label{eq:TV-first-term-app}
\end{equation}
For the second sum, factoring out the numerator and applying the triangle inequality yields the exact same bound:
\begin{align}
\left| \frac{ \sum_{k<j}\pi_k r^\star_{jk} }{ \sum_{\ell<j}\pi_\ell r_{j\ell} } - 1 \right| &= \frac{ \left| \sum_{k<j}\pi_k (r^\star_{jk}-r_{jk}) \right| }{ \sum_{\ell<j}\pi_\ell r_{j\ell} } \\
&\le \frac{ \sum_{k<j}\pi_k |r_{jk}-r^\star_{jk}| }{ \sum_{\ell<j}\pi_\ell r_{j\ell} }. \label{eq:TV-second-term-app}
\end{align}
Combining~\eqref{eq:TV-split-app}--\eqref{eq:TV-second-term-app}, the factor of 2 cancels, and gives
\begin{equation}
\operatorname{TV} \left( \widetilde\nu_j, \nu_j^\star \right) \le \frac{ \sum_{k<j}\pi_k |r_{jk}-r^\star_{jk}| }{ \sum_{\ell<j}\pi_\ell r_{j\ell} }. \label{eq:TV-direct-app}
\end{equation}
Using~\eqref{eq:unnormalized-error-app} and~\eqref{eq:lower-witness-weight-app}, we obtain
\begin{equation}
\operatorname{TV} \left( \widetilde\nu_j, \nu_j^\star \right) \le \frac{ \zeta W_j }{ 32\eta(1-\zeta) }. \label{eq:TV-lower-distributions-app}
\end{equation}

For $k<j$, we have $0\le W_k\le W_j$. 
Therefore~\eqref{eq:ideal-weighted-half-app} and~\eqref{eq:TV-lower-distributions-app} imply
\begin{align}
\E_{k\sim\widetilde\nu_j}[W_k] &\le \E_{k\sim\nu_j^\star}[W_k] + W_j \operatorname{TV} \left( \widetilde\nu_j, \nu_j^\star \right) \\
&\le \frac{W_j}{2} + \frac{ \zeta W_j^2 }{ 32\eta(1-\zeta) } \\
&\le \left( \frac12 + \frac{\zeta} {32\eta(1-\zeta)} \right) W_j, \label{eq:physical-lower-expectation-app}
\end{align}
where the last inequality uses $W_j\le1$.

At this point, we use the assumption that we stated in \Cref{lem:successful-candidate-potential}: the comparator error $\zeta$ is small enough such that $\frac{\zeta}{32\eta(1-\zeta)} \le \frac{1}{16}$. 
Substituting this in directly yields:
\begin{equation}
\E_{k\sim\widetilde\nu_j}[W_k] \le \frac{9}{16}W_j. \label{eq:nine-sixteen-app}
\end{equation}
Thus, conditioned on the witness corresponding to a genuinely lower-energy candidate, the physical comparator still contracts the expected candidate potential by a constant factor of $9/16$.

It remains to account for false witnesses with $k\ge j$. 
For every such candidate, we have
$\mathcal E_j-\mathcal E_k \le 0,$
so~\eqref{eq:pair-pass} gives $r_{jk}\le\zeta$. 
Conditioned on obtaining a witness, the probability that the candidate is not actually below the anchor is therefore
\begin{equation}
\Pr[ k\ge j \mid j,\text{ witness} ] = \frac{ \sum_{k\ge j}\pi_k r_{jk} }{ \lambda_j }\le \frac{ \zeta\sum_{k\ge j}\pi_k }{ \lambda_j }\le \frac{\zeta}{\lambda_j} \le \frac{\zeta}{\eta(1-\zeta)}, \label{eq:false-witness-fraction-app}
\end{equation}
where the final inequality uses~\eqref{eq:lambda-excited-lower}.

To find the total expected potential of an accepted candidate, we combine the lower-energy case~\eqref{eq:nine-sixteen-app} with the worst-case scenario for a false witness (which trivially has potential $W_k \le 1$). This final sum gives:
\begin{align}
\E[ W_{\rm cand} \mid j,\text{ witness} ] &\le \frac{9}{16}W_j + \Pr[ k\ge j \mid j,\text{ witness} ] \\
&\le \frac{9}{16}W_j + \frac{\zeta}{\eta(1-\zeta)}. \label{eq:candidate-potential-app}
\end{align}
This explicitly proves the bound in~\eqref{eq:candidate-potential} for a single successful comparison trial.

Finally, \Cref{lem:measured-direction-preserved} shows that, conditioned on a descent stage finding a witness, its first successful comparison register has the same distribution $q_j(k)$ as a single successful comparison trial. 
Therefore~\eqref{eq:candidate-potential} also applies to the first successful comparison register in a descent stage.
\end{proof}

\section{SD Potential Recursion for General Anchor States}
\label{app:measured-potential-recursion}

We now justify the recursion~\eqref{eq:V-recursion-measured} when the anchor at the beginning of a descent stage is an arbitrary state $\rho_t$. 
The bounds in \Cref{sec:descent-potential} were derived for an anchor with fixed energy $\mathcal E_j$. 
Our goal is to show that these fixed-energy bounds can be applied separately to the different energy sectors of $\rho_t$.

Decompose the anchor state as
\begin{equation}
\rho_t = \sum_{j,j'} \Pi_j\rho_t\Pi_{j'}. \label{eq:rho-energy-decomposition}
\end{equation}
For convenience, define $X_{jj'}\coloneqq\Pi_j\rho_t\Pi_{j'}$. 
The diagonal block $X_{jj}$ has weight
$\Tr(X_{jj})=\Tr(\Pi_j\rho_t)$, whereas the blocks $X_{jj'}$ with $j\neq j'$ describe coherences between distinct anchor energies.

We first show that these off-diagonal blocks do not contribute to the energy populations of the anchor after a descent stage.

Let $\sigma_\psi\coloneqq|\psi\rangle\langle\psi|$ denote the state of a fresh candidate register $B$. 
Its population at energy $\mathcal E_k$ is
$\pi_k=\Tr(\Pi_k\sigma_\psi)$.
By the comparator action~\eqref{eq:comparator-action}, measuring the witness qubit induces the two operators
\begin{equation}
K_0 = \sum_{a,b} A(z_{ab})\, \Pi_a^A\otimes\Pi_b^B, \qquad K_1 = \sum_{a,b} B(z_{ab})\, \Pi_a^A\otimes\Pi_b^B, \label{eq:measured-comparator-Kraus}
\end{equation}
for the no-witness and witness outcomes, respectively. 
By~\eqref{eq:comparator-action} and~\eqref{eq:rab}, $|B(z_{ab})|^2=r_{ab}$ and $|A(z_{ab})|^2=1-r_{ab}$.

Consider first the no-witness outcome. 
Applying $K_0$ to a single energy block $X_{jj'}$ gives
\begin{align}
&K_0 \bigl( X_{jj'}\otimes\sigma_\psi \bigr) K_0^\dagger \nonumber \\
&\qquad= \sum_{a,b,a',b'} A(z_{ab})A(z_{a'b'})^* \bigl( \Pi_a X_{jj'}\Pi_{a'} \bigr) \otimes \bigl( \Pi_b\sigma_\psi\Pi_{b'} \bigr). \label{eq:no-witness-full-expansion}
\end{align}
Since $X_{jj'}=\Pi_j\rho_t\Pi_{j'}$, orthogonality of the spectral projectors gives
$\Pi_aX_{jj'}\Pi_{a'}=0$ unless $a=j$ and $a'=j'$. 
Hence
\begin{equation}
K_0 \bigl( X_{jj'}\otimes\sigma_\psi \bigr) K_0^\dagger = \sum_{k,k'} A(z_{jk})A(z_{j'k'})^* X_{jj'} \otimes \Pi_k\sigma_\psi\Pi_{k'}. \label{eq:no-witness-block-before-trace}
\end{equation}
After a no-witness outcome, register $B$ is discarded. 
Using cyclicity of the trace and $\Pi_{k'}\Pi_k=\delta_{kk'}\Pi_k$,
\begin{equation}
\Tr(\Pi_k\sigma_\psi\Pi_{k'})=\delta_{kk'}\pi_k.
\end{equation}
Therefore
\begin{equation}
\Tr_B \left[ K_0 \bigl( X_{jj'}\otimes\sigma_\psi \bigr) K_0^\dagger \right] = c_{jj'}X_{jj'}, \qquad c_{jj'} \coloneqq \sum_k \pi_k A(z_{jk})A(z_{j'k})^*. \label{eq:no-witness-block-after-trace}
\end{equation}
Thus a no-witness outcome preserves the pair of anchor-energy sectors $(j,j')$: the block $X_{jj'}$ is only multiplied by a scalar.

If $j\neq j'$, this block contributes zero to every anchor energy population. 
Indeed, for any energy level $\mathcal E_\ell$,
\begin{equation}
\Tr( \Pi_\ell X_{jj'} ) = \Tr( \Pi_\ell\Pi_j\rho_t\Pi_{j'} ) = \Tr( \Pi_{j'}\Pi_\ell\Pi_j\rho_t ) = 0, \label{eq:no-witness-coherence-zero}    
\end{equation}
because $j\neq j'$ implies $\Pi_{j'}\Pi_\ell\Pi_j=0$ for every $\ell$. 
By~\eqref{eq:no-witness-block-after-trace}, the same conclusion remains true after any number of consecutive no-witness outcomes.

Now consider a witness outcome. 
Replacing $K_0$ by $K_1$ in the same calculation gives
\begin{equation}
K_1 \bigl( X_{jj'}\otimes\sigma_\psi \bigr) K_1^\dagger = \sum_{k,k'} B(z_{jk})B(z_{j'k'})^* X_{jj'} \otimes \Pi_k\sigma_\psi\Pi_{k'}. \label{eq:witness-block-before-swap}
\end{equation}
After a witness, registers $A$ and $B$ are swapped. 
The state in~\eqref{eq:witness-block-before-swap} therefore becomes
\begin{equation}
\sum_{k,k'} B(z_{jk})B(z_{j'k'})^* \Pi_k\sigma_\psi\Pi_{k'} \otimes X_{jj'}.
\end{equation}
The second factor is now the old anchor and is discarded. 
Tracing it out gives
\begin{equation}
\Tr_B \left[ \operatorname{SWAP} K_1 \bigl( X_{jj'}\otimes\sigma_\psi \bigr) K_1^\dagger \operatorname{SWAP} \right] = \Tr(X_{jj'}) \sum_{k,k'} B(z_{jk})B(z_{j'k'})^* \Pi_k\sigma_\psi\Pi_{k'}. \label{eq:witness-block-after-trace}
\end{equation}
For $j\neq j'$,
\begin{equation}
\Tr(X_{jj'}) = \Tr( \Pi_j\rho_t\Pi_{j'} ) = \Tr( \Pi_{j'}\Pi_j\rho_t ) = 0. \label{eq:off-diagonal-trace-zero}
\end{equation}
Hence a witness outcome removes every off-diagonal anchor-energy block completely.

We can now apply these observations to a complete descent stage. 
Each branch of the stage consists of a sequence of no-witness outcomes, followed either by a witness or by termination after $M$ unsuccessful trials. 
By~\eqref{eq:no-witness-block-after-trace}, consecutive no-witness outcomes keep an initial off-diagonal block $X_{jj'}$ within the same pair of sectors $(j,j')$, so by~\eqref{eq:no-witness-coherence-zero} it never contributes to an anchor energy population. 
If a witness eventually occurs,~\eqref{eq:off-diagonal-trace-zero} removes the block completely. 
Therefore, in every branch of the descent stage, the blocks $X_{jj'}$ with $j\neq j'$ make zero contribution to the final anchor energy populations.

Let $\mathcal S$ denote the map implemented by one complete descent stage. 
The preceding calculation implies, for every output energy $\mathcal E_\ell$,
\begin{equation}
\Tr\!\left( \Pi_\ell\mathcal S(\rho_t) \right) = \sum_j \Tr\!\left( \Pi_\ell\mathcal S(X_{jj}) \right). \label{eq:output-populations-diagonal-blocks}
\end{equation}
Thus only the diagonal energy blocks of $\rho_t$ are relevant for computing the output potential.

Set $p_j\coloneqq\Tr(X_{jj})=\Tr(\Pi_j\rho_t)$. 
For $p_j>0$, define the normalized state
$\widehat\rho_j\coloneqq X_{jj}/p_j$.
This state is supported entirely in the eigenspace with energy $\mathcal E_j$. 
By linearity of $\mathcal S$,
$\mathcal S(X_{jj})=p_j\mathcal S(\widehat\rho_j)$.
Using~\eqref{eq:output-populations-diagonal-blocks} and the definition~\eqref{eq:Vt}, we therefore obtain
\begin{align}
V_{t+1} &= \sum_\ell W_\ell \Tr\!\left( \Pi_\ell\mathcal S(\rho_t) \right) \nonumber \\
&= \sum_j p_j \sum_\ell W_\ell \Tr\!\left( \Pi_\ell\mathcal S(\widehat\rho_j) \right). \label{eq:Vt-diagonal-decomposition}
\end{align}

For $j=0$, the fixed-energy bound~\eqref{eq:measured-ground-potential} gives
\begin{equation}
\sum_\ell W_\ell \Tr\!\left( \Pi_\ell\mathcal S(\widehat\rho_0) \right) \le M\zeta.
\end{equation}
For every $j>0$, the fixed-energy contraction~\eqref{eq:measured-stage-contraction} gives
\begin{equation}
\sum_\ell W_\ell \Tr\!\left( \Pi_\ell\mathcal S(\widehat\rho_j) \right) \le \frac{43}{64}W_j + \frac{\zeta}{\eta(1-\zeta)}.
\end{equation}
Substituting these bounds into~\eqref{eq:Vt-diagonal-decomposition} yields
\begin{align}
V_{t+1} &\le \Tr(\Pi_0\rho_t)\, M\zeta+ \sum_{j>0} \Tr(\Pi_j\rho_t) \left( \frac{43}{64}W_j + \frac{\zeta}{\eta(1-\zeta)} \right) \nonumber \\
&= \frac{43}{64} \sum_{j>0} W_j\Tr(\Pi_j\rho_t) + \Tr(\Pi_0\rho_t)\, M\zeta + \frac{\zeta}{\eta(1-\zeta)} \sum_{j>0} \Tr(\Pi_j\rho_t). \label{eq:Vt-arbitrary-anchor-bound}
\end{align}
Since $\Tr(\Pi_0\rho_t)\le1$ and
$\sum_{j>0}\Tr(\Pi_j\rho_t)\le1$, while $W_0=0$ and hence
$\sum_{j>0}W_j\Tr(\Pi_j\rho_t)=V_t$ by~\eqref{eq:Vt}, we conclude that
\begin{equation}
V_{t+1} \le \frac{43}{64}V_t + M\zeta + \frac{\zeta}{\eta(1-\zeta)}.
\end{equation}
This is exactly~\eqref{eq:V-recursion-measured}.

\section{Witness Amplification in CSD}
\label{app:coherent-fixed-point-search}

\begin{proof}[Proof of \Cref{lem:coherent-witness-amplification}]
Fix an anchor energy level $\mathcal E_j$. 
After preparing the candidate state and applying the comparator, the witness probability is $\lambda_j$. 
For a fixed anchor state in this energy eigenspace, the source reflection in~\eqref{eq:Ssrc} acts as the phase reflection about the prepared state, as shown in \Cref{sec:coherent-witness-search}.
Together with the witness reflection, it preserves the span of the normalized witness and no-witness components, so the success probability depends on the anchor level only through $\lambda_j$.

For an odd sequence length $L$, the FPAA construction of Ref.~\cite{Yoder2014} with failure parameter $\delta_{\rm fp}=1/2$ gives the amplified witness probability
\begin{equation}
P_L(\lambda) = 1-\frac14 T_L\!\left( T_{1/L}(2)\sqrt{1-\lambda} \right)^2, \label{eq:CSD-YLC-success}
\end{equation}
where $T_L$ is the Chebyshev polynomial of the first kind. 
The construction guarantees $P_L(\lambda)\ge3/4$ whenever
\begin{equation}
\lambda \ge 1-T_{1/L}(2)^{-2}.
\end{equation}
Using $T_{1/L}(2)=\cosh(\operatorname{arcosh}(2)/L)$, the right-hand side is
\begin{equation}
1- \frac{1}{ \cosh^2(\operatorname{arcosh}(2)/L) } = \tanh^2\!\left( \frac{\operatorname{arcosh}2}{L} \right) \le \left( \frac{\operatorname{arcosh}2}{L} \right)^2.
\end{equation}
Now choosing $L$ so that it satisfies~\eqref{eq:L-choice}, the right-hand side of the above inequality is at most $\eta(1-\zeta)$. 
Also, we know that every excited anchor level satisfies $\lambda_j\ge\eta(1-\zeta)$ by~\eqref{eq:lambda-excited-lower}.
This means that $\lambda_j \ge 1-T_{1/L}(2)^{-2}$ is satisfied.
Hence,
\begin{equation}
s_j^{\rm CSD} = P_L(\lambda_j) \ge \frac34, \qquad j>0,
\end{equation}
which proves~\eqref{eq:excited-coherent-success}.

It remains to control the opposite regime, where the raw witness probability is very small. 
This is needed because for a ground-state anchor we only know $\lambda_0\le\zeta$, rather than $\lambda_0=0$.

Define the degree-$L$ polynomial
\begin{equation}
Q_L(x) \coloneqq \frac12 T_L\!\left( T_{1/L}(2)x \right).
\end{equation}
Equation~\eqref{eq:CSD-YLC-success} becomes $P_L(\lambda)=1-Q_L(\sqrt{1-\lambda})^2$. 
Moreover, $T_L(T_{1/L}(2))=2$, so $Q_L(1)=1$.

We next verify that $|Q_L(x)|\le1$ for $x\in[-1,1]$. 
If $|T_{1/L}(2)x|\le1$, then $|T_L(T_{1/L}(2)x)|\le1$, and therefore $|Q_L(x)|\le1/2$. 
For $1\le T_{1/L}(2)x\le T_{1/L}(2)$, the Chebyshev polynomial increases from $1$ to $2$, so $Q_L(x)\in[1/2,1]$. 
Because $L$ is odd, the same absolute bound holds for negative $x$.

Since $Q_L$ has degree $L$ and is bounded by one on $[-1,1]$, Markov's polynomial inequality (also called Markov brothers' inequality) gives
\begin{equation}
\max_{x\in[-1,1]} |Q_L'(x)| \le L^2.
\end{equation}
Also, we have
$1-\sqrt{1-\lambda} = \sfrac{\lambda} {(1+\sqrt{1-\lambda})} \le \lambda.$
Applying the mean-value theorem between $\sqrt{1-\lambda}$ and $1$ gives
\begin{equation}
Q_L(\sqrt{1-\lambda}) \ge Q_L(1) - L^2\left(1-\sqrt{1-\lambda}\right) \ge 1-L^2\lambda.
\end{equation}
If $L^2\lambda\le1/2$, then $Q_L(\sqrt{1-\lambda})\ge1/2$. 
Therefore
\begin{align}
P_L(\lambda) &= 1-Q_L(\sqrt{1-\lambda})^2 \\
&= \left[ 1-Q_L(\sqrt{1-\lambda}) \right] \left[ 1+Q_L(\sqrt{1-\lambda}) \right] \\
&\le 2\left[ 1-Q_L(\sqrt{1-\lambda}) \right] \\
&\le 2L^2\lambda. \label{eq:CSD-fixed-point-low-side}
\end{align}

For a ground-state anchor, $\lambda_0\le\zeta$ by~\eqref{eq:lambda0-upper}. 
Hence, whenever $L^2\zeta\le1/2$,
\begin{equation}
s_0^{\rm CSD} = P_L(\lambda_0) \le 2L^2\lambda_0 \le 2L^2\zeta.
\end{equation}
This proves~\eqref{eq:coherent-ground-witness}.
\end{proof}

\section{Candidate Distribution Conditioned on a Witness in CSD}
\label{app:coherent-witness-direction}

\begin{proof}[Proof of \Cref{lem:coherent-direction-preserved}]
Fix an anchor state with energy $\mathcal E_j$. 
We first derive the state of the witness branch produced before any amplitude amplification is applied.

Recall that at the beginning, the anchor register contains $|\mathcal E_j\rangle_A$, while the candidate register and witness qubit are initialized in $|0^n\rangle_B$ and $|0\rangle_q$. 
The initial state is therefore
$|\mathcal E_j\rangle_A|0^n\rangle_B|0\rangle_q,$
and we apply the unitary $V$ as defined in~\eqref{eq:A-search} to this state.

The first operation in $V$ is $U_\psi$, which prepares the candidate state in register $B$. 
Using the notation introduced in \Cref{sec:measured-search}, we write the input state as
$|\psi\rangle = \sum_k\sqrt{\pi_k}\,|\mathcal E_k\rangle,$
where $|\mathcal E_k\rangle$ denotes the normalized component of the input state $|\psi\rangle$ with energy $\mathcal E_k$. 
Thus, after applying $U_\psi$, we get
\begin{equation}
|\mathcal E_j\rangle_A|0^n\rangle_B|0\rangle_q \longmapsto \sum_k \sqrt{\pi_k}\, |\mathcal E_j\rangle_A |\mathcal E_k\rangle_B |0\rangle_q. \label{eq:CSD-after-Upsi}
\end{equation}

We now apply the comparator $C_\zeta$ to all the registers. 
For a single energy pair $(\mathcal E_j,\mathcal E_k)$, its action is given by the following (see~\eqref{eq:comparator-action}):
\begin{equation}
C_\zeta |0\rangle_q |\mathcal E_j\rangle_A |\mathcal E_k\rangle_B = A(z_{jk}) |0\rangle_q |\mathcal E_j\rangle_A |\mathcal E_k\rangle_B + B(z_{jk}) |1\rangle_q |\mathcal E_j\rangle_A |\mathcal E_k\rangle_B.
\end{equation}
Applying this relation term by term to~\eqref{eq:CSD-after-Upsi}, we obtain
\begin{equation}
V |\mathcal E_j\rangle_A|0^n\rangle_B|0\rangle_q = \sum_k \sqrt{\pi_k}\, A(z_{jk}) |\mathcal E_j\rangle_A |\mathcal E_k\rangle_B |0\rangle_q + \sum_k \sqrt{\pi_k}\, B(z_{jk}) |\mathcal E_j\rangle_A |\mathcal E_k\rangle_B |1\rangle_q. \label{eq:CSD-before-amplification}
\end{equation}
The second term in~\eqref{eq:CSD-before-amplification} is the witness branch, that is, $q=1$.

The states $|\mathcal E_k\rangle_B$ belonging to distinct energies are orthogonal. 
Therefore the contribution of candidate energy $\mathcal E_k$ to the squared norm of the witness branch is
$\pi_k|B(z_{jk})|^2 = \pi_k r_{jk},$
where we used $|B(z_{jk})|^2=r_{jk}$. 
Summing over all candidate energies gives the total witness probability
$\sum_k\pi_k r_{jk} = \lambda_j,$
in agreement with~\eqref{eq:lambda-j}. 
Hence, before amplitude amplification, conditioning on the witness ($q=1$) gives
\begin{equation}
\Pr[\mathcal E_k\mid j,\mathrm{witness}] = \frac{\pi_k r_{jk}}{\lambda_j} = q_j(k). \label{eq:CSD-single-witness-distribution}
\end{equation}

It remains to show that the FPAA operation changes the total probability of obtaining this witness branch, but not the relative amplitudes of the candidate-energy components within it.
We first consider $0<\lambda_j<1$, and the endpoint cases are addressed below.

Recall from~\eqref{eq:fpaa-witness-decomposition} that
$|\Phi_j^{(1)}\rangle$ and $|\Phi_j^{(0)}\rangle$ denote the
normalized witness and no-witness components of the prepared state,
respectively.
Using~\eqref{eq:CSD-before-amplification}, the witness component is
\begin{equation}
|\Phi_j^{(1)}\rangle = \frac{1}{\sqrt{\lambda_j}} \sum_k \sqrt{\pi_k}\, B(z_{jk}) |\mathcal E_j\rangle_A |\mathcal E_k\rangle_B |1\rangle_q. \label{eq:CSD-normalized-witness-direction}
\end{equation}
Similarly, $|\Phi_j^{(0)}\rangle$ is the normalized $q=0$ component
of~\eqref{eq:CSD-before-amplification}.
The state prepared by $V$ can therefore be written as
\begin{equation}
|\Phi_j\rangle = \sqrt{1-\lambda_j}\,|\Phi_j^{(0)}\rangle + \sqrt{\lambda_j}\,|\Phi_j^{(1)}\rangle. \label{eq:CSD-two-dimensional-decomposition}
\end{equation}

We now follow the action of the two phase operations used in the FPAA sequence. 
The witness phase operation $S_{\rm wit}(\beta)$ acts as the identity on the $q=0$ component and multiplies the entire $q=1$ component by $\e^{\iu\beta}$. 
Hence
\begin{equation}
|\Phi_j^{(0)}\rangle \longmapsto |\Phi_j^{(0)}\rangle, \qquad |\Phi_j^{(1)}\rangle \longmapsto \e^{\iu\beta}|\Phi_j^{(1)}\rangle.
\end{equation}
In particular, every candidate-energy component inside $|\Phi_j^{(1)}\rangle$ receives the same phase. 
Their relative amplitudes are therefore unchanged.

Next consider $S_{\rm src}(\alpha)$.
By~\eqref{eq:Ssrc}, within the fixed-anchor sector it applies the phase $\e^{-\iu\alpha}$ to $|\Phi_j\rangle$ and acts as the identity on its orthogonal complement.
Since $|\Phi_j\rangle$ lies in $\operatorname{span}\{|\Phi_j^{(0)}\rangle,|\Phi_j^{(1)}\rangle\}$, this operation preserves the same two-dimensional space.

Every step of the FPAA sequence therefore maps
$\operatorname{span}\{|\Phi_j^{(0)}\rangle,|\Phi_j^{(1)}\rangle\}$ to itself. 
Starting from~\eqref{eq:CSD-two-dimensional-decomposition}, the state after the complete sequence must consequently have the form
$a_j|\Phi_j^{(0)}\rangle+b_j|\Phi_j^{(1)}\rangle$
for some amplitudes $a_j$ and $b_j$. 
The FPAA sequence changes $b_j$, and hence changes the total witness probability, but the normalized witness direction $|\Phi_j^{(1)}\rangle$ itself is unchanged.

Conditioning on the final witness branch therefore leaves the candidate-energy probabilities exactly as in~\eqref{eq:CSD-single-witness-distribution}:
\begin{equation}
\Pr[\mathcal E_k\mid j,\mathrm{witness}] = \frac{\pi_k r_{jk}}{\lambda_j} = q_j(k).
\end{equation}
If $\lambda_j=0$, both phase operations preserve the no-witness prepared state up to a phase, so no witness can occur and there is no conditional distribution to consider.
If $\lambda_j=1$, both phase operations preserve the witness prepared state up to a phase, so its candidate-energy distribution is unchanged.
Thus FPAA changes the probability of finding a witness without changing the conditional distribution of the successful candidate energy whenever the final witness probability is nonzero.
\end{proof}

\section{CSD Potential Recursion for General Anchor States}
\label{app:coherent-potential-recursion}

We now justify the recursion~\eqref{eq:V-recursion-coherent} for an arbitrary anchor state in CSD. 
The corresponding argument for SD was proved in detail in \Cref{app:measured-potential-recursion}. 
Here we only verify the additional point required for CSD: the coherent witness-search circuit preserves the anchor-energy sectors before the final SWAP.

Let $\rho_t$ be the anchor state at the beginning of a CSD descent stage, and decompose it as
$\rho_t=\sum_{j,j'}\Pi_j\rho_t\Pi_{j'}.$
As in \Cref{app:measured-potential-recursion}, the diagonal blocks $\Pi_j\rho_t\Pi_j$ determine the initial energy populations, whereas the blocks with $j\neq j'$ describe coherences between distinct anchor energies.

Consider the coherent witness search before the witness qubit measurement and conditional anchor update via SWAP operation.
The operations $U_\psi$ and $U_\psi^\dagger$ act only on the candidate register, and $C_\zeta$ and $C_\zeta^\dagger$ are block diagonal in the anchor-energy basis by~\eqref{eq:comparator-action}.
Thus $V=C_\zeta(I_q\otimes I_A\otimes U_\psi)$ commutes with each anchor spectral projector $\Pi_j$, where identities on the other registers are implicit.
Since $P_{\rm init}$ also commutes with $\Pi_j$, the operator $S_{\rm src}(\alpha)$ in~\eqref{eq:Ssrc} satisfies
\begin{equation}
[S_{\rm src}(\alpha),\Pi_j] = \left[ V\bigl[I_{qAB}-(1-\e^{-\iu\alpha})P_{\rm init}\bigr] V^\dagger,\Pi_j \right] =0.
\end{equation}
The operator $S_{\rm wit}(\beta)$ acts only on $q$, so it also commutes with $\Pi_j$.
Therefore the entire coherent search preserves each
anchor-energy eigenspace, and an input block
$\Pi_j\rho_t\Pi_{j'}$ remains between the same two
anchor-energy sectors $(j,j')$ throughout the search.

After the search, the witness qubit separates the state into the no-witness and witness branches. 
On the no-witness branch, no SWAP is performed. 
Thus an off-diagonal block with $j\neq j'$ remains supported between the same two anchor-energy sectors and contributes zero to every output energy population, exactly as in~\eqref{eq:no-witness-coherence-zero}.

On the witness branch, we apply $\operatorname{SWAP}_{AB}$, which moves the old anchor into register $B$, and then discard $B$.
For an input block $\Pi_j\rho_t\Pi_{j'}$ with $j\neq j'$, tracing out the old anchor gives zero by the same calculation as~\eqref{eq:off-diagonal-trace-zero}. 
Finally, cross terms between the witness and no-witness branches vanish when the witness qubit is traced out because the states $|0\rangle_q$ and $|1\rangle_q$ are orthogonal.

It follows that, exactly as in \Cref{app:measured-potential-recursion}, the energy populations after one CSD descent stage depend only on the diagonal input blocks $\Pi_j\rho_t\Pi_j.$
Let $p_j\coloneqq\Tr(\Pi_j\rho_t),$ and, for $p_j>0$, define $\widehat\rho_j\coloneqq\Pi_j\rho_t\Pi_j/p_j.$
Each $\widehat\rho_j$ is supported entirely in the eigenspace of energy $\mathcal E_j$.

Let $\mathcal S_{\rm CSD}$ denote one complete CSD descent stage. 
Repeating the decomposition derived in~\eqref{eq:Vt-diagonal-decomposition} gives
\begin{equation}
V_{t+1} = \sum_j p_j \sum_k W_k \Tr\!\left( \Pi_k \mathcal S_{\rm CSD}(\widehat\rho_j) \right). \label{eq:CSD-Vt-diagonal-decomposition}
\end{equation}

For $j=0$, the state $\widehat\rho_0$ lies entirely in the ground-energy eigenspace, so~\eqref{eq:coherent-ground-error} gives
\begin{equation}
\sum_k W_k \Tr\!\left( \Pi_k\mathcal S_{\rm CSD}(\widehat\rho_0) \right) \le 2L^2\zeta.
\end{equation}
For every $j>0$, the fixed-energy contraction~\eqref{eq:coherent-excited-contraction} gives
\begin{equation}
\sum_k W_k \Tr\!\left( \Pi_k\mathcal S_{\rm CSD}(\widehat\rho_j) \right) \le \frac{43}{64}W_j + \frac{\zeta}{\eta(1-\zeta)}.
\end{equation}

Substituting these bounds into~\eqref{eq:CSD-Vt-diagonal-decomposition}, we obtain
\begin{align}
V_{t+1} &\le \Tr(\Pi_0\rho_t)\,2L^2\zeta \nonumber \\
&\quad+ \sum_{j>0} \Tr(\Pi_j\rho_t) \left( \frac{43}{64}W_j + \frac{\zeta}{\eta(1-\zeta)} \right) \nonumber \\
&\le \frac{43}{64} \sum_{j>0} W_j\Tr(\Pi_j\rho_t) + 2L^2\zeta + \frac{\zeta}{\eta(1-\zeta)}.
\end{align}
Since $W_0=0$, the definition~\eqref{eq:Vt} gives
$\sum_{j>0}W_j\Tr(\Pi_j\rho_t)=V_t.$
Therefore
\begin{equation}
V_{t+1} \le \frac{43}{64}V_t + 2L^2\zeta + \frac{\zeta}{\eta(1-\zeta)},
\end{equation}
which is exactly~\eqref{eq:V-recursion-coherent}.

\bibliographystyle{alpha}
\bibliography{ref}

\end{document}